\documentclass[11pt]{article}
\usepackage[latin2]{inputenc}
\usepackage[english]{babel}
\usepackage{amsmath}
 \usepackage{amsthm}
\usepackage{amssymb}
\usepackage{graphicx}
\usepackage{tabularx}
\usepackage{amsfonts}
\usepackage{verbatim}
\usepackage[margin=1in]{geometry}
\newtheorem{theorem}{Theorem}[section]
\newtheorem{lemma}[theorem]{Lemma}
\newtheorem{claim}[theorem]{Claim}
\newtheorem{corollary}[theorem]{Corollary}
\newtheorem{proposition}[theorem]{Proposition}
\newtheorem{observation}[theorem]{Observation}

\newtheorem{definition}[theorem]{Definition}
\newtheorem{example}[theorem]{Example}

\newcommand{\ccsp}{\textup{CCSP}}
\newcommand{\ocsp}{\textup{OCSP}}
\newcommand{\csp}{\textup{CSP}}
\newcommand{\mvm}{\textup{MVM}}
\newcommand{\nand}{\textup{NAND}}
\newcommand{\imp}{\textup{IMP}}
\newcommand{\ccspg}{$\ccsp(\Gamma)$}
\newcommand{\ocspg}{$\ocsp(\Gamma)$}
\newcommand{\pr}{\textup{ret}}
\newcommand{\dom}{\textup{dom}}

\newcommand{\B}{{\cal B}}
\newcommand{\C}{{\mathcal C}}
\newcommand{\Z}{{\cal Z}}
\newcommand{\ba}{{\mathbf a}}
\newcommand{\bb}{{\mathbf b}}
\newcommand{\bs}{{\mathbf s}}
\newcommand{\bt}{{\mathbf t}}

\def\ang#1{{\langle #1\rangle}}
\let\sse=\subseteq
\def\nat{{\mathbb N}}
\newcommand\supph{{\sf supp}}

\newcommand{\Kr}{K^*}

\begin{document}

\title{Constraint satisfaction parameterized by solution size\thanks{An extended abstract of the paper appeared in the proceedings of ICALP 2011 \cite{bulatov-marx-ccsp-icalp2011}.}} 
\author{Andrei A.\ Bulatov\thanks{School of Computing Science,
 Simon Fraser University, Burnaby, Canada,
 abulatov@cs.sfu.ca.   Research supported by NSERC Discovery grant}
\and
D\'aniel Marx\thanks{Computer and Automation Research Institute,
Hungarian Academy of Sciences (MTA SZTAKI),
Budapest, Hungary, dmarx@cs.bme.hu.  Research partially supported by the European Research Council (ERC) grant    ``PARAMTIGHT: Parameterized complexity and the search for tight   complexity results,'' reference 280152 and OTKA grant NK105645.}}
\maketitle  
  
\begin{abstract}
  In the constraint satisfaction problem (CSP) corresponding to a
  constraint language (i.e., a set of relations) $\Gamma$, the goal is
  to find an assignment of values to variables so that a given set of
  constraints specified by relations from $\Gamma$ is satisfied. The
  complexity of this problem has received substantial amount of
  attention in the past decade. In this paper, we study the
  fixed-parameter tractability of constraint satisfaction problems
  parameterized by the size of the solution in the following sense:
  one of the possible values, say 0, is ``free,'' and the number of
  variables allowed to take other, ``expensive,'' values is
  restricted. A {\em size constraint} requires that exactly $k$ variables
  take nonzero values.  We also study a more refined version of this
  restriction: a {\em global cardinality constraint} prescribes how many
  variables have to be assigned each particular value. We study the
  parameterized complexity of these types of CSPs where the parameter
  is the required number $k$ of nonzero variables. As special cases,
  we can obtain natural and well-studied parameterized problems such as
  \textsc{Independent set}, \textsc{Vertex Cover}, \textsc{$d$-Hitting
    Set}, \textsc{Biclique}, etc.

  In the case of constraint languages closed under substitution of
  constants, we give a complete characterization of the
  fixed-parameter tractable cases of CSPs with size constraints, and
  we show that all the remaining problems are W[1]-hard. For CSPs with
  cardinality constraints, we obtain a similar classification, but for
  some of the problems we are only able to show that they are
  \textsc{Biclique}-hard. The exact parameterized complexity of the
  \textsc{Biclique} problem is a notorious open problem, although it
  is believed to be W[1]-hard.
\end{abstract}

\thispagestyle{plain}
 \clearpage
 \tableofcontents
 \clearpage

\section{Introduction}
In a constraint satisfaction problem (CSP) we are given a set of
variables, and the goal is to find an assignment of the variables
subject to specified constraints. A constraint is usually
expressed as a requirement that combinations of values of a certain
(usually small) set of variables belong to a certain relation. In the
theoretical study of CSPs, one of the key research
directions has been the complexity of the CSP when there are
restrictions on the type of allowed relations
\cite{Jeavons97:closure,Bulatov05:classifying,DBLP:journals/tocl/Bulatov11,%
Bulatov06:3-element,Barto08:graphs}. 
This research direction has been started by the seminal  Schaefer's Dichotomy
Theorem~\cite{MR80d:68058}, which showed that every
Boolean CSP (i.e., CSP with 0-1 variables) restricted in this way is
either solvable in polynomial time or is NP-complete. An outstanding
open question is the so called \emph{Dichotomy conjecture} of Feder
and Vardi \cite{Feder98:monotone,Kun-Szegedy-STOC2009}, which suggests that the dichotomy
remains true for CSPs over any fixed finite domain. The significance of a
dichotomy result is that it is very likely to provide a comprehensive
understanding of the algorithmic nature of the problem. Indeed, in
order to obtain the tractability part of such a conjecture, one needs
to identify and understand all the algorithmic ideas relevant for the problem. 

Parameterized complexity \cite{MR2001b:68042,grohe-flum-param}
investigates problem complexity in finer 
details than classical complexity. Instead of expressing the running
time of an algorithm as a function of the input size $n$ only, the
running time is expressed as a function of $n$ and a well-defined
parameter $k$ of the input instance (such as the size of the solution
$k$ we are looking for). For many problems and parameters, there is a
polynomial-time algorithm for every fixed value of $k$, i.e., the
problem can be solved in time $n^{f(k)}$. In this case, it makes sense
to ask whether there is an algorithm with running time $f(k)\cdot
n^{O(1)}$, that is, whether the combinatorial explosion can be
limited to the parameter $k$ only. Problems having algorithms with
running time of this form are called {\em fixed-parameter tractable
  (FPT)}; it turns out that many well-known NP-hard problems, such as
\textsc{$k$-Vertex Cover}, \textsc{$k$-Path}, and \textsc{$k$-Disjoint
  Triangles} are FPT. On the other hand, the theory of W[1]-hardness
suggests that certain problems (e.g., \textsc{$k$-Clique},
\textsc{$k$-Dominating Set}) are unlikely to be FPT.

Investigating the fixed-parameter tractability is relevant only for
those problems that can be solved in polynomial time for every fixed
value of $k$. For example, the canonical complete problems of the
W-hierarchy are (formula or circuit) satisfiability problems where the solution
is required to contain exactly $k$ ones; clearly, such a problem is
solvable in time $n^{O(k)}$, which is polynomial for every fixed value
of $k$. This leads us naturally to the study of Boolean CSP, 
where the goal is to find a solution with exactly $k$ ones.

The first attempt to study the parameterized complexity of Boolean CSP
was made in \cite{Marx05:parametrized}. If we consider 0 as a
``cheap'' value available in abundance, while 1 is ``costly'' and of
limited supply then the natural parameter is the number of 1's in a
solution. Boolean CSP asking for a solution that assigns exactly $k$
ones is known as the \textsc{$k$-Ones} problem (see, e.g.\
\cite{DBLP:journals/tocl/CreignouSS10,MR2002k:68058}).  Clearly, the
problem is polynomial-time solvable for every fixed $k$, as we can
search through all assignments with exactly $k$ ones, but it is not at
all obvious whether it is FPT.  For example, it is possible to express
\textsc{$k$-Vertex Cover} (which is FPT) and \textsc{$k$-Independent
  Set} (which is W[1]-hard) as a Boolean CSP.  Therefore,
characterizing the parameterized complexity of \textsc{$k$-Ones}
requires understanding a class of problems that includes, among many
other things, the most basic parameterized graph problems.  
It turned out that the parameterized
complexity of the \textsc{$k$-Ones} problem depends on a combinatorial property
called {\em weak separability} in \cite{Marx05:parametrized}. Assuming that the constraints are restricted to a
finite set $\Gamma$ of Boolean relations, if  every relation in $\Gamma$
is weakly separable, then the problem is FPT; if $\Gamma$ contains
even one relation violating weak separability, then the problem is
W[1]-hard. Another
natural problem of this flavor is deciding whether there exists a
solution with at most $k$ ones; however, this problem is always FPT,
as it follows from Lemma~\ref{lem:minextension} below even in a more
general case.

There have been further parameterized complexity studies of Boolean
CSP
\cite{DBLP:conf/icalp/KratschW10,DBLP:journals/disopt/Szeider11,krokhin-marx-icalp2008,kratsch-mfcs2010-maxexact},
but CSP's with larger domains (i.e., where the variables are not
Boolean) were not studied. In most cases, we expect that results for
larger domains are much more complex than for the Boolean case, and
usually require significant new ideas (compare e.g., Schaefer's
Theorem \cite{MR80d:68058} with the 3-element version
\cite{Bulatov06:3-element}). The goal of the present paper is to
generalize the results of \cite{Marx05:parametrized} to non-Boolean
domains. First, we have to define what the proper generalization of
\textsc{$k$-Ones} is if the variables are not Boolean. One natural
generalization assumes that there is a distinguished ``cheap'' value 0
and requires that in a solution there are exactly $k$ nonzero
variables. We will call this version of the CSP a \emph{constraint
  satisfaction problem with size constraints} and denote it by
\ocsp. In another generalization of \textsc{$k$-Ones}, we have a
cardinality requirement not only on the total number of nonzero
variables, but a separate cardinality requirement for each nonzero
value restricting the number of times this value is used: A mapping
$\pi:D\setminus\{0\}\to\nat$ is given in the input, and it is required that for
each nonzero value $d$, exactly $\pi(d)$ variables are assigned value
$d$.
In the CSP and AI literature, requirements of this form are called
{\em global cardinality constraints}
\cite{Bessiere04:global,Bourdais03:hibiscus,Gomes04:cardinality,bulatov-marx-ccsp-cacm,bulatov-marx-ccsp-lmcs,DBLP:conf/ijcai/GaspersS11,DBLP:journals/constraints/SamerS11} and
have been intensively studied. We will call this problem the
\emph{constraint satisfaction problem with 
cardinality constraints} and denote it by \ccsp. In both versions, the
parameter is the number of nonzero values required, that is, $k$ for
\ocsp, and $\sum_{d\in D\setminus\{0\}}\pi(d)$ for \ccsp.
 We
investigate both 
versions; as we shall see, there are interesting and unexpected
differences between the complexity of the two variants.

A natural minor generalization of CSPs is to allow the use of
constants in the input, that is, certain variables in the input can be
fixed to constant values, or equivalently, the constant unary relation
$\{d\}$ is allowed for every element $d$ of the domain.  Yet another equivalent
way of formulating this generalization is requiring that $\Gamma$ is
closed under substitution of constants.  It is known that to classify
the complexity of the general CSP with respect to polynomial-time
solvability, it suffices to classify constraint languages closed under
substitution of constants \cite{Bulatov05:classifying}.  This
motivates our assumption that the constant relations are available for
CSPs with cardinality and size constraints.
While there is no result similar to that from
\cite{Bulatov05:classifying} for the versions of CSPs we study here
(and thus this assumption somewhat diminishes the generality of our results),
this setting is still quite general and at the same time more robust.
Lots of technicalities can be avoided with this formulation. For
example, being able to substitute constants ensures that the decision and search
problems are equivalent: by repeatedly substituting constants and
solving the resulting decision problems, we can actually find a solution.

Is weak separability the right tractability criterion in the
non-Boolean case? It is not difficult to observe that the algorithm of
\cite{Marx05:parametrized} using weak separability generalizes for non-Boolean problems. (In
fact, we give a much simpler algorithm in this paper.)  
However, it is not true that only weakly separable relations are tractable.  It
turns out that there are certain degeneracies and symmetries that allow
us to solve the problem even for some relations that are not
weakly separable.  To understand these degenerate situations, we
introduce the notion of multivalued morphisms, which is a
generalization of homomorphisms.  While the use of algebraic
techniques and homomorphisms is a standard approach for understanding
the complexity of CSPs \cite{Bulatov05:classifying}, this notion is
new and seems to be essential 
for understanding the problem in our setting.

\textbf{Results.}  
For CSP with size constraints, we prove a dichotomy result:

\begin{theorem}\label{th:mainsingleintro}
  For every finite $\Gamma$ closed under substitution of constants,
  \ocspg\ is either \textup{FPT} or 
  \textup{W[1]}-hard.
\end{theorem}

The precise tractability criterion (which is quite technical) is
stated in Section~\ref{sec:size}. The algorithmic part of
Theorem~\ref{th:mainsingleintro} consists of preprocessing to
eliminate degeneracies and trivial cases, followed by the use of weak
separability. In the hardness part, we take a relation $R$ having a
counterexample to weak separability, and try to use it to show that a
known W[1]-hard problem can be simulated with this relation. In the
Boolean case \cite{Marx05:parametrized}, this can be done in a fairly
simple way: by identifying coordinates and substituting 0's, we can
assume that the relation $R$ is binary, and we need to prove hardness
only for two concrete binary relations.  For larger domains, however,
this approach does not work. We cannot reduce the
counterexample to a binary relation by identifying coordinates, thus it seems that a
complex case analysis would be needed. Fortunately, our hardness proof
is more uniform than that.  We introduce gadgets that control the
values that can appear on the variables. There are certain degenerate
cases when these gadgets do not work. However, these degenerate cases
can be conveniently described using multivalued morphisms, and these
cases turn out to be exactly the cases that we can use in the
algorithmic part of the proof.

In the case of cardinality constraints, we face an interesting
problem.  Consider the binary relation $R$ containing only tuples
$(0,0)$, $(1,0)$, and $(0,2)$. Given a CSP instance with constraints
of this form, finding a solution where the number of 1's is exactly
$k$ and the number of 2's is exactly $k$ is essentially equivalent to
finding an independent set of a bipartite graph with $k$ vertices in
both classes, or equivalently, a complete bipartite graph (biclique)
with $k+k$ vertices. However, the parameterized complexity of the
\textsc{$k$-Biclique} problem is a longstanding open question (it is
conjectured to be W[1]-hard \cite{Grohe07:other-side}). Thus at this point, it is
not possible to give a dichotomy result similar to
Theorem~\ref{th:mainsingleintro} in the case of cardinality constraints,
unless assuming that \textsc{Biclique} is hard.

\begin{theorem}\label{th:mainmultiintro}
  For every finite $\Gamma$ closed under substitution of constants,
  \ccspg\ is either \textup{FPT} or 
\textsc{Biclique}-hard.
\end{theorem}

Note that in many cases, we actually prove \textup{W[1]}-hardness
in Theorem~\ref{th:mainmultiintro} and we state
\textsc{Biclique}-hardness only in very specific situations. Similarly
to Theorem~\ref{th:mainsingleintro}, the algorithmic part of
Theorem~\ref{th:mainmultiintro} exploits weak separability after a
preprocessing phase, but in this case there is a final phase that extends the
solution using certain degeneracies. The hardness proof is similar 
to Theorem~\ref{th:mainsingleintro}, however, it requires at times different
combinatorial arguments, slightly different types of degeneracies, and
reductions from different hard parameterized problems.

The parameterized complexity of \textsc{Biclique} has been a
significant open problem for a number of years and in fact it is the
most natural and most easily stated problem on graphs whose
parameterized complexity status is open. The parameterized complexity
of similar basic problems (such as \textsc{$k$-Clique},
\textsc{$k$-Path}, \textsc{$k$-Disjoint Triangles}, etc.) is well
understood, but it seems that very different techniques are required
to tackle \textsc{Biclique}. Our results give a further incentive for
the study of \textsc{Biclique}: its hardness would explain hardness in
all the cases that are not classified as FPT or W[1]-hard in this paper.
In other words, in a very general
problem family, the hardness of \textsc{Biclique} is the only question
that is not yet understood, highlighting the importance of this
problem. We believe this observation to be a significant byproduct of
our study.

\textbf{Organization.}  The paper is organized as follows. In
Section~\ref{sec:preliminaries} we introduce constraint satisfaction
problems with cardinality constraints and  argue that only 0-valid
constraint languages need to be studied, so from that point on,
we consider only cc0-languages.
Weak separability and various
degeneracies of constraint languages are introduced and studied in
Section~\ref{sec:properties}.
The main classification results are proved in Section~\ref{sec:size}
(for size constraints) and Section~\ref{sec:cardinality} (for
cardinality constraints). In both sections, we first present the
algorithmic side of the classification, then prove hardness in various
cases. After going through certain fairly degenerate situations, the
most generic (and from the point of view of proof techniques, most
interesting) proofs are found at the end of these sections.
Section~\ref{sec:examples} contains examples that demonstrate some of
the concepts introduced in the paper. We will continuously refer to
these examples throughout the paper.

\section{Preliminaries}\label{sec:preliminaries}

%
Let $D$ be a set. We assume that $D$ contains a distinguished element
$0$. Let $D^n$ denote the set of all $n$-tuples of elements from
$D$. An $n$-ary \emph{relation} on $D$ is a subset of $D^n$, and a 
\emph{constraint language} $\Gamma$ is a set of relations on $D$. In this paper,
constraint languages are assumed to be finite. We denote by
$\dom(\Gamma)$ the set of all values that appear in tuples of the
relations in $\Gamma$. In proofs throughout the paper, we use languages 
derived from the original language $\Gamma$, so this set does not have to be
equal to $D$.
Given a set $D$ and a constraint language $\Gamma$, an instance of the
\emph{constraint satisfaction problem} (CSP) is a pair $I=(V,\C)$,
where $V$ is a set of \emph{variables}, and $\C$ is a 
set of \emph{constraints}. Each constraint is a pair $\ang{\bs,R}$,
where $R$ is a (say, $n$-ary) relation from $\Gamma$, and $\bs$ is an
$n$-tuple of variables (repetitions of variables are allowed).  
A \emph{satisfying assignment} of $I$ is a mapping $\tau:V\to D$ such that for
every $\ang{\bs,R}\in\C$ with $\bs=(s_1\ldots,s_n)$ the image
$\tau(\bs)=(\tau(s_1),\ldots,\tau(s_n))$ belongs to $R$. The question in
the CSP is whether there is a satisfying assignment for
a given instance. The CSP over constraint language $\Gamma$ is denoted
by $\csp(\Gamma)$.  

The \emph{size} of an assignment is the number of variables receiving
nonzero values. A \emph{size constraint} is a requirement to a satisfying
assignment  to have a prescribed size. The variant of
$\csp(\Gamma)$ that allows size constraints will be denoted by
$\ocsp(\Gamma)$.
A \emph{cardinality constraint} for a CSP instance $I$ is a mapping
$\pi:D\setminus\{0\}\to\nat$ with $\sum_{a\in D}\pi(a)\le|V|$. A
satisfying assignment $\tau$ of $I$ satisfies the cardinality
constraint $\pi$ if the number of variables mapped to each $a\in D$
equals $\pi(a)$. We denote by $\ccsp(\Gamma)$ the variant of
$\csp(\Gamma)$ where the input also contains a cardinality constraint
$\pi$ and the size constraint $k=\sum_{a\in D\setminus \{0\}}\pi(a)$
(this constraint is used a parameter); the question is, given an
instance $I$, an integer $k$, and a cardinality constraint $\pi$,
whether there is a satisfying assignment of $I$ of size $k$ and
satisfying $\pi$.  (Obviously, $\sum_{a\in D\setminus \{0\}}\pi(a)=k$
must hold, thus specifying $k$ in the input is in a sense
redundant. However, for the uniform treatment of both problems, it
will be convenient to assume that $k$ appears in the input in both
cases.)  A \emph{solution} of an \ocsp\ (resp., \ccsp) instance is a
satisfying assignment that also satisfies the size (resp., cardinality)
constraints. So, instances of \ocsp\ (resp., \ccsp) are triples
$(V,\C,k)$ (resp., quadruples $(V,\C,k,\pi)$). For both \ocspg\ and
\ccspg, we are interested in fixed-parameter tractability with respect
to the parameter $k$, i.e., the goal is an algorithm with running time
$f(k)\cdot n^{O(1)}$ for some computable function $f$. Note that we
are making a distinction between two terms: in ``satisfying
assignment,'' we do not require that the cardinality constraint is
satisfied, while a ``solution'' always means that the cardinality
constraint is satisfied.

If $\Gamma$ is a constraint language on the 2-element set $\{0,1\}$,
then to specify a global cardinality constraint it suffices to specify
the number of ones we want to have in a solution. This problem is also
known as the {\sc $k$-Ones}$(\Gamma)$ problem
\cite{DBLP:journals/tocl/CreignouSS10,Marx05:parametrized}. Note that this
problem can be viewed as both $\ccsp(\Gamma)$ and $\ocsp(\Gamma)$. 
In the general case, $\ocsp(\Gamma)$ polynomial time reduces to 
$\ccsp(\Gamma)$.

\begin{lemma}\label{lem:ocsp-to-ccsp}
For any constraint language $\Gamma$ over a set $D$, there is a 
polynomial-time reduction from $\ocsp(\Gamma)$ to $\ccsp(\Gamma)$.
\end{lemma}

Indeed, for any instance $(V,\C,k)$ of $\ocsp(\Gamma)$ it suffices to solve 
$O(k^{d-1})$ instances $(V,\C,k,\pi)$ of $\ccsp(\Gamma)$ such that
$\sum_{a\in D\setminus \{0\}}\pi(a)=k$.

Examples~\ref{exa:graphs}--\ref{exa:k-partite-complete-subgraph} demonstrate
some concrete problems that can be expressed in this framework.

\subsection{Closures and 0-validity}\label{sec:closures-0-validity}
A constraint language $\Gamma$ is called \emph{constant closed} (cc-, for
short) if along with every (say, $n$-ary) relation $R\in\Gamma$, any $i$, $1\le
i\le n$, and any $d\in D$ the relation obtained by \emph{substitution
  of constants} $ R^{|i;d}=\{(a_1,\ldots,a_{i-1},a_{i+1},\ldots,
a_n)\mid (a_1,\ldots,a_{i-1},d,a_{i+1},\ldots, a_n)\in R\}, $ also
belongs to $\Gamma$.  Substitution of constants $d_1,\ldots,d_q$ for
coordinate positions $i_1,\ldots,i_q$ is defined in a similar way; the
resulting relation is denoted by
$R^{|i_1,\ldots,i_q;d_1,\ldots,d_q}$. We call the smallest cc-language
containing a constraint language $\Gamma$ the \emph{cc-closure} of
$\Gamma$. It is easy to see that the cc-closure of $\Gamma$ is the set
of relations obtained from relations of $\Gamma$ by substituting
constants.  Let $f$ be a satisfying assignment for an instance
$I=(V,\C,k)$ of \ocspg\ and $S=\{v\mid f(v)\ne0\}$.  We say that the
instance $I'=(V',\C',k')$ is \emph{obtained by substituting the
  nonzero values of $f$ as constants} if $I'$ is constructed as
follows: $V'=V\setminus S$, and for each constraint $\ang{\bs,R}\in\C$
such that $i_1,\ldots i_r$ are the indices of the variables from $\bs$
contained in $S$ and $\{j_1,\ldots,j_q\}$ is
the set of indices of variables from $\bs$ not in $S$, we include in $\C'$ the
constraint $\ang{(v_{j_1},\ldots,v_{j_q}),R^{|i_1,\ldots,
    i_r;f(v_{i_1}),\ldots f(v_{i_r})}}$. The size constraint $k'$ is
set to $k$ minus the size of $f$ (we always make sure that $f$ 
is such that $k'\ge0$).  
Observe that assigning 0 to every
variable of $I'$ is a satisfying assignment.  This operation is
defined similarly for a \ccspg\ instance $I=(V,\C,k,\pi)$, but in this
case the new cardinality constraint $\pi'$ is given by
$\pi'(d)=\pi(d)-|\{v\in V\mid f(v)=d\}|$ (we always assume that $\pi$
and $f$ are such that $\pi'$ is nonnegative).
    
In general, \ocspg\ or \ccspg\ can become harder if we replace
$\Gamma$ with its cc-closure, see Examples~\ref{exa:cc-hard} and
\ref{exa:cc-hard2}. Thus a classification for cc-languages does not
immediately imply a classification for all languages.

A relation is said to be \emph{0-valid} if the all-zero tuple belongs
to the relation.  A constraint language $\Gamma$ is a
\emph{cc0-language} if every $R\in \Gamma$ is 0-valid, and every
0-valid relation $R'$ obtained from $R$ by substitution of constants
belongs to $\Gamma$ (see Example~\ref{exa:cc0-languages}). 
An instance $I$ is said to be \emph{0-valid} if the all-0 assignment
is a satisfying one. In particular every instance of $\ocsp(\Gamma)$ 
or $\ccsp(\Gamma)$ for a cc0-language $\Gamma$ is 0-valid.
The
following observation is clear (but note that $\Gamma_0$ is not
necessarily a cc-language, as substitution into a 0-valid relation does not necessary result in a 0-valid relation):

\begin{proposition}\label{prop:cc0}
  If $\Gamma$ is a cc-language and $\Gamma_0$ is the set of 0-valid
  relations in $\Gamma$, then $\Gamma_0$ is a cc0-language.
  \end{proposition}
  
  We say that tuple $\bt_1=(a_1,\dots,a_r)$ is an {\em extension} of
  tuple $\bt_2=(b_1,\dots,b_r)$ if they are of the same length and for
  every $1 \le i \le r$, $a_i=b_i$ if $b_i\neq 0$. Tuple $\bt_2$ is
  then called a \emph{subset} of $\bt_1$. A {\em minimal satisfying
    extension} of an assignment $f$ is an extension $f'$ of $f$ (where
  $f,f'$ are viewed as tuples) such that $f'$ is satisfying, and $f$
  has no satisfying extension $f''\neq f'$ such that $f'$ is an
  extension of $f''$. We show that the minimal satisfying extensions
  of size at most $k$ can be enumerated by a simple branching
  algorithm, implying a bound on the number of such
  extensions.

\begin{lemma}\label{lem:minextension}
Let $\Gamma$ be a finite constraint language over $D$. There are functions 
$d'_{\Gamma}(k)$ and $e'_{\Gamma}(k)$ such that for every 
instance of $\csp(\Gamma)$ with $n$ variables, every assignment $f$ has at most
$d'_{\Gamma}(k)$ minimal satisfying extensions of size at most $k$
and all these minimal extensions can be enumerated in time $e'_{\Gamma}(k)n^{O(1)}$.
\end{lemma}
\begin{proof}
  Let $k_0\le k$ be the size of $f$.  The minimal satisfying
  extensions of $f$ can be enumerated by a bounded-depth search tree
  algorithm. If $f$ is a satisfying assignment, then $f$ itself is the
  unique minimal satisfying extension of $f$. Suppose therefore that a
  constraint is not satisfied by $f$. If every variable of the
  constraint has nonzero value, then the assignment has no satisfying
  extension. Otherwise, we try to assign a nonzero value to one of the
  0-valued variables of the constraint, thus we branch into at most
  $(|D|-1)r_\text{max}$ directions, where $r_\text{max}$ is the
  maximum arity of the relations in $\Gamma$. If the assignment
  obtained by this modification is still not satisfying, then the
  process is repeated with some unsatisfied constraint. If the
  assignment is still not satisfying after making $k-k_0$ variables
  nonzero, then this branch of the search is terminated. If we obtain
  a satisfying assignment after assigning a nonzero value to a set $S$
  of at most $k'\le k-k_0$ variables, then we check whether this
  extension is minimal by trying each of the $2^{k'}$ subsets of the
  changed variables.  This last check ensures that every extension
  produced by the algorithm is indeed minimal. To see that every
  minimal satisfying extension $f'$ of size at most $k$ is eventually
  enumerated, observe that at every branching we can make a
  substitution that is compatible with $f'$, i.e., the current
  assignment is a subset of $f'$. Thus after at most $k$ steps, the
  algorithm finds a satisfying assignment that is a subset of $f'$. As
  $f'$ is a minimal satisfying extension, this subset cannot be proper
  and has to be $f'$ itself.

The height of the search tree is at most $k$ and the branching factor
is at most $(|D|-1)r_\text{max}$, thus at most $((|D|-1)r_\text{max})^k$
assignments are enumerated. The running time is polynomial at each node
and the final check at each leaf takes time $2^{k}n^{O(1)}$. Thus the
total running time is $e'_\Gamma(k)n^{O(1)}$ for a suitable $e'_\Gamma$.
\end{proof}

A consequence of Lemma~\ref{lem:minextension} is that, as in
\cite{Marx05:parametrized}, \ccspg\ and $\ocsp(\Gamma)$ can be 
reduced to a set of 0-valid instances. 
Let $k$ be the parameter of the $\ocsp(\Gamma)$ or \ccspg\ instance
(i.e., the required size of the solution).
We enumerate all the minimal 
satisfying extensions of size at most $k$ of the all zero assignment and obtain the
0-valid instances by substituting the nonzero values as constants.

\begin{corollary}\label{cor:0valid}
  Let $\Gamma$ be a cc-language and let $\Gamma_0\sse\Gamma$ be the
  set of all 0-valid relations. If $\ccsp(\Gamma_0)$ is
  \textup{FPT}/\textup{W[1]}-hard/\textsc{Biclique}-hard, then \ccspg\
  is as well. The same holds for $\ocsp(\Gamma)$ and
  $\ocsp(\Gamma_0)$.
\end{corollary}

By Prop.~\ref{prop:cc0}, if $\Gamma$ is a cc-language, then $\Gamma_0$
is a cc0-language. Thus Corollary~\ref{cor:0valid} reduces the
classification of the complexity of cc-languages into the
classification of cc0-languages. Thus in the rest of the paper, we have
to deal with cc0-languages only.

\subsection{Reductions}
\label{sec:reductions}

To obtain the W[1]-hardness results, we use the standard notion of
{\em parameterized reduction} \cite{MR2001b:68042,grohe-flum-param}.
 
\begin{definition}
A {\em parameterized reduction} from parameterized problem $A$ to
parameterized problem $B$ is a mapping $R$ from the instances of $A$ to
the instances of $B$ with the following properties:
\begin{enumerate}
\item $I$ is a yes-instance of $A$ if and only if $R(I)$ is a
  yes-instance of $B$,
\item $R(I)$ is computable in time $f(k)\cdot n^{O(1)}$, where
  $n$ is the size of $I$, $k$ is the parameter of $I$, and $f(k)$ is a
  computable function depending only on $k$,
\item the parameter of $R(I)$ is at most $g(k)$, where $k$ is the
  parameter of $I$ and $g(k)$ is a computable function depending only on $k$.
\end{enumerate}
\end{definition}

It is easy to see that if there is a parameterized reduction from $A$
to $B$ and $B$ is FPT, then $A$ is FPT as well.

In the hardness proofs, we will be reducing from the following
parameterized problems ($t$ is the parameter of the instance). 
\begin{itemize}
\item 
\textsc{Independent Set}.
 Given a graph $G$ with
 vertices $v_{i}$ ( $1\le j \le n$), find an
 independent set of size $t$.
\item 
\textsc{Multicolored Independent Set:} Given a graph $G$ with
  vertices $v_{i,j}$ ($1\le i \le t$, $1\le j \le n$), find an
  independent set of size $t$ of the form
  $\{v_{1,y_1},\dots,v_{t,y_t}\}$.
\item 
\textsc{Implications}:
 Given a directed graph $G$ and an integer
 $t$, find a set $C$ of vertices with exactly $t$ vertices such that
 there is no directed edge $\overrightarrow{uv}$ with $u\in C$ and
 $v\not\in C$. 
\item 
\textsc{Multicolored Implications:} Given a directed graph $G$
  with vertices $v_{i,j}$ ($1\le i \le t$, $1\le j \le n$), find a set
  $C=\{v_{1,y_1},\dots,v_{t,y_t}\}$ with exactly $t$
  vertices such that there is no directed edge $\overrightarrow{uv}$
  with $u\in C$ and $v\not\in C$.
\item 
\textsc{Biclique}. 
 Given a bipartite graph $G(A,B)$, find two
 sets $A'\subseteq A$ and $B'\subseteq B$, each  of size exactly $t$, such
 that every vertex of $A'$ is adjacent to every vertex of $B'$.
\end{itemize}

\textsc{Independent Set} is one of the basic W[1]-hard problems and it
is not difficult to show that \textsc{Multicolored Independent Set} is
W[1]-hard as well (see \cite{DBLP:journals/tcs/FellowsHRV09}). 
\textsc{Implications} was introduced and proved to be W[1]-hard in
\cite{Marx05:parametrized}. Lemma~\ref{lem:mimphard} below shows that
\textsc{Multicolored Implications} is W[1]-hard. The parameterized
complexity of \textsc{Biclique} is a longstanding open question, and
it is expected to be W[1]-hard (cf.~\cite[Section 8]{Grohe07:other-side}).

\begin{lemma}\label{lem:mimphard}
\textsc{Multicolored Implications} is \textup{W[1]}-hard.
\end{lemma}

\begin{proof}
  The proof is by reduction from \textsc{Clique}: let $H$ be a graph
  where a clique of size $k$ has to be found. It can be assumed that
  the number $n$ of vertices in $H$ equals the number of edges: adding
  isolated vertices or removing acyclic components does not change the
  problem. Let $u_1$, $\dots$, $u_n$ be the vertices of $H$ and let
  $e_1$, $\dots$, $e_n$ be the edges of $H$. We construct the graph
  $G$ of the \textsc{Multicolored Implications} instance as follows.
  Set $t:=k+\binom{k}{2}$. The vertex set of $G$ is $\{v_{i,j}\mid
  1\le i\le t,\ 1\le j\le n\}$.  Let $\iota(i,j)$ be a bijective
  mapping between the 2-element subsets of $\{1,\dots, k\}$ and the
  set $\{k+1,\dots, k+\binom{k}{2}\}$.  Intuitively, for $1\le i \le
  k$, the choice of vertex $v_{i,y_i}$ represents the choice of the
  $i$-th vertex of the clique $K$ of $H$ we are looking for.
  Furthermore, for $1 \le i < j \le k$, the choice of vertex
  $v_{\iota(i,j),y_{\iota(i,j)}}$ represents the edge between the
  $i$-th and $j$-th vertex of the clique $K$. To enforce this
  interpretation, for every $1 \le i < j \le n$ and $1 \le s \le n$,
  if $u_a$ and $u_b$ ($a<b$) are the endpoints of edge $e_s$ of $G$,
  then let us add two directed edges
  $\overrightarrow{v_{\iota(i,j),s}v_{i,a}}$ and
  $\overrightarrow{v_{\iota(i,j),s}v_{j,b}}$ to $H$. This completes
  the description of the reduction.

  Suppose that $G$ has a clique $u_{x_1}$, $\dots$, $u_{x_k}$ of size
  $k$ ($x_1<\dots<x_k$) and $e_{z_{i,j}}$ is the edge connecting $u_{x_i}$ and
  $u_{x_j}$. In this case, no directed edge of $H$ leaves the set $C$
  that contains $v_{i,x_i}$ for $1\le i\le k$ and
  $v_{\iota(i,j),z_{i,j}}$ for $1 \le i < j\le k$, hence it is a solution of the
  \textsc{Multicolored Implications} instance.

  For the reverse direction, suppose that $C=\{v_{i,y_i}\mid 1\le i\le
  t\}$ is a solution of the \textsc{Multicolored Implications}
  instance. We claim that $u_{y_1}$, $\dots$, $u_{y_k}$ is a clique of
  size $k$ in $G$. We claim that edge $e_{y_{\iota(i,j)}}$ connects
  $u_{y_i}$ and $u_{y_j}$ (in particular, this means that $y_i\neq
  y_j$). Let $s=y_{\iota(i,j)}$ and let $u_a$ and $u_b$ ($a<b$) be the
  two endpoints of edge $e_{s}$. Since $v_{\iota(i,j),s}\in C$, the
  edges $\overrightarrow{v_{\iota(i,j),s}v_{i,a}}$ and
  $\overrightarrow{v_{\iota(i,j),s}v_{j,b}}$ imply that $v_{i,a},
  v_{j,b}\in C$. As $C$ contains exactly one of the vertices
  $v_{i,1}$, $\dots$, $v_{i,n}$, it follows that $y_i=a$ and,
  similarly, $y_j=b$. That is, $u_{y_i}$ and $u_{y_j}$ are connected by
  the edge $e_s$.
\end{proof}

Let $\Gamma$ be a constraint language. A relation $R$ is
\emph{intersection definable} in $\Gamma$ if $R$ is the set of all
solutions to a certain instance of $\csp(\Gamma)$. 

\begin{proposition}\label{pro:int-defin}
Let $\Gamma$ be a constraint language and $R$ a relation
intersection definable in $\Gamma$. Then
$\ccsp(\Gamma\cup\{R\})$ ($\ocsp(\Gamma\cup\{R\})$) is polynomial time reducible
to $\ccsp(\Gamma)$ (respectively, to \ocspg). 
\end{proposition}

\begin{proof}
Indeed, let $I$ be an instance of \ccspg\ (\ocspg) that expresses $R$. To
reduce we just need to replace every occurrence of $R$ with $I$.
\end{proof}

\section{Properties of constraints}\label{sec:properties}

By Corollary~\ref{cor:0valid}, it is sufficient for proving
Theorems~\ref{th:mainsingleintro} and \ref{th:mainmultiintro} to
consider only cc0-languages. Thus in the rest of the paper, we always
assume that constraint languages are cc0-languages.

\subsection{Weak separability}
 
In this subsection we introduce the key property regulating the
complexity of CSPs with cardinality constraints.  In the Boolean case,
the tractability of 0-valid constraints depends only on weak
separability \cite{Marx05:parametrized}. This is not true exactly this
way for larger domains: as we shall see (Theorems~\ref{th:singlemain}
and \ref{th:multimain}), the complexity characterizations have further
technical conditions.

 \begin{definition}
Two tuples $\bt_1=(a_1,\dots,a_r)$ and $\bt_2=(b_1,\dots,b_r)$
 of the same length $r$
are called {\em disjoint} if $a_i=0$ or $b_i=0$ for every $1 \le i \le
r$. The {\em union} of two disjoint tuples $\bt_1$ and
$\bt_2$ is the tuple $\bt_1+\bt_2=(c_1,\dots, c_r)$ where $c_i=a_i$ if
$a_i\neq 0$ and $c_i=b_i$ otherwise. If $(a_1,\dots,a_r)$ is an extension of  
$(b_1,\dots,b_r)$, then their {\em difference} is the tuple
$(c_1,\dots, c_r)$ where $c_i=a_i$ if $b_i=0$ and $c_i=0$ otherwise. A
tuple $\bt=(c_1,\dots,c_r)$ is \emph{contained} in a set $C\sse D\setminus \{0\}$ if every nonzero $c_i$
 is in $C$.
 \end{definition}

The following property plays a crucial role in our classification:
 \begin{definition}
A 0-valid relation $R$ is said to be {\em weakly separable} if the following two
conditions hold:
 \begin{enumerate}
 \item 
 For every pair of disjoint tuples $\bt_1,\bt_2\in
  R$, we have $\bt_1+\bt_2\in R$.
 \item 
 For every pair of disjoint tuples $\bt_1,\bt_2$ with
  $\bt_2,\bt_1+\bt_2\in R$, we have $\bt_1\in R$.
 \end{enumerate}
A constraint language $\Gamma$ is said to be weakly separable if every 
relation from $\Gamma$ is weakly separable. 
 \end{definition}

 If constraint language $\Gamma$ is not weakly separable, then a
 triple $(R,\bt_1,\bt_2)$, $R\in\Gamma$, witnessing that is called a
 {\em union counterexample} if $\bt_1,\bt_2$ violate condition~(1),
 while if $\bt_1,\bt_2$ violate condition~(2) it is called a {\em
   difference counterexample.} Examples~\ref{exa:weaklysep1} and
 \ref{exa:weaklysep} demonstrate these notions and show how can we prove hardness in the Boolean case  if there is a
 counterexample to weak separability \cite{Marx05:parametrized}. However, as
 Example~\ref{exa:nonweaklysepeasy} shows, in case of larger domains,
 \ocspg\ or \ccspg\ can be fixed-parameter tractable even if the
 cc0-language $\Gamma$ is not weakly separable.

 Throughout the paper, we often refer to satisfying assignments that
 are nonzero, but have as few nonzero elements as possible.  A
 satisfying assignment $f$ is said to be a {\em minimal satisfying
   assignment} if it has a nonzero value and is not a proper extension
 of any nonzero satisfying assignment (note that we explicitly exclude
 the all-0 assignment from this definition). A consequence of
 Lemma~\ref{lem:minextension} bounds the number of minimal satisfying
 assignments:
\begin{lemma}\label{lem:minsolution}
Let $\Gamma$ be a finite constraint language. There are functions $d_\Gamma(k)$ and
$e_\Gamma(k)$ such that for any  
instance of $\csp(\Gamma)$ with $n$ variables, every
variable $v$ is nonzero in at most  
$d_{\Gamma(k)}$ minimal satisfying assignments of size at most $k$
and all these minimal satisfying assignments can be enumerated in time
$e_{\Gamma}(k)n^{O(1)}$. 
\end{lemma}
\begin{proof}
  Let $\delta_{v,d}$ be the assignment that assigns value $d$ to
  variable $v$ and $0$ to every other variable. For a fixed variable
  $v$ and every nonzero $d\in \dom(\Gamma)$, let us use the algorithm of
  Lemma~\ref{lem:minextension} to enumerate all the minimal satisfying
  extensions of $\delta_{v,d}$. We claim that every minimal satisfying
  assignment $f$ with $f(v)\neq 0$ appears in the enumerated
  assignments. Indeed, if $f(v)=d$, then $f$ is a satisfying extension
  of $\delta_{v,d}$, and clearly it is a minimal satisfying extension,
  as no nonzero subset of $f$ is satisfying. Thus
  $d_\Gamma(k)=(|\dom(\Gamma)|-1)d'_\Gamma(k)$ bounds the number of minimal
  satisfying assignments where variable $v$ is nonzero. The time bound
  follows from the time bound of Lemma~\ref{lem:minextension}.
\end{proof}

The following combinatorial property is the key for solving weakly
separable instances (this property does not necessarily hold for arbitrary
relations, see Example~\ref{exa:nonweaklysep}):

\begin{lemma}\label{lem:disjoint-decomp}
Let $\Gamma$ be a weakly separable finite cc0-language 
over $D$ and $I$ an instance of \ccspg\ or $\ocsp(\Gamma)$.
\begin{enumerate}
\item Every
satisfying assignment $f$ of $I$ is the 
union of pairwise disjoint minimal satisfying assignments.
\item If there is a satisfying assignment $f$ with
$f(v)=d$ for some variable $v$ and $d\in D$, then there is a minimal
satisfying assignment $f'$ with $f'(v)=d$.
\end{enumerate}
\end{lemma}

\begin{proof}
  1.  The proof is by induction on the size of $f$; if the size is
  0, then there is nothing to show. Let $f'$ be a subset of $f$ that
  is a minimal satisfying assignment. Let $f''$ be the difference of
  $f$ and $f'$; since $\Gamma$ is weakly separable, $f''$ is also a
  satisfying assignment. By the induction hypothesis, $f''$ is the
  disjoint union of minimal satisfying assignments $f_1,\dots,f_\ell$,
  and hence $f$ is the disjoint union of $f_1,\dots,f_\ell$, $f'$.

  2.  Since $f$ is the disjoint union of minimal satisfying
  assignments $f_1,\dots,f_\ell$, there has to be an $1\le i \le \ell$
  for which $f_i(v)=d$.
\end{proof}

Lemma~\ref{lem:disjoint-decomp} allows us to solve weakly 
separable instances by enumerating all minimal satisfying 
assignments and then finding a subset of these assignments 
that are disjoint and together satisfy the size/cardinality constraint. 
When finding these disjoint assignments, we can exploit the 
fact that by Lemma~\ref{lem:minsolution}, each such 
assignment is non-disjoint from a bounded number of other assignments.
\begin{theorem}\label{lem:finddisjoint}
If $\Gamma$ over $D$ is a finite weakly separable cc0-language, then 
\ccspg\ and $\ocsp(\Gamma)$ are fixed-parameter tractable.
\end{theorem}

\begin{proof}
  We present the proof for \ccspg; the proof for \ocspg\ is analogous
  and actually simpler. Alternatively, 
  by Lemma~\ref{lem:ocsp-to-ccsp}, $\ocsp(\Gamma)$ reduces to
  $\ccsp(\Gamma)$ in polynomial time.  Thus the fixed-parameter
  tractability of \ccspg\ immediately implies those of \ocspg.

  First, we enumerate all minimal satisfying assignments of size at
  most $k$ using Lemma~\ref{lem:minsolution}; let $S$ be the set of
  all these assignments (note that $S$ has size at most
  $d_{\Gamma}(k)\cdot n$, where $n$ is the number of variables).  By
  Lemma~\ref{lem:disjoint-decomp}, every solution can be formed as the
  disjoint union of members of $S$.  Furthermore, by weak
  separability, the disjoint union of satisfying assignments is always
  satisfying.  Thus the question is whether it is possible to find a
  subset $S'$ of $S$ that contains pairwise disjoint assignments and
  their union satisfies the cardinality constraints.

  Clearly, $|S'|\le k$.  We can associate a mapping $\pi_f$ to each
  assignment $f$ in $S$, with the meaning that $f$ sets exactly
  $\pi_f(i)$ variables to value $i$. Let $K:=k\cdot d_\Gamma(k)$ for
  the function $d_\Gamma$ appearing in Lemma~\ref{lem:minsolution}.
  For each mapping $\pi'$, let $S_{\pi'}$ contain the first $K$
  assignments whose associated mapping is $\pi'$ (or all such
  assignments if there are less than $K$ of them). Let $S^*$ be the
  union of all these sets $S_{\pi'}$. We claim that if there is a
  solution $S'\subseteq S$, then there is a solution which is a subset
  of $S^*$. Thus we can find a solution by trying all subsets of size
  at most $k$ in $S^*$.

To prove the claim, let $S'$ be a solution such that $|S'\setminus S^*|$ is minimum
  possible. Let $f\in S'\setminus S^*$; this means that $f\not \in
  S_{\pi_f}$ and hence $|S_{\pi_f}|=K$.  Assignments in $S'\setminus
  \{f\}$ assign nonzero values to less than $k$ variables, denote by
  $X$ these variables. By Lemma~\ref{lem:minsolution}, each variable is
  nonzero in at most $d_\Gamma(k)$ assignments of $S$ (and of the
  subset $S_{\pi_f}$), and hence there is at least one assignment $f'\in
  S_{\pi_f}$ that is zero on every variable of $X$. Replacing $f$ with
  $f'\in S^*$ yields a solution with strictly smaller number of
  assignments not in $S^*$, contradicting the minimality of $S'$.
\end{proof}

\subsection{Morphisms}

Homomorphisms and polymorphisms are standard tools for understanding
the complexity of constraints
\cite{Bulatov05:classifying,Jeavons99:expressive}. We make use of the
notion of multivalued morphisms, a generalization of homomorphisms,
that in a different context has appeared in the literature for a while
(see, e.g.\ \cite{Rosenberg98:hyperstructures}) under the guise of
hyperoperation.
For a constraint language, we introduce a classification of values from its 
domain into 4 types according to the
existence of such morphisms of the language (Definition~\ref{def:4-types}). This
classification into 4 types of elements and the observation that these
types play an essential role in the way the MVM gadgets
(Section~\ref{sec:gadgets}) work are the main technical ideas behind
the hardness proofs.

For a subset $0\in D'\sse D$ and an
$n$-ary relation $R$ on $D$, by $R_{|D'}$ we denote the relation $R\cap
(D')^n$.
For a constraint language $\Gamma$, the language $\Gamma_{|D'}$ denotes the
set of all relations $R_{|D'}$ for $R\in\Gamma$.

An {\em endomorphism} of $\Gamma$ is a mapping $h:\dom(\Gamma)\to
\dom(\Gamma)$ such that $h(0)=0$ and for every $R\in \Gamma$ and
$(a_1,\dots,a_r)\in R$, the tuple $(h(a_1),\dots, h(a_r))$ is also in
$R$. Note that the requirement $h(0)=0$ is nonstandard, but it is very
natural in our setting. For a tuple $\bt=(a_1,\dots, a_r)\in
\dom(\Gamma)^r$, we denote by $h(\bt)$ the tuple $(h(a_1),\dots,
h(a_r))$. Observe that the mapping sending all elements of
$\dom(\Gamma)$ to $0$ is an endomorphism of any 0-valid language. An
{\em inner homomorphism} of $\Gamma$ from $D_1$ to $D_2$ with $0\in
D_1,D_2\sse \dom(\Gamma)$ is a mapping $h:D_1\to D_2$ such that
$h(0)=0$ and $h(\bt)\in R$ holds for any relation $R\in
\Gamma$ and $\bt\in R_{|D_1}$.

A {\em multivalued morphism} of $\Gamma$ is a mapping
$\phi:\dom(\Gamma)\to 2^{\dom(\Gamma)}$ such that $\phi(0)=\{0\}$ and
for every $R\in \Gamma$ and $(a_1,\dots,a_r)\in R$, we have
$\phi(a_1)\times \dots \times \phi(a_r)\subseteq R$. An \emph{inner
  multivalued morphism} $\phi$ from $D_1$ to $D_2$ where $0\in
D_1,D_2\sse \dom(\Gamma)$ is defined to be a mapping $\phi: D_1\to
2^{D_2}$ such that $\phi(0)=\{0\}$ and for every $R\in \Gamma$ and
$(a_1,\dots,a_r)\in R_{|D_1}$, we have $\phi(a_1)\times \dots \times
\phi(a_r)\subseteq R_{|D_2}$. If $\phi$ is an (inner) multivalued
morphism, and $\bt=(a_1,\dots,a_r)$ is a tuple, then we define
$\phi(\bt)=\phi(a_1)\times \dots \times \phi(a_r)$.
Example~\ref{exa:endo-homo} shows several concrete examples.

\begin{observation}\label{obs:submorphisms}
Let $\phi:\dom(\Gamma)\to 2^{\dom(\Gamma)}$ be a multivalued morphism
  [$\psi:D_1\to 2^{D_2}$ 
  be an inner multivalued morphism] of a constraint language $\Gamma$,
and let $\phi':\dom(\Gamma)\to 2^{\dom(\Gamma)}$ [resp., $\psi':D_1\to
  2^{D_2}$] be a mapping 
such that $\phi'(d)\sse \phi(d)$ for $d\in \dom(\Gamma)$ [resp., $\psi'(d)\sse
  \psi(d)$ for $d\in D_1$]. Then $\phi'$ is a multivalued morphism
[$\psi'$ is an inner multivalued morphism]. 
\end{observation}

The {\em product} $g\circ h$ of two endomorphisms or inner
homomorphisms is defined by $(g\circ h)(x)=h(g(x))$ for every $x\in
D$. That is, $g\circ h$ means that $g$ is applied first and then
$h$. It is easy to see that $g\circ h$ is also an endomorphism (inner
homomorphism). If $\phi$ and $\psi$ are (inner) multivalued morphisms,
then their product $\phi\circ\psi$ is given by
$(\phi\circ\psi)(x)=\bigcup_{y\in\phi(x)}\psi(y)$. Finally, let $g$ be
an endomorphism or an inner homomorphism, and $\phi$ is an (inner)
multivalued morphism. Then $\phi\circ g$ is given by $(\phi\circ
g)(x)=\{g(y)\mid y\in\phi(x)\}$ and $g \circ \phi$ is given by $(g
\circ \phi)(x)=\phi(g(x))$. Both mappings are (inner) multivalued morphisms.

We classify the values based on a type of degeneracy defined in 
the following way (see     Example~\ref{exa:produces}):

\begin{definition} 
For $x,y\in \dom(\Gamma)$, we say that {\em $x$ produces $y$ in
    $\Gamma$} if $\Gamma$ has a multivalued morphism $\phi$ with
  $\phi(x)=\{0,y\}$ and $\phi(z)=\{0\}$ for every $z\neq x$. 
\end{definition}

In other words, for every $R\in \Gamma$ and $\bt\in R$, replacing an
arbitrary subset of the $x$ values with $y$ and making every other
coordinate 0 gives another tuple in $R$.  Observe that the relation
``$x$ produces $y$'' is transitive, but not necessarily reflexive.
The following classification of values plays a central role in the
paper:
\begin{definition}\label{def:4-types}
A value $y\in \dom(\Gamma)$ is

\begin{enumerate}
\item {\em regular} if there is no multivalued morphism $\phi$
  where $0,y\in \phi(x)$ for some $x\in \dom(\Gamma)$,
\item {\em semiregular} if there is a multivalued morphism $\phi$
  where 
  $0,y\in \phi(x)$ for some $x\in \dom(\Gamma)$, but there is no $x\in
  \dom(\Gamma)$ that produces 
  $y$,
\item {\em self-producing} if $y$ produces $y$, and for every $x$
  that produces $y$, $y$ also produces $x$.
\item {\em degenerate} otherwise.
\end{enumerate}
\end{definition}

Example~\ref{exa:types} demonstrates all four types.
It will sometimes be convenient to say that a value $y$ has {\em type}
1, 2, 3, or 4 corresponding to the four cases of
Definition~\ref{def:4-types}, and use the natural order on these
types.  
We need some simple properties of the four types of elements.

\begin{proposition}\label{prop:homomtype}
If there is an endomorphism $h$ with $h(x)=y$, then the type of $x$
cannot be larger than that of $y$. 
\end{proposition}

\begin{proof}
  Suppose that $x$ is semiregular, and $\psi$ is a multivalued
  morphism witnessing that, i.e., $x,0\in \psi(a)$ for some value $a$.
  Then $0,y\in (\psi\circ h)(a)$, meaning that $y$ cannot be regular.
  Suppose that some value $a$ produces $x$, let $\psi$ be the
  corresponding multivalued morphism. Then $\psi \circ h$ witnesses
  that $a$ produces $y$ as well. Therefore, if $x$ is self-producing,
  then $y$ is either self-producing or degenerate. Finally, suppose
  that $x$ is degenerate and let $a$ be a value such that $a$ produces
  $x$, but $x$ does not produce $a$.  Now $a$ produces $y$: this is
  shown by taking the product $\psi\circ h$ of a multivalued morphism
  $\psi$ witnessing that $y$ produces $x$ and the endomorphism $h$. If
  $y$ produces $a$ and this is witnessed by multivalued morphism
  $\psi'$, then (as shown by $h\circ \psi'$) $x$ would produce $a$, a
  contradiction. Thus $a$ produces $y$, but $y$ does not produce $a$,
  i.e., $y$ is degenerate as well.
\end{proof}

\begin{proposition}\label{prop:regproduce}
Every degenerate value $y$ is produced by a nondegenerate
value $x$.
\end{proposition} 
\begin{proof}
  We construct a sequence $x_0,x_1,x_2,\dots$ as follows: Let
  $x_0=y$. If $x_i$ is degenerate, then there is a value $x_{i+1}\neq x_i$ that
  produces $x_i$, but $x_i$ 
  does not produce $x_{i+1}$. As $\dom(\Gamma)$ is
  finite, either $x_i$ is non-degenerate for some $i$ or $x_i=x_j$ for
  some $i>j$. If some $x_i$ is non-degenerate, then by
  transitivity such $x_i$ produces $x_0=y$ and we are done.   
Suppose that $x_i=x_j$ for some $i>j$.
Since $x_{j+1}\neq x_j$, we have
  $i>j+1$. Thus $x_i=x_j$ produces $x_{j+1}$, contradicting the
  definition of $x_{j+1}$.
\end{proof}

\begin{lemma}\label{lem:substitute}
Let $\Gamma$ be a cc0-language, $R\in \Gamma$ a relation, and $h$ an endomorphism of $\Gamma$.
If $\bt_1$ and $\bt_2$ are disjoint tuples such that $\bt_1,\bt_1+\bt_2\in R$,
then $\bt_1+h(\bt_2)\in R$.
\end{lemma}

\begin{proof}
  Without loss of generality, we assume that
  $\bt_1=(a_1,\dots,a_r,0,\dots,0)$ and $\bt_2=(0,\dots,0,b_1,\dots,
  b_q,0,\dots,0)$. Let $R'$ be the relation obtained from $R$ by
  substituting $a_1, \dots, a_r$ on the first $r$ coordinates,
  respectively. Since $\bt_1\in R$, we know that $R'$ is 0-valid and
  hence it is in $\Gamma$. Thus $h$ is an endomorphism of $R'$ as well. From
  $\bt_1+\bt_2\in R$, we have $(b_1,\dots,b_q,0,\dots,0)\in R'$.
  Applying $h$ on this tuple gives $(h(b_1),\dots,h(b_q),0,\dots,0)\in
  R'$, which implies, by definition of $R'$, that $(a_1,\dots,
  a_r,h(b_1),\dots,h(b_q),0,\dots,0)=\bt_1+h(\bt_2)\in R$.
\end{proof}

We can extend Lemma~\ref{lem:substitute} to multivalued morphisms with
the same proof:
\begin{lemma}\label{lem:substitute2}
  Let $\Gamma$ be a cc0-language, $R\in \Gamma$ a relation, and $\phi$
  a multivalued morphism of $\Gamma$.  If $\bt_1$ and $\bt_2$ are
  disjoint tuples such that $\bt_1,\bt_1+\bt_2\in R$, then
  $\bt_1+\bt_2'\in R$ for every $\bt_2'\in \phi(\bt_2)$.
\end{lemma}

\subsection{Components}

The structure of endomorphisms and inner homomorphisms plays an
important role in our study. The following notion helps to capture
this structure. 
\begin{definition} 
Let $\Gamma$ be a cc0-language. A {\em retraction}
  to $X\subseteq D\setminus \{0\}$ is a 
  mapping $\pr_X$ such that $\pr_X(x)=x$ for $x\in X$ and $\pr_X(x)=0$
  otherwise.  A nonempty subset $C\sse D\setminus \{0\}$ is a {\em
    component} of $\Gamma$ if $\pr_C$ is an endomorphism of
  $\Gamma$. A component $C$ is {\em 
    minimal} if there is no component that is a proper subset of
  $C$.
\end{definition}

Note that by definition, a component contains only nonzero values. The
set $D\setminus \{0\}$ is trivially a component.  If a set $C$ is
not a component, then there is a relation $R\in\Gamma$ and $\bt\in R$
such that $\bt'=\pr_C\bt\not\in R$.

We prove certain combinatorial properties of components.  It is easy
to see that, as the composition of retractions is a retraction, the set of 
components is closed under intersection (if the intersection is not empty):
\begin{observation}\label{prop:compintersect}
  The intersection of two nondisjoint components is also a
  component. Hence for every nonempty $X\subseteq D\setminus \{0\}$,
  there is a unique inclusion-wise minimal component that contains
  $X$; this component is called the component {\em generated by $X$}
  (or simply the component of $X$).
\end{observation}

In case of cc0-languages, components are closed also under union:
\begin{proposition}\label{prop:compunion}
If $\Gamma$ is a cc0-language, then the union
of two components is also a component.
\end{proposition} 

\begin{proof}
Suppose that $C_1,C_2$ are components. For a relation $R\in\Gamma$ and tuple
$\bt$, let $\bt_1=\pr_{C_1}\bt$, $\bt_2=\pr_{C_2\setminus C_1}\bt$,
$\bt_3=\pr_{D\setminus (C_1\cup C_2\cup\{0\})}\bt$; clearly, $\bt=\bt_1+\bt_2+\bt_3$. As
  $C_1$ is a component, we have $\bt_1\in R$. Thus by
  Lemma~\ref{lem:substitute}, we have
  $\bt_1+\pr_{C_2}(\bt_2+\bt_3)=\bt_1+\bt_2=\pr_{C_1\cup C_2}\bt\in R$. This is
  true for every $R$ and $\bt$, thus $C_1\cup C_2$ is a component.
\end{proof}

A consequence of Proposition~\ref{prop:compunion} is that the component
generated by a subset $X\subseteq D\setminus \{0\}$ is the union
$C_X=\bigcup_{d\in X}C_d$, where $C_d$ is the component generated by
$d$. Indeed, $C_X$ is a component by Proposition~\ref{prop:compunion},
and it is clear that no proper subset of $C_X$ can be a component
containing $X$.

\begin{proposition}\label{prop:setitemgen}
If $\Gamma$ is a cc0-language and $d\in \dom(\Gamma)$ is in the
component generated by $X$ for some $X\sse\dom(\Gamma)\setminus \{0\}$, then there is
a $d'\in X$ such that $d$ is
in the component generated by $d'$.
\end{proposition}

The difference of two components is not necessarily a component, but in this case there is a difference counterexample:
\begin{proposition}\label{prop:compdiff}
If $C_1$ and $C_2$ are two  components of $\Gamma$ such that the nonempty set $C_1\setminus C_2$ is not a component, then $\Gamma$ has a difference counterexample contained in $C_1$.
\end{proposition}
\begin{proof}
  As $C_1\cap C_2$ is also a component by
  Observation~\ref{prop:compintersect}, we may assume that
  $C_2=C_1\cap C_2$, that is, $C_2\subseteq C_1$.  As $C_1\setminus
  C_2$ is not a component, there is an $R\in \Gamma$ and a $\bt\in R$
  such that $\bt_1=\pr_{C_1\setminus C_2} \bt\not\in R$. Since both
  $C_1$ and $C_2$ are components, we have $\bt_2=\pr_{C_2}\bt\in R$ and
  $\bt_1+\bt_2=\pr_{C_1}\bt\in R$. Thus $(R,\bt_1,\bt_2)$ is a difference
  counterexample.
\end{proof}

The following statement will be used when we restrict the language to
a subset of the domain:
\begin{proposition}\label{prop:comprestrict}
Let $0\in D' \sse \dom(\Gamma)$ be such that $D'\setminus \{0\}$ is a component of
$\Gamma$. For every $d\in D'$, 
\begin{enumerate}
\item the component generated by $d$ is the same in $\Gamma$
  and $\Gamma_{|D'}$.
\item the type of $d$ in $\Gamma_{|D'}$ is
not greater than that in $\Gamma$.
\end{enumerate}
\end{proposition}

\begin{proof}
Let $C$ (resp., $C'$)  be the component generated by $d$ in $\Gamma$
(resp., $\Gamma_{|D'}$). Since $D'\setminus \{0\}$ is a component containing $d$, we
have $C\subseteq D'\setminus \{0\}$. As $\pr_{C}$ is an endomorphism of $\Gamma$, it is an
endomorphism of $\Gamma_{|D'}$ as well, thus $C'\subseteq C$.
Furthermore,
as $\pr_{D'\setminus 0}\circ \pr_{C'}$ is an endomorphism of $\Gamma$,
the set $C'$ is a component of $\Gamma$, implying $C\subseteq C'$.

For the second statement, suppose that $d$ is semiregular in
$\Gamma_{|D'}$ and let $\psi$ be a multivalued morphism witnessing
this, i.e., $0,d\in \psi(c)$ for some $c\in D'$. Then $\pr_{D'\setminus \{0\}}\circ
\psi$ witnesses that $d$ cannot be regular in $\Gamma$. Let us next
show that for any $a,b\in D'$, $a$ produces $b$ in $\Gamma$ if
and only if $a$ produces $b$ in $\Gamma_{|D'}$. The forward direction
follows from the fact that for any multivalued morphism $\phi$ of
$\Gamma$, the mapping $\phi \circ \pr_{D'\setminus \{0\}}$ is a multivalued morphism
of $\Gamma_{|D'}$, and if $\phi(a)=\{0,b\}$, then
$(\phi \circ \pr_{D'\setminus \{0\}})(a)=\phi(a)=\{0,b\}$. 
For the backward direction, let $\psi$ be a multivalued morphism of
$\Gamma_{|D'}$ witnessing that $a$ produces $b$. Then $\pr_{D'}\circ
\psi$ witnesses that $a$ produces $b$ in $\Gamma$. It is now
immediate that if $d$ is degenerate in $\Gamma_{|D'}$, then it is
degenerate in $\Gamma$ as well, and if $d$ is self-producing in
$\Gamma_{|D'}$, then (as $d$ produces itself) it cannot be regular or
semiregular in $\Gamma$.
\end{proof}

The importance of components comes from the following result: there is
 a counterexample to weak separability where each of $\bt_1$ and $\bt_2$ is
 contained in one component. This observation will be essential in the
 hardness proofs.

\begin{lemma}\label{lem:regularcounter}
If $\Gamma$ is not weakly separable, then there is a
counterexample $(R,\bt_1,\bt_2)$ which is either
 \begin{enumerate}
 \item 
 a union counterexample, and $\bt_1$ (resp., $\bt_2$) is
  contained in a component generated by a 
  value $a_1$ (resp., $a_2$), or
 \item 
 a difference counterexample, and both $\bt_1$ and $\bt_2$ are
  contained in a component generated by a 
  value $a_1$.
 \end{enumerate}
\end{lemma}

\begin{proof}
  Let $K_1$, $\dots$, $K_r$ be the distinct components generated by
  the nonzero values in $\Gamma$. Assume first that there are two
  components $K_i$, $K_j$ that intersect; without loss of generality,
  we can assume that $K_i\setminus K_j\neq \emptyset$. Let $a$ be a
  value that generates $K_i$. Clearly, $a\not \in K_i\cap K_j$:
  otherwise Prop.~\ref{prop:compintersect} implies that $K_i\cap
  K_j\subset K_i$ is a component containing $a$, contradicting the
  assumption that $a$ generates $K_i$.  Thus $K_i\setminus K_j\subset
  K_i$ is not a component, since otherwise it would be a strictly
  smaller component containing $a$. Now Prop.~\ref{prop:compdiff}
  implies that there is a difference counterexample contained in the component $K_i$
  generated by $a$, satisfying the requirements.

Thus in what follows, we can assume that $K_1$, $\dots$, $K_r$ are
pairwise disjoint, i.e., they partition $D$. Suppose
that there is a union counterexample $(R,\bt_1,\bt_2)$. Tuple $\bt_1$ can be
represented as a union $\bt_{1,1}+\dots+\bt_{1,r_1}$ of nonzero
disjoint tuples such that each $\bt_{1,i}=\pr_{K_{j_i}}\bt_1$ is
contained in one of the components 
$K_1,\dots,K_r$.  The tuples $\bt_{2,1},\dots,\bt_{2,r_2}$ are
defined similarly.  Let $\bs_1,\dots,\bs_{r_1+r_2}$ be an
arbitrary ordering of these $r_1+r_2$ tuples. It is clear that each
$\bs_i$ is in $R$, since $K_1$, $\dots$, $K_r$ are components. As the
union of these tuples is not in $R$, there
is an integer $j\ge 1$ such that the union of any $j$ of these tuples
is in $R$, but there is a subset of $j+1$ tuples whose union is not in
$R$. Without loss of generality, suppose that  $\bigcup_{i=1}^{j+1}\bs_i$ is
not in $R$. If $j=1$, then $(R,\bs_1,\bs_2)$ is a required
counterexample. If $j>1$ then define
$\bs_0:=\sum_{i=1}^{j-1}\bs_i$. By assumption, 
$\bs_0\in R$, hence substituting the nonzero values of $\bs_0$ into $R$ as
constants gives a 0-valid relation $R'$. Furthermore,
$\bs_0+ \bs_j, \bs_0+\bs_{j+1}\in R$ by the definition of $j$; let $\bs'_j,\bs'_{j+1}\in
R'$ be the corresponding tuples. Now $(R',\bs'_j,\bs'_{j+1})$ is a union
counterexample: tuples $\bs'_j$, $\bs'_{j+1}$ are disjoint
and $\bs'_j+\bs'_{j+1}\not\in R'$ follows from $\bs_0+\bs_j+
\bs_{j+1}\not\in R$. 

Thus we can assume that there is no union counterexample.
Suppose that there is difference counterexample
$(R,\bt_1,\bt_2)$.  Assume that $(R,\bt_1,\bt_2)$ is minimal in the sense that
$\bt_1+\bt_2$ has minimal number of nonzero coordinates among such
counterexamples.  We claim 
that $\bt_1+\bt_2$ is contained in one of the components $K_i$ defined
above. Suppose that $\bt_1+\bt_2$ contains nonzero values from components
$K_1,\dots,K_g$ for $g\ge 2$. We show that $\pr_{K_i}(\bt_1)\in R$
for every $1 \le i \le g$.  This clearly holds if $\pr_{K_i}(\bt_1)$
equals $\pr_{K_i}(\bt_1+\bt_2)$ or the zero tuple. Thus we can assume that
$\pr_{K_i}(\bt_1+\bt_2)\in R$ is the disjoint union of nonzero tuples
$\pr_{K_i}(\bt_1)$ and $\pr_{K_i}(\bt_2)\in R$. If $\pr_{K_i}(\bt_1)\not\in
R$, then $(R,\pr_{K_i}(\bt_1),\pr_{K_i}(\bt_2))$ is a difference
counterexample, contradicting the minimality of $(R,\bt_1,\bt_2)$ as $g\ge 2$. Thus
$\pr_{K_i}(\bt_1)\in R$ for every $i$. However, the disjoint union of
these tuples is also in $R$ (since by assumption there is no union
counterexample), that is, $\bt_1\in R$, a contradiction. It follows that
$\bt_1+\bt_2$ belongs to some component $K_i$, i.e., there is a value $a$ that
generates $K_i$.
\end{proof}

\subsection{Multivalued morphism gadgets}\label{sec:gadgets}
The main technical tool in the hardness proofs are the gadgets defined
in this section. Intuitively, assuming that $D=\{0,1,\dots,\Delta\}$,
we want a gadget consisting of variables $v_0$, $v_1$, $\dots$,
$v_{\Delta}$ that has only two possible satisfying assignments: (1)
either 0 appears on every $v_i$, or (2) value $i$ appears on $v_i$ for
every $i$. Such gadgets would allow us to use a counterexample to weak
separability to prove hardness in similar way as hardness is proved in
the Boolean case (see Example~\ref{exa:weaklysep}). However, due to
the endomorphisms of the constraint language, such a gadget is not
always possible to create: a nontrivial endomorphism can be used to
transform satisfying assignments into new ones.

Therefore, our goal is more modest: we want a gadget that is either
fully zero or represents an endomorphism in every satisfying
assignment. That is, if variable $v_i$ gets value $c_i$, then the
mapping $h$ defined by $h(i)=c_i$ is an endomorphism. We 
enforce this requirement by introducing, for every $R\in \Gamma$ and
$(a_1,\dots,a_r)\in R$, a constraint $\langle
(v_{a_1},\dots,v_{a_r}),R\rangle$. Such a constraint ensures that
applying the mapping $h$ given by an assignment to the tuple
$(a_1,\dots,a_r)$ gives a tuple in $R$. 

The gadgets we use in the reductions are more general than the one
described in the previous paragraph: instead of a single variable
$v_i$ representing value $i$, we have a bag $B_i$ of a variables
representing this value. Setting the size of these bags and the
cardinality/size constraint is an essential and delicate part of the
reduction. The requirement that we want to enforce now is that if
$\psi(i)$ is the set of values appearing on the variables of $B_i$ in
a satisfying assignment, then $\psi$ is a multivalued morphism of
$\Gamma$. This can be ensured in a way similar to the construction in
the previous paragraph (see below for details).

A minor technical detail is that we defined morphisms in such a way
that 0 is always mapped to 0, thus there is no need to introduce
variables representing what 0 is mapped to; it is more convenient to
use constant 0's instead. We need the following definition to
formulate this conveniently.  For a relation $R$ and a tuple $\bt\in
R$, we denote by $\supph(\bt)$ the set of coordinate positions of $\bt$
occupied by nonzero elements.  Let $\supph_\bt(R)$ denote the relation
obtained by substituting 0 into all coordinates of $R$ except for
$\supph(\bt)$, i.e.\ if $R$ is $r$-ary and
$\supph(\bt)=\{1,\dots,r\}\setminus \{i_1,\ldots,i_q\}$, then
$\supph_\bt(R)=R^{|i_1,\ldots, i_q;0,\ldots,0}$.

For a cc0-language $\Gamma$ and some $0\in D'\sse \dom(\Gamma)$, a
{\em multivalued morphism gadget $\mvm(\Gamma,D')$} consists of
$|D'|-1$ bags of vertices $B_d$, $d\in D'\setminus\{0\}$. The number of
variables in each bag will be specified every time we use the gadget in a proof.
The gadget is equipped with the following set of constraints. For
every $R\in \Gamma$ and every tuple $\bt=(a_1,\dots,a_r)\in R_{|D'}$,
we add all possible constraints $\ang{\bs,\supph_\bt(R)}$ where
$\supph(\bt)=\{i_1,\ldots,i_q\}$, $\bs=(v_{{i_1}},\dots,
v_{{i_q}})$, and $v_{{i_j}}\in B_{a_{i_j}}$ for every 
$1\le j \le q$. The {\em standard assignment} of
a gadget assigns $a$ to every variable in bag $B_a$. Observe that the
standard assignment satisfies all the constraints of the gadget. We
say that bag $B_a$ and the variables in bag $B_a$ {\em represent}
$a$. When we say that a gadget is {\em fully nonzero,} then we mean
that all the variables are assigned nonzero values.

\begin{proposition}\label{prop:mvmgadget}
  Let $0\in D'\sse \dom(\Gamma)$. Consider a satisfying
  assignment $f$ of an $\mvm(\Gamma,D')$ gadget. If $h_f:D'\to
  2^{\dom(\Gamma)}$ is a mapping such that $h_f(a)$ is the set of values
  appearing in bag $B_a$ of the gadget and $h_f(0)=\{0\}$, then $h_f$ is
  an inner multivalued morphism of $\Gamma$ from $D'$ to $\dom(\Gamma)$.
\end{proposition}

\begin{proof}
  Let $R$ be a relation of $\Gamma$ and let $\bt=(a_1,\dots,a_r)\in
  R_{|D'}$. Let $(b_1,\dots, b_r)\in h_f(\bt)$. By the definition of $h_f$, for every $1\le i\le r$, either
  $a_i=b_i=0$ or $a_i\neq 0$ and there is a variable $v_i$ in bag
  $B_{a_i}$ having value $b_i$. Let
  $\supph_\bt=\{{i_1},\dots,{i_{r'}}\}$. The gadget is defined such
  that there is a constraint
  $\ang{(v_{i_1},\dots,v_{i_{r'}}),\supph_\bt(R)}$, implying that
  $(b_{i_1},\dots,b_{i_{r'}})\in \supph_\bt(R)$ and hence
  $(b_1,\dots,b_r)\in R$. It follows that $h_f(\bt)\subseteq R$.
\end{proof}

Note that if $D'=\dom(\Gamma)$, then $h_f$ is a multivalued morphism
of $\Gamma$.  Mapping $h_f$, or any mapping $h'$ with $h'(a)\subseteq
h_f(a)$ (for $a\in D'$), is said to be an (inner) multivalued morphism
{\em given} by the gadget and assignment $f$. If $|h'(a)|=1$ for all 
$a\in D'$, we call it an endomorphism (inner homomorphism) 
{\em given} by the gadget and assignment $f$.

We define two types of gadgets connecting MVM gadgets. The
gadget $\nand(G_1,G_2)$ on $\mvm(\Gamma,D')$ gadgets $G_1$,
$G_2$ consists of the following constraints. For every $R\in \Gamma$
and disjoint tuples $\bt_1=(a_1,\dots,a_r)$, $\bt_2=(b_1,\dots,b_r)$ in
$R_{|D'}$, we add a constraint $\ang{\bs,\supph_{\bt_1+\bt_2}(R)}$, for every
$\bs=(v_{i_1},\ldots, v_{i_q})$ with $\{i_1,\ldots
i_q\}=\supph(\bt_1+\bt_2)$, such that $v_{i_j}$
is in bag $B_{a_{i_j}}$ of $G_1$ if $a_{i_j}\neq 0$ and $v_{i_j}$ is in bag
$B_{b_{i_j}}$ of $G_2$ if $b_{i_j}\neq 0$.

It is not difficult to see (Lemma~\ref{lem:nandconstraint} below) that
if one of $G_1$, $G_2$ has the standard assignment and the other is
fully zero, then all the constraints in $\nand(G_1,G_2)$ are
satisfied. On the other hand, if both $G_1$ and $G_2$ have the
standard assignment and there is a union counterexample, then
$\nand(G_1,G_2)$ is not satisfied. For the reductions, we need this
second conclusion not only if both $G_1$ and $G_2$ have the standard
assignment, but also assignments that ``behave well'' in some
sense. The right notion for our purposes is the following: An inner
homomorphism $h:D'\to \dom(\Gamma)$ of $\Gamma$ is {\em
  $\bt$-recoverable} for some tuple $\bt$ if $h$ is invertible on
elements of $\bt$ in the following sense: $\Gamma$ has a multivalued morphism $\phi$ such
that $\bt\in (h\circ \phi)(\bt)$.  In particular, this is true if there is an endomorphism $h'$ of $\Gamma$ with 
$(h\circ h')(\bt)=\bt$. We say that a $\mvm(\Gamma,D')$  gadget is
$\bt$-recoverable in a given assignment if at least one of the inner homomorphisms given by
it is $\bt$-recoverable.

\begin{lemma}\label{lem:nandconstraint}
Let $0\in D'\sse \dom(\Gamma)$ and let there be a $\nand(G_1,G_2)$
gadget on $\mvm(\Gamma,D')$ gadgets $G_1$, 
$G_2$.
\begin{enumerate}
 \item 
 If one of $G_1$ and $G_2$ has the standard assignment and the
  other gadget is fully zero, then all constraints of $\nand(G_1,G_2)$ 
  are satisfied.
\item 
Let $f$ be a satisfying assignment of $\nand(G_1,G_2)$. If $h_i$ is 
an inner homomorphism given by gadget $G_i$ and assignment $f$ 
for $i=1,2$ and $\bt_1,\bt_2$ are disjoint tuples in $R_{|D'}$, 
then $h_1(\bt_1)+h_2(\bt_2)\in R$.

\item If there is a union counterexample $(R,\bt_1,\bt_2)$ in
  $\Gamma_{|D'}$ and gadget $G_i$ is $\bt_i$-recoverable in assignment $f$ (for
  $i=1,2$), then some constraint of $\nand(G_1,G_2)$ is not satisfied by $f$.
\end{enumerate}
\end{lemma}
\begin{proof}
  1. Suppose without loss of generality that $G_1$ has the standard
  assignment and $G_2$ is fully zero. Consider a relation $R\in
  \Gamma$ and disjoint tuples $\bt_1=(a_1,\dots,a_r)$,
  $\bt_2=(b_1,\dots,b_r)\in R_{|D'}$. In a corresponding constraint
  $\ang{(v_1,\dots,v_r),\supph_{\bt_1+\bt_2}(R)}$, if $a_i\neq 0$, then
  variable $v_i$ has value $a_i$ (since it is in bag $B_{a_i}$ of
  $G_1$) and has value 0 otherwise. Thus $(a_1,\dots,a_r)\in R$
  implies that the constraint is satisfied.

2. If the $\nand(G_1,G_2)$ instance is satisfied,
then one of the constraints corresponding to $\bt_1$ and $\bt_2$ ensures  that $h_1(\bt_1)+h_2(\bt_2)\in R$.

3. For $i=1,2$, let $h_i$ be a $\bt_i$-recoverable inner homomorphism
given by $G_i$ and let $\phi_i$ be a multivalued morphism such that
$\bt_i\in (h_i\circ \phi_i)(\bt_i)$.  Since $h_i$ is an inner
homomorphism, we have that $h_1(\bt_1),h_2(\bt_2)\in R$.  Statement 2
implies that $h_1(\bt_1)+h_2(\bt_2)\in R$.  By
Lemma~\ref{lem:substitute2}, we have that $h_1(\bt_1)+\bt'_2\in R$ for
any $\bt'_2\in \phi_2(h_2(\bt_2))$; in particular, this means that
$h_1(\bt_1)+\bt_2$ is in $R$. As $\bt_2\in R$, we can apply
Lemma~\ref{lem:substitute2} once more to get that $\bt'_1+\bt_2\in R$
for any $\bt'_1\in \phi_1(h_1(\bt_1))$; in particular, $\bt_1+\bt_2\in
R$, a contradiction.
\end{proof}

The $\imp(G_1,G_2)$ gadget is defined similarly, but instead of
$\bt_1,\bt_2\in R_{|D'}$, we require $\bt_2,\bt_1+\bt_2\in
R_{|D'}$. 

\begin{lemma}\label{lem:impconstraint}
Let $0\in D'\sse \dom(\Gamma)$ and let there be an $\imp(G_1,G_2)$
gadget on $\mvm(\Gamma,D')$ gadgets $G_1$, 
$G_2$.
\begin{enumerate}
\item 
 If $G_2$ has the standard assignment and $G_1$ either has the
  standard assignment or fully zero, then all constraints of $\imp(G_1,G_2)$ 
  are satisfied.
\item 
Let $f$ be a satisfying assignment of $\imp(G_1,G_2)$. If $h_i$ 
is an inner endomorphism given by gadget $G_i$ and assignment $f$ 
for $i=1,2$ and $\bt_1,\bt_2$ are disjoint tuples such with  
$\bt_2,\bt_1+\bt_2\in R_{|D'}$, then $h_1(\bt_1)+h_2(\bt_2)\in R$.

\item 
 If there is a difference counterexample $(R,\bt_1,\bt_2)$ in
  $\Gamma_{|D'}$, $G_1$
  is $\bt_1$-recoverable in assignment $f$ and $G_2$
  gives an inner homomorphism $h_2$ with $h_2(\bt_2)=0$, then some
   constraint of $\imp(G_1,G_2)$  is not satisfied.
\end{enumerate}
\end{lemma}
\begin{proof}
  1. Similarly to the proof of Lemma~\ref{lem:nandconstraint}, the
  constraint corresponding to $(R,\bt_1,\bt_2)$ is satisfied, since
  $\bt_2,\bt_1+\bt_2\in R$.

  2. As in Lemma~\ref{lem:nandconstraint},
  follows from the definition of $\imp(G_1,G_2)$. 

  3. Let $h_1$ be a $\bt_1$-recoverable inner homomorphism given by
  $G_1$. Let $\phi_1$ be a multivalued morphism such that $\bt_1\in
  (h_1\circ \phi_1)(\bt_1)$. One of the constraints of $\imp(G_1,G_2)$
  ensure that $h_1(\bt_1)+h_2(\bt_2)=h_1(\bt_1)$ is in $R$. Now
  $\bt_1\in \phi_1(h_1(\bt_1))$ is also in $R$, a contradiction.
\end{proof}

When the multivalued morphism gadgets are used in the reductions, it
will be essential that the bags of the gadgets have very specific
sizes. We will ensure somehow that in a solution each bag is either
fully zero or fully nonzero. Our aim is to choose the sizes of the
bags 
 in such a way that if the sum of the sizes of a collection of 
bags add up to  a certain integer,
then this is only possible if the collection contains exactly one bag
of each size.

In most of the reductions, the $\mvm(\Gamma,D')$ gadgets are arranged in $t$ groups (corresponding to the $t$ groups of vertices in the \textsc{Multicolored Independent Set} or \textsc{Multicolored Implications} instance we are reducing from). For a gadget in group $i$, the bag representing $d\in D'\setminus \{0\}$ has size 
$Z_{i,d}^{t,D'}$, which is defined as follows.
Fix an integer $t$ and a set $0\in D'\sse D$. It will be convenient to
assume that $D'=\{0,1,\ldots,\Delta\}$. For $1 \le i \le t$ and $1 \le d\le \Delta$, we define  
\[
Z_{i,d}^{t,D'}:=(4t\Delta)^{2t\Delta+(i\Delta+d)}+(4t\Delta)^{5t\Delta-(i\Delta+d)}. 
\]
By $\Z^{t,D'}$ we denote the set of integers $Z_{i,d}^{t,D'}$ for
$1\le i \le t$ and $1\le d \le \Delta$.  Note that these integers have
exactly two nonzero digits if written in base-$(4t\Delta)$ expansion;
the positions of these two digits depend on $i\Delta+d$. Furthermore,
the ``larger nonzero digit'' of any number in $\Z^{t,D'}$ is always
larger than the ``smaller nonzero digit'' of any other number in
$\Z^{t,D'}$.  We will use the following property of these integers:

\begin{lemma}\label{lem:unique}
  Let us fix $t$ and $D'=\{0,1,\ldots,\Delta\}$. If $A$ is a subset of
  $\Z^{t,D'}$ and $\B$ is a 
  multiset of values from $\Z^{t,D'}$ such that $|\sum_{S\in A}S-\sum_{S\in
    \B}S|< (4t\Delta)^{2t\Delta}$, then $\B$ is a set and $\B=A$.
\end{lemma}

\begin{proof}
 It can be assumed that $A\cap \B=\emptyset$, since
  removing integers from both $A$ and $\B$ does not change the difference
  of the sums. Let $T$ be the largest integer in $A\cup \B$. Assume
  first that $T\in \B\setminus A$. This means that $A$ contains only
  integers strictly smaller than $T$, and (as $A$ is a set) an integer
  appears at most once in $A$. Since there are $t\Delta$ integers in
  $\Z^{t,D'}$ and every integer smaller than $T$ is at most $2T/(4t\Delta)$, we have
  that the sum of the integers in $A$ is at most $T/2$. Thus the
  difference of the sums is at least $(4t\Delta)^{2t\Delta}$, a contradiction.

  Assume now that $T\in A\setminus \B$. Suppose that
  $T=Z^{t,D'}_{i,d}$ and let $X:=(4t\Delta)^{2t\Delta+(i\Delta+d)}$.
  Since $T$ is the largest integer in $A\cup \B$ (i.e, $i\Delta+d$ is
  as small as possible), every integer in $A\cup \B$ other than $T$ is
  divisible by $(4t\Delta)^{2t\Delta+(i\Delta+d+1)}=X\cdot 4t\Delta$, while $T$ is equal to $X$ modulo
  $X\cdot 4t\Delta$. Thus $\sum_{S\in A}S$ is $X$ modulo $X\cdot
  4t\Delta$, while $\sum_{S\in \B}S$ is divisible by $X\cdot 4t\Delta$.
  Therefore, $|\sum_{S\in A}S-\sum_{S\in \B}S|$ is at least $\min\{X,
  4t\Delta\cdot X-X\}>(4t\Delta)^{2t\Delta}$, a contradiction.
\end{proof}

\subsection{Frequent instances}\label{sec:frequent-instances}

If a value $d$ ``can appear'' only on a small number of variables,
then we can branch on all possible ways this value appears and then
reduce the problem to simpler instances whose domain does not contain
$d$. Formally, we say that an instance of $\ccsp(\Gamma)$ or
$\ocsp(\Gamma)$, with parameter $k$ is {\em $c$-frequent} (for some
integer $c$) if for every $d\in \dom(\Gamma)\setminus\{0\}$ there are
at least $c$ variables $v$ such that $f(v)=d$ for a  satisfying assignment $f$
of size at most $k$ (note that we do not require that these satisfying
assignments satisfy the cardinality constraints). The algorithm of
Lemma~\ref{lem:minextension} can be used to decide  whether
an instance is $c$-frequent.
As we shall see in Lemma~\ref{lem:frequent}, if an instance is not
$c$-frequent, then it can be reduced to $c$-frequent instances by
trying all possibilities for the values that appear on less than $c$
variables.  We prove a stronger result, where the instances satisfy
an additional technical requirement.  A subset $0\in D'\subseteq
\dom(\Gamma)$ is {\em closed} (with respect to $\Gamma$) if $\Gamma$
has no inner homomorphism from $D'$ to $\dom(\Gamma)$ that maps some
element of $D'$ to an element in $\dom(\Gamma)\setminus D'$.

\begin{lemma}\label{lem:frequent}
Let $\Gamma$ be a finite cc0-language. 
  Given an instance $I$ of $\ccsp(\Gamma)$ or $\ocsp(\Gamma)$ with parameter
  $k$ and an integer $c$, we can construct in time
  $f_{\Gamma}(k,c)n^{O(1)}$ a set of $c$-frequent instances
  such that

\begin{enumerate}
\item instance $I$ has a solution if and only if at least one of the
  constructed instances has a solution,
\item each instance $I_i$ is 
  an instance of $\ccsp(\Gamma_{|D_i})$, respectively,
  $\ocsp(\Gamma_{|D_i})$, for some $D_i\subseteq \dom(\Gamma)$ closed
  in $\Gamma$, and
\item the parameter $k_i$ of $I_i$ is at most $k$.
\end{enumerate}
\end{lemma}
\begin{proof}
We state the proof for $\ccsp(\Gamma)$, the proof is the same for
$\ocsp(\Gamma)$.  The reduction performs the following branching step
repeatedly.  Let $I'$ be the current instance, which is an instance of
$\ccsp(\Gamma_{|D'})$ for some $D'$. If $I'$ is not $c$-frequent, then
we branch as follows.  Let $S_d$ be the set of those variables where
value $d$ can appear in a satisfying assignment of size at most $k$.
This set can be found using the algorithm Lemma~\ref{lem:minextension}
as follows: To decide if $v\in S_d$, let us assign $d$ to $v$ and find
if there is any minimal satisfying extension of this assignment of
size at most $k$.  Suppose that $|S_d|<c$. We branch into
$\binom{|S_d|}{\pi(d)}< 2^{|S_d|}\le 2^c$ directions by considering every subset
$S'_d\subseteq S_d$ of size exactly $\pi(d)$ and creating an assignment that gives value $d$ to
the variables of $S'_d$, and 0 to the remaining variables.  If this
assignment does not satisfy $I'$, then we use
Lemma~\ref{lem:minextension} to enumerate all the minimal satisfying
extensions $f'_1$, $\dots$, $f'_t$ of this assignment. For each such
satisfying assignment $f'_i$, we obtain an instance $I'_i$ by
substituting the nonzero variables as constants. Since $f'_i$ is a
satisfying assignment, every relation of $I'_i$ is 0-valid, hence it
is an instance of $\ccsp(\Gamma_{|D'})$ as well. Furthermore, since we
have already considered all possible appearances of value $d$, the
correctness of the algorithm does not change if we consider $I'_i$ as
an instance of $\ccsp(\Gamma_{|D'\setminus \{d\}})$. If the instance
$I'_i$ is still not $c$-frequent, then we repeat the branching step.

  In each step, the maximum number of directions we branch into is at
  most a constant depending only on $c$, $\Gamma$, and the current
  parameter $k'\le k$. The depth of the branching tree is
  at most $|D|$, since we decrease the domain at each step. Thus it is
  clear that the running time is $f_{\Gamma}(k,c)n^{O(1)}$, for an
  appropriate function $f_{\Gamma}(k,c)$.

Let instance $I_i$ of $\ccsp(\Gamma_{|D_i})$ be a 0-valid $c$-frequent
instance obtained by the algorithm and let $k_i$ be its size
constraint.  To show that $D_i$ is closed, we argue as follows.
Suppose that there is an inner homomorphism $g:D_i\to D$ of $\Gamma$
such that $g(d)=d'$ for some $d\in D_i$, $d'\in D\setminus D_i$. Since
$I_i$ is $c$-frequent, there are $c$ variables $v_1$, $\dots$, $v_c$
and satisfying assignments $f_1$, $\dots$, $f_c$ of size at most $k_i$
such that $f_j(v_j)=d$. On the branch of the algorithm that produced
instance $I_i$, there has to be an instance $I^{(1)}_i$ of
$\ccsp(\Gamma_{|D^{(1)}})$  that is
reduced to an instance $I^{(2)}_i$ 
of $\ccsp(\Gamma_{|D^{(2)}})$ such
that $D_i\subseteq D^{(1)}$, 
$D^{(2)}=D^{(1)}\setminus \{d'\}$, and $d'$ is not $c$-frequent in
$I^{(1)}$. If we consider instance $I_i$ as an instance of
$\ccsp(\Gamma_{|D^{(1)}})$, then
the assignment $f_j\circ g$ assigns value 
$d'$ to $v_i$. Since instance $I_i$ is obtained from instance
$I^{(1)}$ via substitutions, we get that variables $v_1$, $\dots$,
$v_c$ can get value $d'$ in $I^{(1)}$ in assignments whose  size does
not exceed the parameter of $I^{(1)}$. This contradicts the assumption that $d'$ is
not $c$-frequent in $I^{(1)}$.
\end{proof}

\section{Classification for size constraints}\label{sec:size}

Unlike in the Boolean case, weak separability of $\Gamma$ is not
equivalent to the tractability of $\ocsp(\Gamma)$: it is possible that
$\Gamma$ is not weakly separable, but $\ocsp(\Gamma)$ is FPT (see
Example~\ref{exa:nonweaklysepeasy}). However, if there is a subset
$D'\sse \dom(\Gamma)$ of the domain such that $\Gamma_{|D'}$ is not
weakly separable and $D'$ has ``no special problems'' in a certain
technical sense, then $\ocsp(\Gamma)$ is W[1]-hard.  We need the
following definitions.  A value $d\in \dom(\Gamma)$ is {\em weakly
  separable} if $\Gamma_{|\{0,d\}}$ is weakly separable.  A {\em
  contraction} of $\Gamma$ to $D'$ with $0\in D' \subseteq
\dom(\Gamma)$ is an endomorphism $h: \dom(\Gamma)\to D'$ such that
$h(d)\neq 0$ for any $d\in \dom(\Gamma)\setminus \{0\}$. Contraction
$h$ is \emph{proper} if $D'\subset \dom(\Gamma)$.  As the contraction
can be applied on any solution of an $\ocsp(\Gamma)$ instance without
changing the number of nonzero variables, restriction to $D'$ does not
change the problem, i.e., replacing every $R$ with $R_{|D'}$ does not
change the solvability of the instance (see
Example~\ref{exa:nonweaklysep-easy}).

The main result for the size constraints CSP is the following
dichotomy theorem.

\begin{theorem}\label{th:singlemain}
Let $\Gamma$ be a finite cc0-language. If there are two
sets $\{0\}\sse D_2 \subseteq D_1 \subseteq \dom(\Gamma)$ such that
 \begin{enumerate}
 \item 
 $D_1$ is closed in $\Gamma$,
 \item 
 $\Gamma_{|D_1}$ has a contraction $h$ to $D_2$,
 \item 
 $\Gamma_{|D_2}$ has no proper contraction,
 \item 
 $\Gamma_{|D_1}$ has no weakly separable value that is either
  degenerate or self-producing, and
 \item 
 $\Gamma_{|D_2}$ is not weakly separable,
 \end{enumerate}
then \ocspg\ is \textup{W[1]}-hard. If there are no such $D_1,D_2$, then \ocspg\ is
fixed-parameter tractable.
\end{theorem}

 We present an algorithm solving the FPT cases of the problem in
 Section~\ref{sec:ocspalgorithm}.  Section~\ref{sec:ocsphardness}
 presents an important case of the hardness proof, demonstrating 
 the concepts introduced in Section~\ref{sec:properties}.

\subsection{The algorithm}\label{sec:ocspalgorithm}

Let $I=(V,\C,k)$ be an instance of \ocspg.
Let us use the algorithm of 
Lemma~\ref{lem:frequent} to obtain instances  $I_1,\dots,I_\ell$ such that $I_i$ is a
$k$-frequent instance of $\ocsp(\Gamma_{|D^i})$ for some closed set
$D^i\subseteq \dom(\Gamma)$.  Fix some $i$ and let $h$ be a contraction of
$\Gamma_{|D^i}$ such that $|h(D^i)|$ is minimum possible. Set
$D_1:=D^i$ and $D_2:=h(D^i)$. 

In the cases where Theorem~\ref{th:singlemain} claims fixed-parameter
tractability, the pair $D_1,D_2$ violates one of the properties
(1)--(5). By the way $D_1$ and $D_2$ defined, it is clear that (1) and
(2) hold. For (3), suppose that $\Gamma_{|D_2}$ has a proper
contraction $g$. Then $h\circ g$ is a contraction of $\Gamma_{|D_1}$
such that $|g(h(D_1))|$ is strictly less than $|h(D_1)|$, a
contradiction.

If $D_1$ violates (4), then instance $I_i$ always has a solution.
Indeed, suppose that $d\in D_1$ is weakly separable and $d$ is
produced by $d'\in D_1$ (possibly $d=d'$).  Let $k_i$ be the parameter
of $I_i$; then $k_i\le k$ by Lemma~\ref{lem:frequent}(3). Since $I_i$
is $k$-frequent, if we denote by $S$ the set of variables of $I_i$
where $d'$ can appear in a satisfying assignment of size at most $k$,
then $|S|\ge k$.
As $d'$ produces $d$, $\Gamma_{|D_1}$ has a multivalued morphism
$\phi$ such that $\phi(d')=\{0,d\}$ and $\phi(a)=\{0\}$ for $a\in
D_1\setminus\{d'\}$. Applying multivalued morphism $\phi$ on {\em any} satisfying
assignment $f$ with $f(v)=d'$ shows that the assignment $\delta_{v,d}$
with $\delta_{v,d}(v)=d$ and 0 everywhere else is a satisfying
assignment of $I_i$. Therefore, for every $v\in S$, assignment
$\delta_{v,d}$ satisfies $I_i$. As $d$ is weakly separable in
$\Gamma_{|D_1}$, the disjoint union of $k_i$ such assignments
$\delta_{v,d}$ is a solution to $I_i$.
Finally, suppose that (5) is violated and $\Gamma_{|D_2}$ is weakly separable. Instance
$I_i$ of $\ocsp(\Gamma_{|D_1})$ has a solution if and only if it has a
solution restricted to $D_2$ (because of the contraction $h$), and the latter can be decided using
the algorithm of 
Lemma~\ref{lem:finddisjoint}.

\subsection{Hardness}\label{sec:ocsphardness}

Suppose that we have sets $D_1$ and $D_2$ as in
Theorem~\ref{th:singlemain}. The aim of this section is to show that
\ocspg\ is \textup{W[1]}-hard in this case. The reduction is based on
a counterexample to the weak separability of $\Gamma_{|D_2}$, which
exists by condition (5) of Theorem~\ref{th:singlemain}. To ensure that
the MVM gadgets work as intended, we have to make use of conditions
(1)--(4) as well. Our first goal is to handle the cases when some
value in $\Gamma_{|D_1}$ is not regular
(Sections~\ref{sec:degen-self-prod-ocsp}--\ref{sec:semiregular-values-ocsp}). The
main part of the proof is to prove hardness in the case when every
value in $\Gamma_{|D_1}$ is regular
(Section~\ref{sec:regular-values-ocsp}). This case contains the most
important proof ideas; the reader is suggested to skim
Sections~\ref{sec:degen-self-prod-ocsp}--\ref{sec:semiregular-values-ocsp}
and concentrate on Section~\ref{sec:regular-values-ocsp} on a first
reading.

\subsubsection{Degenerate and self-producing values}
\label{sec:degen-self-prod-ocsp}
Recall that a relation $R$ is intersection definable in a constraint language 
$\Gamma$ if $R$ is the set of all solutions to a certain instance of $\csp(\Gamma)$. 
Let $U_\Gamma$ be the set of all 0-valid unary relations 
intersection definable in the set of 0-valid relations from $\Gamma$. 

\begin{lemma}\label{lem:unaryproduce}
Let $D_1$ be a closed set in $\Gamma$. If $P$ is the set of nonzero
values produced by $d\in D_1$ in $\Gamma_{|D_1}$, then 
$P\cup \{0\} \in U_\Gamma$. 
\end{lemma}

\begin{proof}
For every  ($r$-ary) $R\in \Gamma$, every tuple
$\bb\in R_{|D_1}$ where $d$ appears, and every subset $\ba=(a_1,\dots,a_r)$ of
$\bb$ that contains only 0 and $d$ (recall that it means that $a_i=b_i$ 
whenever $a_i\ne0$), we set $R_\ba=R^{|i_1,\ldots, i_q;0,\ldots,0}$, where
$i_1,\ldots,i_q$ are the positions such that $a_{i_j}=0$.  Let
$T$ be the unary relation expressed by the instance $(\{v\},\C)$,
where $\C$ contains constraints $\{\ang{(v,\ldots,v),R_\ba})$ for all
such $R\in\Gamma$, $\bb$, and $\ba$. We claim that $T=P\cup\{0\}$.

For every $a\in P$, the fact that $d$ produces $a$ in $\Gamma_{|D_1}$
implies the tuple obtained from $\bb\in R_{|D_1}$ by replacing  value
$d$ with $0$ or $a$ and replacing 
everything else with 0 gives a tuple of $R$. Thus setting $v$ to $a$
is a satisfying assignment.

On the other hand, suppose that there is a satisfying assignment with value 
$a$ on $v$. It follows that $\psi(d)=\{a,0\}$ and $\psi(d')=\{0\}$ for every
$d'\in D_1\setminus \{d\}$ is an inner multivalued morphism $\psi$ of 
$\Gamma$ from $D_1$. Indeed, similar to the previous paragraph it 
means that the tuple obtained from $\bb\in R_{|D_1}$ by replacing  value
$d$ with $0$ or $a$ and replacing everything else with 0 gives a tuple of $R$. 
By Observation~\ref{obs:submorphisms}, the mapping $h$ such
that $h(d)=a$ and $h(d')=0$ for every $d'\in D_1\setminus \{d\}$ is
an inner homomorphism. Since $D_1$ is closed, $a\in D_1$, thus $d$ produces $a$ in
$\Gamma_{|D_1}$.
\end{proof}

The following lemma proves hardness in the case when $\Gamma_{|D_1}$
contains some degenerate and self-producing values. In this case, by
condition (4) of Theorem~\ref{th:singlemain}, none of the values are
weakly separable.
\begin{lemma}\label{lem:hardproduce}
If $D_1$ is a closed set such that $\Gamma_{|D_1}$ has degenerate or
self-producing values,  but no such value is weakly separable, then
\ocspg\ is \textup{W[1]}-hard. 
\end{lemma}
\begin{proof}
Let $d\in D_1$ be a value that produces at least one nonzero value in $\Gamma_{|D_1}$; 
  let $P$ be the set of (nonzero) values produced by $d$. By
  Lemma~\ref{lem:unaryproduce}, $P\cup \{0\}\in U_\Gamma$. Let
  $P'\subseteq P$ be a smallest nonempty set such that $P'\cup
  \{0\}\in U_\Gamma$.

  Let $x$ be an arbitrary nonzero element of $P'$. As $x$ is produced
  by $d$, the assumption of the lemma implies that $x$ is not weakly
  separable in $\Gamma_{|D_1}$. Suppose first that $\Gamma_{|\{0,x\}}$
  has a union counterexample $(R,\bt_1,\bt_2)$, where $R\in
  \Gamma_{|\{0,x\}}$ is $r$-ary. Let $A,B$ be the set of coordinates
  of $R$ such that $\bt_1$ equals $x$ in positions from $A$, $\bt_2$
  equals $x$ in positions from $B$, and they are equal to 0
  otherwise. By substituting 0's, it can be assumed without loss of
  generality that $A\cup B=\{1,\ldots,r\}$ and $A=\{1,\ldots,q\}$. Let
  $R'$ be the binary relation expressed by the instance
  $(\{v,w\},\{\ang{(v,\ldots,v,w,\ldots,w),R})$ where $v$ occupies the
  first $q$ positions.
  As is easily seen, $(0,0),(x,0),(0,x)\in R'$, but $(x,x)\not\in R'$.
  We show that $(y,y')\not\in R'$ for arbitrary nonzero values
  $y,y'\in P'$. If this is true, then by
  Proposition~\ref{pro:int-defin} this binary relation $R'$ and the
  unary relation restricting to $P'\cup\{0\}$ can be used to reduce
  \textsc{Independent Set} to \ocspg. That is, a binary constraint
  with relation $R'$ can represent each edge: this ensures that that a
  set of variables can be simultaneously nonzero if and only if they
  correspond to an independent set in the graph.

  Suppose that $(y,y')\in R'$ for some $y,y'\in P'$ (including the
  possibility that $y$ or $y'$ is equal to $x$). If $(y,0)\not\in R'$,
  then by substituting 0 in the second coordinate we get a 0-valid
  unary relation that includes $x$, but does not include $y$,
  contradicting the minimality of $P'$ (the intersection of this set 
  with $P'\cup\{0\}$ is a proper subset of $P'\cup\{0\}$). It follows that $(y,x)\in R'$:
  otherwise, by substituting $y$ in the first coordinate, we would get
  a 0-valid relation containing $y'$, but not $x$. Finally, by
  substituting $x$ in the second coordinate, we get a 0-valid unary
  relation containing $y$, but not $x$, a contradiction.

Suppose next that   $\Gamma_{|\{0,x\}}$ has a difference
counterexample. Again, without loss of generality, it can be assumed
that $(0,0),(0,x),(x,x)\in R'$, but $(x,0)\not\in R'$ for a binary
relation $R'$ intersection definable in $\Gamma$. We claim that
$(y,0)\not\in R'$ for any nonzero $y\in 
P'$: otherwise, by substituting 0 in the second coordinate, we would
get a 0-valid unary relation containing $y$, but not $x$. Now by
Proposition~\ref{pro:int-defin},  we can 
use this binary relation $R'$ to reduce \textsc{Implications} to \ocspg.
\end{proof}

\subsubsection{Semiregular values}
\label{sec:semiregular-values-ocsp}
Thus in the following, we assume that $\Gamma_{|D_1}$ has no
degenerate or self-producing values. Next we prove hardness in the
case when there is a semiregular value $d\in D_1$ in $\Gamma_{|D_1}$.
We say that a multivalued morphism $\phi$ 
{\em witnesses} that $d$ is semiregular if $0,d\in \phi(c)$ for some $c\in\dom(\Gamma)$.

\begin{lemma}\label{lem:semihard}
If $D_1$ is a closed set in $\Gamma$ such that there are no self-producing or degenerate values in
$\Gamma_{|D_1}$, but there is a semiregular value $d\in D_1$, then \ocspg\ is
\textup{W[1]}-hard. 
\end{lemma}
\begin{proof}
 From the semiregular values
  of $\Gamma_{|D_1}$, let us choose $d$ and its witness $\phi$ such
  that the size of $S:=\phi(D_1)$ is minimum possible (note that $0\in
  S$).

The proof is by reduction from \textsc{Implications}. Let $G,t$ be an 
instance of \textsc{Implications}. For each vertex
$v_i$ of $G$, we introduce a gadget $\mvm(\Gamma,S)$ denoted by $G_i$. The size of
each bag of each gadget is $Z:=t+1$, except the bag corresponding to $d$,
whose size is 1. We set the cardinality constraint to
$k:=tZ(|S|-2)+t$. To finish the construction of the instance, we encode the
directed edges of the \textsc{Implications} instance by adding the
gadget $\imp(G_{x},G_{y})$ for each directed
edge $\overrightarrow{v_xv_y}$.

Suppose that there is a solution $C$ of size exactly $t$ for the
\textsc{Implications} instance. If vertex $v_i$ is in $C$, then set
the standard assignment on gadget $G_{i}$. It is clear that this
results in an assignment of size exactly $k=tZ(|S|-2)+t$ and the
constraints of the $\mvm(\Gamma,S)$ gadgets are satisfied. From the
fact that there is no directed edge $\overrightarrow{v_xv_y}$ with
$v_x\in C$, $v_y\not\in C$ and from Lemma~\ref{lem:impconstraint}(1),
it follows that the constraints of  $\imp(G_{x},G_{y})$ are
also satisfied.

For the other direction, we have to show that if there is a solution
of size exactly $k$ for the \ocspg\ instance, then there is a
solution $C$ of size exactly $t$ for \textsc{Implications}. Observe
first that only values from $D_1$ can appear in a solution: if some
value $c\in D\setminus D_1$ appears on a gadget $G_{x}$, then there is
a corresponding inner homomorphism $g_{x}$ of $\Gamma$ from $S$ to $D$
such that for a certain inner homomorphism from $D_1$ to $D$
given by $\phi\circ g_{x}$, value $c$ appears in the image of
$D_1$, contradicting the fact that $D_1$ is closed.

Next we show that if the variable in bag $B_d$ of $G_{x}$ is nonzero,
then $G_{x}$ is fully nonzero. Suppose that the variable in bag $B_d$
of $G_x$ is nonzero, but there is a nonzero variable in bag $B_c$ for
some some nonzero $c\in S$ (since there is a single variable in bag
$B_d$, we have $c\neq d$).  Let $g$ be an inner homomorphism of
$\Gamma_{|D_1}$ (from $S$) given by $G_{x}$ such that $g(d)\neq 0$ and
$g(c)=0$. The multivalued morphism $\phi \circ g$ witnesses that
$g(d)$ is semiregular in $\Gamma_{|D_1}$ (as there are no
self-producing or degenerate values by assumption), and $(\phi\circ
g)(D_1)=g(S)$ has size strictly less than $S$, contradicting the
minimality of $S$.

Let us show that for every $v_x$ and $c\in S\setminus \{0,d\}$, the bag
$B_c$ of
$G_{x}$ is either fully zero or fully nonzero.
Suppose that both 0 and $d'$ appear in this bag. Let $\psi'_{x}$ be an
inner multivalued morphism of $\Gamma_{|D_1}$ from $S$ given by $G_{x}$ such that
$\psi'_{x}(c)=\{0,d'\}$ and $|\psi'_{x}(c')|=1$ for every $c'\in S$,
$c'\neq c$. The multivalued morphism $\phi \circ \psi'_{x}$ witnesses
that $d'$ is semiregular in $\Gamma_{|D_1}$ (note that by assumption, there are no self-producing or degenerate values in $\Gamma_{|D_1}$). Thus the minimality of
$S$ would be violated by $\psi'_{x}(d)=\{0\}$, hence the variable in
bag $B_d$ is nonzero. In this case, by the
previous paragraph, the nonzero variable in bag $B_d$ implies that
every variable of $G_{x}$ is nonzero.

Since the bags not corresponding to $d$ have size $Z$ and the
cardinality constraint $k$ equals $t$ modulo $Z$, there have to be at
least $t$ gadgets where  bag $B_d$ is nonzero. We
have seen that in these gadgets all the other $Z(|S|-2)$ variables are
nonzero as well.  Thus it follows that there are exactly $t$ gadgets
where all the variables are nonzero, and every variable of every other
gadget is zero.

Let $C$ be the set such that $v_x\in C$ if and only if the variables
of $G_{x}$ are nonzero; the previous paragraph implies that $|C|=t$.
We claim that $C$ is a solution for the \textsc{Implications}
instance.  Suppose that there is an edge $\overrightarrow{v_xv_y}$
with $v_x\in C$ and $v_y\not\in C$. Let $c$ be the value of the
variable of $G_{x}$ in the bag $B_d$. We arrive to a contradiction by
showing that in this case $c$ is produced by some value $d'$ in
$\Gamma_{|D_1}$ (recall that by assumption, there are no self-producing or
degenerate values in $\Gamma_{|D_1}$). Let $d'\in D_1$ be such that
$0,d\in \phi(d')$. To show that $d'$ produces $c$, let $\bt\in R$ for
some $R\in \Gamma_{|D_1}$, and let $\bt_c$ be a tuple such that $c$ is
the only nonzero value appearing in $\bt_c$ and whenever $c$ appears
in some coordinate of $\bt_c$, then $d'$ appears in the same
coordinate of $\bt$. If we show that every such $\bt_c$ is in $R$,
then we prove that $d'$ produces $c$. Let $\bt_d$ be the same as
$\bt_c$ with every $c$ replaced by $d$.  The multivalued morphism
$\phi$ shows that there is a tuple $\bt'\in R_{|S}$ disjoint from
$\bt_c$ such that $\bt_d+\bt'\in R_{|S}$. Gadget $G_x$ gives an inner
homomorphism $f_x$ with $f_x(d)=c$ and gadget $G_y$ gives an inner
homomorphism $f_y$ with $f_y(\bt')=\mathbf0$, the zero tuple. As both $\bt'\in R_{|S}$ and $\bt_d+\bt'\in R_{|S}$ hold, it follows from 
Lemma~\ref{lem:impconstraint}(2) that the $\imp(G_{x},G_{y})$ constraint
implies that $f_x(\bt_d)+f_y(\bt')=f_x(\bt_d)=\bt_c\in R$, and we are
done.
\end{proof}

\subsubsection{Regular values}
\label{sec:regular-values-ocsp}
Let $D_1,D_2$ be a pair satisfying (1)--(5) of
Theorem~\ref{th:singlemain}. By previous sections, we can assume in
the following, that every value is regular in $\Gamma_{|D_1}$. It
follows that every value is regular in $\Gamma_{|D_2}$ as well: if
$\psi$ is a multivalued morphism of $\Gamma_{|D_2}$ with $0,d\in
\psi(c)$ for some nonzero $c,d\in D_2$, then, as $h$ is a contraction, $h\circ \psi$ witnesses
that $d$ is not regular in $\Gamma_{|D_1}$.

A technical tool in the proofs is that given a set of endomorphisms that
are disjoint in the sense that the image of a value is nonzero in
exactly one of the endomorphisms, we would like to construct a mapping
that is the ``sum'' of these mapping, hoping that it is also an
endomorphism.  Formally, we say that a set $p_1$, $\dots$, $p_\ell$ of
endomorphisms of $\Gamma$ is a {\em partition set} if, for every $d\in
D\setminus \{0\}$, $p_i(d)\neq 0$ for exactly one $i$. The {\em sum}
of the partition set is the mapping $h:D\to D$ defined such that
$h(d)$ is the unique nonzero value in $p_1(d)$, $\dots$, $p_\ell(d)$.
The partition set is {\em good} if the sum of these pairwise disjoint
endomorphisms is also an endomorphism; otherwise, the partition set is
{\em bad.} 
We can define partition sets similarly for inner endomorphisms.

The hardness proofs are simpler if we assume that there are no bad
partition sets: we prove W[1]-hardness under this assumption in
Lemma~\ref{lem:regularhard} below if there is a union counterexample
and in Lemma~\ref{lem:regularhardimp} if there is a difference
counterexample.

Note that if a partition set is bad, then there is a union
counterexample in $\Gamma$ using the values
$\bigcup_{i=1}^{\ell}p_i(\dom(\Gamma))$. Indeed, suppose that there is
a relation $R\in \Gamma$ and a tuple $\bt\in R$ such that
$h(\bt)=p_1(\bt)+p_2(\bt)+\dots +p_\ell(\bt)\not\in R$. If $1<\ell'\le
\ell$ is the smallest value such that $p_1(\bt)+p_2(\bt)+\dots+
p_{\ell'}(\bt)\not\in R$, then $\bt_1=p_1(\bt)+p_2(\bt)+\dots+
p_{\ell'-1}(\bt)$ and $\bt_2=p_{\ell'}(\bt)$ is a union counterexample.
Lemma~\ref{lem:regularhardset} exploits this union counterexample to
show hardness in case there is a bad partition set, completing the
proof of Theorem~\ref{th:singlemain}.

\begin{lemma}\label{lem:regularhard}
If every value is regular in $\Gamma_{|D_2}$, there is no bad
partition set in $\Gamma_{|D_2}$,  and there is a union
counterexample in $\Gamma_{|D_2}$, then \ocspg\ is \textup{W[1]}-hard.
\end{lemma}

\begin{proof}
  The reduction is from \textsc{Multicolored Independent Set} (see
  Section~\ref{sec:reductions}). Assume $D_2=\{0,\ldots,\Delta\}$. 
For each vertex $v_{x,y}$ ($1\le x
  \le t$, $1 \le y \le n$), we introduce a gadget $\mvm(\Gamma,D_2)$
  denoted by $G_{x,y}$. The bag of $G_{x,y}$ corresponding to value
  $d\in D_2\setminus \{0\}$ has size $Z^{t,D_2}_{x,d}$. The size
  constraint is $k:=\sum_{x=1}^t\sum_{d\in D_2\setminus
    \{0\}}Z^{t,D_2}_{x,d}$.  If $v_{x,y}$ and $v_{x',y'}$ are
  adjacent, then we add the gadget $\nand(G_{x,y},G_{x',y'})$.
  Furthermore, for every $1\le x \le t$, $1 \le y,y' \le n$, $y\neq y'$, we add
  the $\nand(G_{x,y},G_{x,y'})$ gadget.

Suppose that there is a solution $C$ for the
\textsc{Multicolored Independent Set} instance.  If vertex $v_{x,y}$
is in $C$, then set the standard assignment on gadget $G_{x,y}$,
otherwise set the zero assignment. It is 
clear that this results in an assignment satisfying the size constraint.
The constraints of the $\mvm(\Gamma,D_2)$ gadgets are
satisfied and the constraints of $\nand(G_{x,y},G_{x,y'})$ are
satisfied as well (by Lemma~\ref{lem:nandconstraint}(1)).

For the other direction, suppose that there is a solution satisfying
the size constraint. First, we observe that a solution contains
values only from $D_1$. Indeed, if $c\not\in D_1$ appears in
bag $B_d$ of a gadget $G_{x,y}$, then $G_{x,y}$ gives an
inner homomorphism $g$ of $\Gamma$ from $D_2$ with $g(d)=c$. Now $h\circ g$  maps a  
value of $D_1$ to $c$, contradicting the assumption that $D_1$ is a
closed set. Furthermore, by applying the contraction
 $h$
on a solution, it can be assumed that only values from $D_2$ are used.
Thus the $\mvm(\Gamma,D_2)$ gadgets give multivalued morphisms of $\Gamma_{|D_2}$.
Since every value is regular in $\Gamma_{|D_2}$, each bag is either
fully zero or fully nonzero.  The sizes of the nonzero bags add up
exactly to the size constraint $k$. Thus by Lemma~\ref{lem:unique},
the only way this is possible if there is exactly one nonzero bag with
size $Z^{t,D_2}_{x,d}$ for every $1\le x \le t$ and $d \in
D_2\setminus \{0\}$.  However, the nonzero bags of sizes
$Z^{t,D_2}_{x,d_1}$ and $Z^{t,D_2}_{x,d_2}$ could appear in different
gadgets.

Take a union counterexample $(R,\bt_1,\bt_2)$ in $\Gamma_{|D_2}$; by
Lemma~\ref{lem:regularcounter}, we can assume that $\bt_1$, $\bt_2$
are in the components of $\Gamma_{|D_2}$ generated by some $a_1,a_2\in
D_2$, respectively.  We show that for every $1 \le x \le t$, there are
values $y^1_x$ and $y^2_x$ such that every endomorphism of
$\Gamma_{|D_2}$ given by $G_{x,y^1_x}$ (resp., $G_{x,y^2_x}$) is
$\bt_1$-recoverable (resp., $\bt_2$-recoverable).  For a fixed $x$,
let $g_1$, $\dots$, $g_n$ be arbitrary endomorphisms of
$\Gamma_{|D_2}$ given by $G_{x,1}$, $\dots$, $G_{x,n}$, respectively.
Since the sizes of nonzero bags are all different, these endomorphisms
are pairwise disjoint and hence they form a partition set.  As there
is no bad partition set in $\Gamma_{|D_2}$, their sum $g$ is an
endomorphism of $\Gamma_{|D_2}$ and in fact a contraction. Since $\Gamma_{|D_2}$ has no proper
contraction by assumption, $g$ has to be a permutation and hence $g^s$
is the identity for some $s\ge 1$. There is a unique $1 \le y^1_x \le
n$ such that $g_{y^1_x}(a_1)=g(a_1)\neq 0$.  The endomorphism $g_{y^1_x}
\circ g^{s-1}$ maps $a_1$ to $(g\circ g^{s-1})(a_1)=g^s(a_1)=a_1$ and maps every $a\in D_2$ either to 0 or $a$; i.e.,
$g_{y^1_x}\circ g^{s-1}=\pr_S$ for some set $S\subseteq D_2\setminus
\{0\}$ containing $a_1$. As $S$ is a component containing $a_1$, it
has to contain the component generated by $a_1$ and $S$ contains every
value of $\bt_1$. It follows that $g_{y^1_x}$ given by $G_{y^1_x}$ is
$\bt_1$-recoverable: $g^{s-1}(g_{y^1_x}(\bt_1))=\bt_1$.
A similar argument works 
for $y^2_x$, thus the required values $y^1_x$, $y^2_x$ exist.
Let us observe that it is not possible that $y^1_x\neq y^2_x$: by Lemma
\ref{lem:nandconstraint}(3) the constraints of
$\nand(G_{x,y^1_x},G_{x,y^2_x})$ are not satisfied in this
case. Let $C$ contain vertex $v_{x,y}$ if $y=y^1_x=y^2_x$. It
follows that $C$ is a multicolored independent set: if
vertices $v_{x,y}$, $v_{x',y'}$ are adjacent, then again by Lemma
\ref{lem:nandconstraint}(3), some constraint of
$\nand(G_{x,y},G_{x',y'})$=$\nand(G_{x,y^1_x},G_{x',y^2_x})$
is not satisfied.
\end{proof}

The proof using a difference counterexample is similar:
\begin{lemma}\label{lem:regularhardimp}
If every value is regular in $\Gamma_{|D_2}$, there is no bad
partition set in $\Gamma_{|D_2}$, and there is a difference
counterexample in $\Gamma_{|D_2}$, then \ocspg\ is \textup{W[1]}-hard.
\end{lemma}

\begin{proof}
The reduction is from \textsc{Multicolored Implications}. For each
vertex $v_{x,y}$ ($1\le x \le t$, $1 \le y \le n$), we introduce a
gadget $\mvm(\Gamma,D_2)$ denoted by $G_{x,y}$. The bag $B_d$
of $G_{x,y}$ has size $Z^{t,D_2}_{x,d}$. The
size constraint is $k:=\sum_{i=x}^t\sum_{d\in D_2\setminus \{0\}}Z^{t,D_2}_{x,d}$.
If there is a directed edge $\overrightarrow{v_{x,y},v_{x',y'}}$, then
we add a constraint
$\imp(G_{x,y},G_{x',y'})$.

Suppose that there is a solution $C$ of size exactly $t$ for the
\textsc{Multicolored Implications} instance. If vertex $v_i$ is in $C$, then set
the standard assignment on gadget $G_{v_i}$, otherwise set the zero
assignment. It is clear that this 
results in an assignment satisfying the size constraint. The
constraints of the $\mvm(\Gamma,D_2)$ gadgets are satisfied and the
$\imp(G_{x,y},G_{x,y'})$ constraints are satisfied as well
(Lemma~\ref{lem:impconstraint}(1)). 

For the other direction, suppose that there is a solution satisfying
the size constraint. As in Lemma~\ref{lem:regularhard}, we can
assume that only values from $D_2$ are used in the solution. Since
every value is regular in 
$\Gamma_{|D_2}$, every bag is either fully zero or fully nonzero. The
sizes of the nonzero bags add up exactly to the size constraint
$k$. Thus 
by Lemma~\ref{lem:unique}, the only way this is possible is if there
is exactly one nonzero bag of size $Z^{t,D_2}_{x,d}$ for 
every $1\le x \le t$ and $d \in D_2$.

 Let us choose a difference counterexample
$(R,\bt_1,\bt_2)$ in $\Gamma_{|D_2}$; by Lemma~\ref{lem:regularcounter},
we can assume that $\bt_1+\bt_2$ is in the component $C_1$ of $\Gamma_{|D_2}$
generated by some $a_1\in D_2$. We show that for every
$1 \le x \le t$, there is an integer $y_x$ such that
$G_{x,y_x}$ is $\bt_1$-recoverable.

Let $g_1$, $\dots$, $g_n$ be arbitrary
endomorphisms given by $G_{x,1}$, $\dots$, $G_{x,n}$, respectively.
The uniqueness of the sizes of the nonzero bags implies that these
endomorphisms are pairwise disjoint and they form a partition set.
We assumed that there is no bad partition set in $\Gamma_{|D_2}$,
thus the sum $g$ of the set is an endomorphism. Since $\Gamma_{|D_2}$ has no
proper contraction, we have that $g$ is a permutation and $g^s$
is the identity for some $s\ge 1$. There is a $1 \le y_x \le n$ such
that $g_{y_x}(a_1)\neq 0$. The homomorphism $g_{y_x} \circ
g^{s-1}$ maps every value $a\in D_2$ either to 0 or $a$; i.e.,
$g_{y_x}\circ g^{s-1}=\pr_S$ for some set $S\subseteq D_2$
containing $a_1$. This means that $S$ is a component containing $a_1$,
hence $C_1\subseteq S$. It follows that $g_{y_x}$
is $\bt_1$-recoverable. Moreover, as the bag $B_a$ of $G_{x,y_x}$ for $a\in C_1$ is fully nonzero, we have that bag $B_a$ of $G_{x,y}$ for $a\in C_1$ and $y\neq y_x$ is fully zero.

Let $C=\{v_{1,y_1},\dots,v_{t,y_t}\}$.  It follows immediately that
$C$ does not violate any of the implications: if there is a
directed edge $\overrightarrow{v_{x,y_x}v_{x',y'}}$, $y'\neq y_{x'}$,
then the gadget $G_{x,y_x}$ is $\bt_1$-recoverable and
$h(\bt_2)=0$ for every homomorphism given by $G_{x',y'}$, as $\bt_2$
is contained in $C_1$, and
therefore Lemma~\ref{lem:impconstraint}(3) implies that
$\imp(G_{x,y_x},G_{x',y'})$ is not satisfied.
\end{proof}

The last step is to handle the case when there is a bad partition
set. As mentioned earlier, this implies that there is a union
counterexample; the following proof exploits this fact.

\begin{lemma}\label{lem:regularhardset}
If every value is regular in $\Gamma_{|D_2}$ and there is a bad
partition set in $\Gamma_{|D_2}$, then \ocspg\ is \textup{W[1]}-hard.
\end{lemma}

\begin{proof}
  Let $p_1$, $\dots$, $p_\ell$ be a minimal bad partition set of $\Gamma_{|D_2}$ in the
  sense that $D_3:=\bigcup_{i=1}^{\ell}p_i(D_2)$ has
  minimum size.  Assume $D_3=\{0,\ldots,\Delta\}$. 
Because of the bad partition set, $\Gamma_{|D_3}$
  contains a union counterexample. Furthermore, every value $d\in
  D_3\setminus \{0\}$
  is regular in $\Gamma_{|D_3}$: if $\Gamma_{|D_3}$ has a multivalued morphism $\psi$
  with $0,d\in \psi(c)$ for some $c\in D_3$, and $p_i(c')=c$ for some
  $1\le i\le \ell$ and $c'\in D_2$, then $p_i \circ \psi$ witnesses that $d$ is not
  regular in $\Gamma_{|D_2}$.

The reduction is the same as in Lemma~\ref{lem:regularhard}, with the
only difference is that we use $\mvm(\Gamma,D_3)$ gadgets instead of
$\mvm(\Gamma,D_2)$ and the sizes of the bags are set using the
values $Z^{t,D_3}_{x,d}$. It remains true that a solution for
\textsc{Multicolored Independent Set} implies a solution for the
\ocspg\ instance.

For the other direction, let us argue first that only values from
$D_1$ appear in a solution.  Suppose that a value $d\not\in D_1$
appears in bag $B_c$ of a gadget $G_{x,y}$, which means that $G_{x,y}$
gives an inner homomorphism $g$ from $D_3$ to $D$ with $g(c)=d$. Let
$1\le s \le \ell$ be such that $p_s(c')=c$ for some $c'\in D_2$. Now
$h \circ p_s \circ g$ is an inner homomorphism from $D_1$ to $D$
mapping a value of $D_1$ to $d$, contradicting the assumption that
$D_1$ is a closed set. Furthermore, by applying the contraction $h$ on
a solution, it can be assumed in the following that only values from
$D_2$ appear in the solution. That is, each multivalued gadget
describes an inner multivalued morphism from $D_3$ to $D_2$.

We show that every bag is either fully zero or fully nonzero. Suppose
that $\psi$ is an inner multivalued morphism of $\Gamma_{|D_2}$ from
$D_3$ given by a gadget with $0,d\in \psi(c)$ for some nonzero $c\in D_3$ and $d\in D_2$. Suppose
that $p_s(c')=c$ for some $c'\in D_2$ and $1\le s \le \ell$. Now $p_s\circ \psi$
witnesses that $d$ is not regular in $\Gamma_{|D_2}$, and this
contradiction shows that every bag is either zero or fully nonzero.
The sizes of the nonzero bags add up exactly to the size constraint
$k$. Thus by Lemma~\ref{lem:unique}, there is exactly one nonzero bag
with size $Z^{t,D_3}_{x,d}$ for every $1\le x \le t$ and $d \in
D_2\setminus \{0\}$.

We know that there is a union counterexample in $\Gamma_{|D_3}$
(because of the bad partition set whose image is in $D_3$).
Let us choose a union counterexample $(R,\bt_1,\bt_2)$ in $\Gamma_{|D_3}$;
by Lemma~\ref{lem:regularcounter}, we can assume that $\bt_i$ is in the
component of $\Gamma_{|D_3}$ generated by some $a_i\in D_3$, for
$i=1,2$.  We show that for every $1 \le x \le t$, there are values
$y^1_x$ and $y^2_x$ such that $G_{x,y^1_x}$ (resp., $G_{x,y^2_x}$)
gives a $\bt_1$-recoverable (resp., $\bt_2$-recoverable) inner
homomorphism from $D_3$ to $\Gamma_{|D_2}$.

Let $p$ be the sum of this bad partition set $p_1$, $\dots$, $p_\ell$
(note that $p$ is {\em not} an endomorphism of $\Gamma_{|D_2}$).  The
uniqueness of the sizes of the nonzero bags imply that at most
$|D_3|-1$ of the gadgets $G_{x,1}$, $\dots$, $G_{x,n}$ have nonzero
bags. Furthermore, if we choose one inner homomorphism given by each
such gadget, then it is clear that these inner homomorphisms $g_1$,
$\dots$, $g_m$ form a partition set, i.e., for any $a\in D_3$, the
value $g_{i}(a)$ is nonzero for exactly one $i$. Let $g$ be the sum of
$g_1$, $\dots$, $g_m$ (note that we have no reason to assume that $g$
is an inner homomorphism from $D_3$ to $\Gamma_{|D_2}$).

We show that $g\circ p$ is a permutation of $D_3$.  Let $P$ be the set
of all endomorphisms of $\Gamma_{|D_2}$ that arise in the form
$p_{z_1} \circ g_{z_2} \circ p_{z_3}$ for some $1 \le z_1,z_3\le
\ell$, $1 \le z_2 \le m$. Observe that for every $a\in D_2$, there is
a unique triple $(z_1,z_2,z_3)$ such that $(p_{z_1} \circ g_{z_2}
\circ p_{z_3})(a)$ is nonzero: this follows from the fact that both
$p_1,\dots,p_\ell$ and $g_1,\dots,g_m$ are partition sets. Thus the
endomorphisms in $P$ also form a partition set. Let $p^*$ be the sum
of this set.  We have $p^*(D_2)\subseteq (g\circ p)(D_3)$: if
$p^*(a)=b$, then there is an $a'\in D_3$ with $p_{z_1}(a)=a'$ and
$(g_{z_2}\circ p_{z_3})(a')=b$ for some $z_1,z_2,z_3$.  If $g\circ p$
is not a permutation of $D_3$, then $(g\circ p) (D_3)$ has size
strictly smaller than $|D_3|$, and hence $p^*(D_2)$ has size strictly
smaller than $|D_3|$ as well.  If $p^*$ is an endomorphism of
$\Gamma_{|D_2}$, then (as there are no proper contractions by
assumption) it has to be a permutation, contradicting
$|p^*(D_2)|<|D_3|\le |D_2|$.  Otherwise, suppose that $p^*$ is not an
endomorphism, i.e., the partition set $P$ is bad.  Now
$|p^*(D_2)|<|D_3|=|p(D_2)|$ contradicts the minimality of the bad
partition set $p_1$, $\dots$, $p_\ell$.

Since $g\circ p$ is a permutation, there is an $s\ge 1$ such that
$(g\circ p)^s$ is the identity. This means that for an arbitrary
sequence $u_1, u'_1,\dots,u_s,u'_s$, the endomorphism $(g_{u_1}\circ
p_{u'_1} \circ \dots \circ g_{u_s} \circ p_{u'_s})$ is a retraction
$\pr_S$ of $\Gamma_{|D_3}$, and we can choose the sequence such that
$a_1\in S$. As $S$ 
is a component containing $a_1$, it contains all the values of $\bt_1$.
It follows that $g_{u_1}$ is $\bt_1$-recoverable, hence we can set
$y^1_{x}:=u_1$. The values $y^2_{x}$ can be defined similarly. From
this point, we can finish the proof as in Lemma~\ref{lem:regularhard}.
\end{proof}

\section{Classification for cardinality constraints}\label{sec:cardinality}

The characterization of the complexity of \ccspg\ requires a new
definition, which was not relevant for \ocspg. 
The {\em core} of $\Gamma$ is the component generated by the set of
all nondegenerate values in $\dom(\Gamma)$. Note that by
Proposition~\ref{prop:regproduce}, the set of nondegenerate values
is not empty, and thus the core is not empty. We say that
$\Gamma$ is a core if the core of $\Gamma$ is $\dom(\Gamma)$ (see
Example~\ref{exa:ccsphard}). 

\begin{lemma}\label{lem:corecore}
Let $\Gamma$ be a finite cc0-language over $D$. If $C\subseteq D$ is the core of $\Gamma$, then $\Gamma_{|C\cup\{0\}}$ is a core.
\end{lemma}
\begin{proof}
  Every nondegenerate value of $\Gamma$ is in $C$. By
  Lemma~\ref{prop:comprestrict}(2), every such value is nondegenerate also
  in $\Gamma_{|C\cup\{0\}}$ and by Lemma~\ref{prop:comprestrict}(1),
  they generate the same component $C$ in $\Gamma_{|C\cup\{0\}}$ as in
  $\Gamma$. Thus $\Gamma_{|C\cup\{0\}}$ is a core.
\end{proof}

The statement of the classification theorem for \ccsp\ is actually
simpler than for \ocsp: we prove hardness if some core is not weakly
separable.
\begin{theorem}\label{th:multimain}
Let $\Gamma$ be a finite cc0-language. If there is a set
$D'$ with $0\in D'\subseteq \dom(\Gamma)$ such that $\Gamma_{|D'}$ is a core and not weakly
separable, then \ccspg\ is \textsc{Biclique}-hard, and fixed-parameter
tractable otherwise. 
\end{theorem}

Note that our proof shows \textup{W[1]}-hardness in most of
the cases: there is only one specific situation in the proof
(Lemma~\ref{lem:biclique}) where only \textsc{Biclique}-hardness is
shown. One can extract from the proof the following sufficient condition for proving
W[1]-hardness:
\begin{corollary}
  Let $\Gamma$ be a core that is not weakly separable and minimal in
  the sense that there is no subset $0\in D'\subset \dom(\Gamma)$ such
  that $\Gamma_{|D'}$ is a core and not weakly separable.  If $\Gamma$
  contains at least one semiregular or regular value, then \ccspg\ is
  \textup{W[1]}-hard.
\end{corollary}

\subsection{The algorithm}
We present an algorithm solving the FPT cases of the problem. The
algorithm consists of three steps. First, as a preprocessing step, we
use Lemma~\ref{lem:frequent} to ensure that every value is
``frequent.'' Next we solve the problem restricted to the core, which
is weakly separable by our assumption and hence the algorithm of
Lemma~\ref{lem:finddisjoint} can be used. Finally, we show that by a
postprocessing step, we can extend the solution on the core to the
original domain. For this last step, we need the following lemma.  For
a set of variables let $\delta_{v,d}$ be the assignment that assigns
value $d$ to variable $v$ and 0 to every other variable.  If 
$\Gamma$ is weakly separable, then satisfying assignments of this form can be
freely combined together (as there is no union counterexample). The
following lemma shows something similar under the weaker assumption
that $\Gamma_{|D'}$ is weakly separable whenever it is a core.
\begin{lemma}\label{lem:ubiquitous}
  Suppose that for every $D'$ with $0\in D'\subseteq \dom(\Gamma)$, if
  $\Gamma_{|D'}$ is a core, then it is weakly separable. Let $I$ be an
  instance of $\ccsp(\Gamma)$ having the following property: for every
  nonzero $d\in \dom(\Gamma)$, there are at least $k|\dom(\Gamma)|$
  variables $v$ such that $\delta_{v,d}$ is a satisfying assignment.
  Then $I$ has a solution satisfying the cardinality constraints and
  such a solution can be found in polynomial time.
\end{lemma}
\begin{proof}
  We prove the statement by induction on $|\dom(\Gamma)|$; for
  $\dom(\Gamma)=\{0\}$, we have nothing to show.  Let $K$ be the core
  of $\Gamma$ and $\pi$ the cardinality constraint in $I$. By
  Lemma~\ref{lem:corecore}, $\Gamma_{|K\cup\{0\}}$ is a core, hence
  weakly separable by assumption. For every $d\in K$, let $V_d$ be the
  set of those variables $v$ for which $\delta_{v,d}$ is a satisfying
  assignment. Since $|V_d|\ge k|\dom(\Gamma)|$, with greedy selection
  we can find a $V'_d\subseteq V_d$ of size exactly $\pi(d)$ for every
  $d\in K$ such that these sets are pairwise disjoint. Consider the
  assignment $f$ that assigns, for every $d\in K$, value $d$ to every
  variable of $V'_d$ and 0 to every variable that is not in
  $S:=\bigcup_{d\in K}V'_d$. Since $f$ can be obtained as the disjoint
  union of assignments $\delta_{v,d}$ with $v\in V'_d$ and
  $\Gamma_{|K\cup\{0\}}$ is weakly separable, we have
  that $f$ is a satisfying assignment. Let $I'=(V',\C',\pi')$ be the
  0-valid instance obtained by substituting the nonzero values of $f$
  as constants. Note that $\pi'(d)=0$ for every $d\in K$, since $f$
  assigns value $d$ to exactly $\pi(d)$ variables.  It is clear that
  if $I'$ has a solution, then $I$ has a solution.

  For any $v\in V'$ and $d$, let $\delta'_{v,d}$ be the assignment of
  $I'$ that assigns $d$ to variable $v$ and 0 to every other variable.
  By definition, every value in $\dom(\Gamma)\setminus K$ is
  degenerate in $\Gamma$.  Thus by Proposition~\ref{prop:regproduce}, for every
  $c\in \dom(\Gamma)\setminus K$, there is a $d\in K$ such that $d$
  produces $c$ in $\Gamma$. We claim that $\delta'_{v,c}$ is a
  satisfying assignment of $I'$ for any variable $V_d\setminus S$.
  Using the weak separability of $\Gamma_{|K\cup\{0\}}$, we get that
  $f+\delta_{v,d}$ is a satisfying assignment of $I$ for any $v\in
  V_d\setminus S$. Thus $\delta'_{v,d}$ on $V'$ is a satisfying
  assignment of $I'$, and, using that fact that $d$ produces $c$ in
  $\Gamma$, we get that $\delta'_{v,c}$ is a satisfying assignment of
  $I'$.  As $|S|\le k$, there are at least $|V_d|-k\ge |\dom(\Gamma)|k-k=
  (|\dom(\Gamma)|-1)k\ge |\dom(\Gamma)\setminus K|\cdot k$ variables
  $v$ such that $\delta'_{v,c}$ is a satisfying
  assignment of $I'$. Since $\pi'(d)=0$ for every $d\in K$, instance
  $I'$ can be viewed as an instance of
  $\ccsp(\Gamma_{|\dom(\Gamma)\setminus K})$.  Thus we can apply the
  induction hypothesis to conclude that $I'$ has a solution.
\end{proof}

\begin{lemma}\label{lem:ccsp-fpt}
  Suppose that for every $D'$ with $0\in D'\subseteq \dom(\Gamma)$, if
  $\Gamma_{|D'}$ is a core, then it is weakly separable. Then \ccspg\ is fixed-parameter tractable.
\end{lemma}
\begin{proof}
Let $I=(V,\C,k,\pi)$ be an instance of \ccspg.
Set
$F:=k^2(|\dom(\Gamma)|+d_\Gamma(k))$, where $d_\Gamma(k)$ is the function from
Lemma~\ref{lem:minsolution}.  Let us use the algorithm of
Lemma~\ref{lem:frequent} to obtain instances $I_1$, $\dots$,
$I_\ell$ such that $I_i$ is an $F$-frequent instance of
$\ccsp(\Gamma_{|D_i})$ for some set $D_i\subseteq \dom(\Gamma)$.

Consider an instance $I_i=(V_i,\C_i,k_i,\pi_i)$. Let $K_i$ be the core
of $\Gamma_{|D_i}$. Let $I'_i=(V'_i,\C'_i,k'_i,\pi'_i)$ be the
instance restricted to $K_i\cup\{0\}$, that is, every constraint
$\ang{\bs,R}\in\C_i$ is replaced by $\ang{\bs,R_{|K_i\cup\{0\}}}$, and
$\pi'_i(d)=\pi_i(d)$ for $d\in K_i$ and $\pi'_i(d)=0$ otherwise.  Note
that the retraction $\pr_{K_i}$ ensures that $I'_i$ is $F$-frequent as
well (by definition $K_i$, is a component).  We show that $I_i$ has a
solution if and only if $I'_i$ has.  As
$\ccsp(\Gamma_{|K_i\cup\{0\}})$ is weakly separable by assumption, the
algorithm of Theorem~\ref{lem:finddisjoint} can be used to check in
fpt-time whether $I'_i$ has a solution.

The retraction $\pr_{K_i}$ shows that if $I_i$ has a solution $f$,
then $I'_i$ has a solution $f'=pr_{K_i} f$. 
For the other direction, let $f'$ be a
solution of $I'_i$ and let $I''_i$ be the instance of
$\ccsp(\Gamma_{|D_i})$ obtained from $I_i$ by substituting the nonzero
values of $f'$ as constants. Since $f'$ satisfies $\pi'_i$,
the cardinality constraint is 0 in instance $I''_i$ for every $d\in
K_i$. Thus $I''_i$ can be viewed as an $\ccsp(\Gamma_{|D_i\setminus
  K_i})$ instance.  We show that the conditions of
Lemma~\ref{lem:ubiquitous} hold for $I''_i$ (viewed as an
$\ccsp(\Gamma_{|D_i\setminus K_i})$ instance), hence it has a solution
$f''$. This means that solution $f'$ of $I'_i$ can be extended by
$f''$ to obtain a solution $f$ of $I_i$.

Let $c\in D_i\setminus K_i$. By Proposition~\ref{prop:regproduce},
there is a
$d\in K_i$ producing $c$ in $\Gamma_{|D_i}$. As $I_i$ is
$F$-frequent, $I_i$ has distinct variables $v_1$, $\dots$, $v_F$ and
(not necessarily distinct) satisfying assignments $g_1$, $\dots$,
$g_F$ of size at most $k$ such that $g_j(v_j)=d$. We can assume that
each $g_j$ is contained in $K_i\cup\{0\}$ (as $d\in K_i$ and $\pr_{K_i}$ is an 
endomorphism of $\Gamma_{|D_i}$). Since each $g_j$ has size at most
$k$, there are at least $F/k$ distinct assignments in the sequence
$g_1$, $\dots$, $g_F$. By Lemma~\ref{lem:disjoint-decomp}(2), we can assume
that every $g_j$ is a minimal assignment. By
Lemma~\ref{lem:minsolution}, each nonzero variable of $f'$ is nonzero
in at most $d_\Gamma(k)$ minimal assignments of size at most
$k$. Hence, among the $F/k$ distinct minimal assignments,
there are at most $k\cdot d_\Gamma(k)$ assignments nondisjoint with
$f'$, that is, there are at least $F/k-k\cdot d_\Gamma(k)\ge |\dom(\Gamma)|k$
assignments disjoint with $f'$. Let us consider such an assignment
$g_j$. As $g_j$ and $f'$ are disjoint, both use only values from $K_i$, and
$\Gamma_{|K_i\cup\{0\}}$ is weakly separable, their sum is a satisfying
assignment. This means that $I''_i$ has a satisfying assignment where
$v_j$ has value $d$. Using the fact that $d$ produces $c$, it follows
that $\delta_{v_j,c}$ is a satisfying assignment of $I''_i$. Thus for
every $c\in D_i\setminus K_i$, there are at least $|\dom(\Gamma)|k$ variables $v$
such that $\delta_{v,c}$ is a satisfying assignment of $I''_i$. By
Lemma~\ref{lem:ubiquitous}, this means that $I''_i$ has a solution.
\end{proof}

\subsection{Hardness}
A crucial difference between \ocspg\ and \ccspg\ is that for every
$0\in D'\sse\dom(\Gamma)$, it is trivial to reduce
$\ccsp(\Gamma_{|D'})$ to \ccspg. Indeed, a $\ccsp(\Gamma_{|D'})$
instance can be interpreted as a \ccspg\ instance with $\pi(d)=0$ for
every $d\in \dom(\Gamma)\setminus D'$.

\begin{proposition}\label{prop:restriction}
If $\ccsp(\Gamma_{|D'})$ is \textup{W[1]}-hard  for some 
$0\in D'\sse\dom(\Gamma)$, then \ccspg\ is \textup{W[1]}-hard.
\end{proposition}
In particular, if $\Gamma_{|\{0,a\}}$ is not weakly separable for some value $a\in D$,
then the result of \cite{Marx05:parametrized} on the Boolean case
implies that $\ccsp(\Gamma_{|\{0,a\}})$ and hence \ccspg\ are W[1]-hard
(see also Example~\ref{exa:weaklysep}).

Proposition~\ref{prop:restriction} allows us to assume that the
language $\Gamma$ satisfies the hardness condition of
Theorem~\ref{th:multimain}, but no restriction $\Gamma_{|D'}$
satisfies it for any $0\in D'\subset \dom(\Gamma)$. That is, $\Gamma$
is a core and not weakly separable, but every core $\Gamma_{|D'}$ with
$0\in D'\subset \dom(\Gamma)$ is weakly separable.   Indeed, if $0\in D'\subset \dom(\Gamma)$ is a set
such that $\Gamma_{|D'}$ is a core and not weakly separable, then it
is sufficient to prove hardness for the constraint language
$\Gamma_{|D'}$ and the hardness for $\Gamma$ follows by
Prop.~\ref{prop:restriction}.

We proceed in the following way. Lemma~\ref{lem:semihardmulti} of Section~\ref{sec:semiregular-values}
proves W[1]-hardness in the case when there is a semiregular value in
$\Gamma$.
Section~\ref{sec:self-prod-valu} considers the case when every element
is degenerate or self-producing.  The main part of the proof appears
in Section~\ref{sec:regular-values}, where we prove W[1]-hardness
using a counterexample involving regular values; as in
Section~\ref{sec:size}, the reader is encouraged to focus on this part
of the proof.  The proof of a technical claim is deferred to
Section~\ref{sec:claimproof}.

\subsubsection{Semiregular values}
\label{sec:semiregular-values}
In the case when there is a semiregular value, we can identify a
difference counterexample and use it to simulate the constraints in an
\textsc{Implications} instance. We say that a multivalued morphism $\phi$ 
{\em witnesses} that $d$ is semiregular if $0,d\in \phi(c)$ for some $c\in\dom(\Gamma)$.
\begin{lemma}\label{lem:semidiff}
  If $\Gamma$ contains a semiregular value, then there is a difference
  counterexample. Moreover, if $\phi$ witnesses that $d$ is
  semiregular, then there is a difference counterexample in
  $\Gamma_{|\phi(\dom(\Gamma))}$.
\end{lemma}
\begin{proof}
  Suppose that $0,d\in \phi(c)$.  As $d$ is semiregular, no value
  produces $d$. In particular, $c$ does not produce $d$, thus there is
  a relation $R\in \Gamma$, a tuple $\bt\in R$, and a nonzero tuple
  $\bt_d\not\in R$ such that $d$ is the only nonzero value appearing
  in $\bt_d$ and at every coordinate where $d$ appears in $\bt_d$,
  value $c$ appears in the same coordinate of $\bt$.  Applying $\phi$
  on $\bt$ and turning each $c$ into 0 yields a tuple $\bt'\in R$
  disjoint from $\bt_d$. Applying $\phi$ on $\bt$ also shows that
  $\bt'+\bt_d\in R$: instead of turning each $c$ into $0$, we can turn
  it to either 0 or $d$ (depending on the tuple $\bt_d$). Now
  $(R,\bt_d,\bt')$ is a difference counterexample in
  $\Gamma_{|\phi(\dom(\Gamma))}$.
\end{proof}

\begin{lemma}\label{lem:semihardmulti}
Let $\Gamma$ be a core. If there is a semiregular value $d$ in
$\Gamma$, then \ccspg\ is \textup{W[1]}-hard. 
\end{lemma}
\begin{proof}
  Let $\psi:\dom(\Gamma)\to 2^{\dom(\Gamma)}$ be a multivalued morphism witnessing that
  $d$ is semiregular. Let us choose $d$ and $\psi$ such that
\begin{enumerate}
\item the size
  of $\{a\in \dom(\Gamma)\mid \psi(a)\neq \{0\}\}$ is minimum possible, and
\item among
  such $d$ and $\psi$, the size of $S:=\psi(\dom(\Gamma))$ is minimum possible.
\end{enumerate}
Observe that we can assume that $\psi(c)=\{0,d\}$ for a unique value
$c\in \dom(\Gamma)$ and $|\psi(a)|=1$ for every $a\neq c$.
Furthermore, we can assume that $d$ cannot be produced by any $a\in S$
in $\Gamma_{|S}$. Otherwise, if $\psi_d:S\to 2^{S}$ is the multivalued
morphism witnessing that $a\in S$ produces $d$ in $\Gamma_{|S}$, then $\psi
\circ \psi_d$ witnesses that $d$ is produced by some $a'\in
\dom(\Gamma)$ in $\Gamma$, hence $d$ is not semiregular in $\Gamma$.
Let $S_1\subseteq S\setminus \{d\}$ contain the regular and
semiregular values in $\Gamma_{|S}$ other than $d$ and let
$S_2\subseteq S$ contain the self-producing and degenerate values
(thus $S=S_1\cup S_2\cup\{d,0\}$). Note that by
Lemma~\ref{lem:semidiff}, there is a difference counterexample
$(R,\bt_1,\bt_2)$ in $\Gamma_{|S}$.

We show that $\ccsp(\Gamma_{|S})$ is W[1]-hard, hence (by
Proposition~\ref{prop:restriction}) \ccspg\ is W[1]-hard as well.
The proof is by reduction from \textsc{Implications}. For each vertex
$v_i$ of $G$, we introduce a gadget $\mvm(\Gamma_{|S},S)$ denoted by
$G_{i}$. The size 
of bag $B_c$ of each gadget is $Z:=2t|S|$ if $c\in S_1$ and it is 1 if
$c\in S_2\cup \{d\}$.  We set the cardinality constraint
$\pi'(c)=tZ$ for every $c\in S_1$ 
and $\pi'(c)=t$ for every $c\in S_2\cup \{d\}$, i.e., the parameter
$k$ equals $tZ|S_1|+t|S_2|+t$.
To finish the construction of the instance, we encode the directed
edges of the \textsc{Implications} instance by adding the constraint
$\imp(G_{x},G_{y})$ for each directed edge
$\overrightarrow{v_xv_y}$ of $G$.

Suppose that there is a solution $C$ of size exactly $t$ for the
\textsc{Implications} instance. If vertex $v_i$ is in $C$, then set
the standard assignment on gadget $G_{i}$. It is clear that this
results in an assignment of size exactly $tZ|S_1|+t|S_2|+t$ and the
constraints of the $\mvm(\Gamma_{|S},S)$ gadgets are satisfied. From the
fact that there is no directed edge $\overrightarrow{v_xv_y}$ with
$v_x\in C$, $v_y\not\in C$ and from Lemma~\ref{lem:impconstraint}(1),
it follows that the constraints of $\imp(G_{x},G_{y})$ are
also satisfied.

For the other direction, we have to show that if there is a solution
of the \ccspg\ instance satisfying the cardinality constraint, then there is a
solution $C$ of size exactly $t$ for \textsc{Implications}. 

We show first that if a value $c\in S_1\cup \{d\}$ appears in bag
$B_d$ of $G_x$, then every variable of $G_{x}$ is nonzero. Suppose
that $c$ appears on the variable in bag $B_d$, but 0 appears in some
variable of bag $B_{d'}$ of $G_x$ for some $d'\in S$. As bag $B_d$
contains only a single variable, we have $d'\neq d$. Let $g$ be an
endomorphism of $\Gamma_{|S}$ given by $G_{x}$ such that $g(d)=c$ and
$g(d')=0$. Now the multivalued morphism $\psi \circ g$ witnesses that
$c=g(d)$ is semiregular in $\Gamma$: value $c\in S_1\cup \{d\}$ is not
self-producing or degenerate, hence $\psi\circ g$ cannot show that $c$
is produced by some value. Moreover, $(\psi\circ
g)(\dom(\Gamma))=g(S)$ has size strictly less than $|S|$,
contradicting either the first or second minimality condition of
$\psi$. It follows that there can be at most $t$ gadgets where the bag
corresponding to $d$ contains a value from $S_1\cup \{d\}$: otherwise
there would be more than $t(Z|S_1|+|S_2|+1)$ nonzero variables.

Next we show that if a value $c\in S_1\cup \{d\}$ appears in a bag
$B_{c'}$ of a gadget $G_{x}$ for some $c'\in S_1$, then every variable
of that bag is nonzero. Otherwise, let $\psi'$ be a multivalued
morphism of $\Gamma_{|S}$ given by $G_{x}$ with $\psi'(c')=\{0,c\}$
and $|\psi'(b)|=1$ for every $b\neq c'$. Let $h$ be the endomorphism
of $\Gamma$ defined such that $h(a)=a'$ if $\psi(a)=\{a'\}$ and
$h(a)=0$ if $\psi(a)=\{0,d\}$. Now $h\circ \psi'$ cannot witness that
$c$ is produced by some value (as $c\in S_1\cup \{d\}$), hence it witnesses
that $c$ is semiregular in $\Gamma$ and
$|\{a\in\dom(\Gamma)\mid \psi'(a)\ne\{0\}\}|<|\{a\in\dom(\Gamma)\mid \psi(a)\ne\{0\}\}|$,
contradicting the minimality of $d$ and $\psi$.

Consider those bags that contain values from $S_1\cup \{d\}$.  Each
such bag represents a value in $S_1\cup\{d\}$: if $c\in S_1\cup\{d\}$
appeared in a bag representing a value from $S_2$, then by
Proposition~\ref{prop:homomtype}, $c$ would be self-producing or
degenerate in $\Gamma_{|S}$. We have seen that at most $t$ bags
representing $d$ can contain values from $S_1\cup \{d\}$. The total cardinality constraint of these values is $t|S_1|Z+t$. Thus at least $t|S_1|$ bags
representing $S_1$ contain values from $S_1\cup \{d\}$. Moreover,
there are exactly $t|S_1|$ such bags: as shown in the previous
paragraph, these bags are fully nonzero, thus $t|S_1|+1$ such bags would
mean that the size of the assignment is at least $(t|S_1|+1)Z>
tZ|S_1|+t|S_2|+t$. As $\sum_{c\in S_1\cup \{d\}}\pi'(c)$ is
exactly $tZ|S_1|+t$, it follows that there are exactly $t$ gadgets
where the variable in bag $B_d$ has a nonzero value from $S_1\cup
\{d\}$. We have already observed that the variables of these $t$ gadgets
are fully nonzero, and the cardinality constraint $\pi'$ imply that every
variable of every other gadget is zero.

Let us construct the set $C$ such that $v_x\in C$ if and only if the
variables of $G_{x}$ are nonzero; the previous paragraph implies that
$|C|=t$.  We claim that $C$ is a solution for the
\textsc{Implications} instance.  Suppose that there is an edge
$\overrightarrow{v_xv_y}$ with $v_x\in C$ and $v_y\not\in C$.  If
$v_x\in C$ and $h$ is an endomorphism of $\Gamma_{|S}$ given by
$G_{x}$, then $h$ has to be a permutation of $S$: otherwise, $\psi
\circ h$ witnesses that $h(d)$ is semiregular and $|(\psi \circ
h)(\dom(\Gamma))|$ is strictly less than $|S|$, contradicting the
choice of $d$ and $\psi$ (note that $h$ maps every nonzero value to a
nonzero value, thus $\psi \circ h$ cannot witness that $h(d)$ is
produced by some element).  We have seen that there is a difference
counterexample $(R,\bt_1,\bt_2)$ in $\Gamma_{|S}$. Since $h$ is a
permutation, $h^s(\bt_1)=\bt_1$ for some $s\ge 1$, i.e., $h$ is
$\bt_1$-recoverable. Thus if $G_{y}$ is fully zero, then
$\imp(G_{x},G_{y})$ is not satisfied by
Lemma~\ref{lem:impconstraint}(3).
\end{proof}

\subsubsection{Self-producing values}
\label{sec:self-prod-valu}
In this section, we consider the case when every element is either
self-producing or degenerate. By Proposition~\ref{prop:regproduce},
there is at least one self-producing element. It is not hard to see
that the component generated by self-producing elements contains only
self-producing elements. Indeed, the component generated by a
self-producing element $d\in D$ equals $\{d\}$, and by
Proposition~\ref{prop:compunion} the union of components is a
component.  Lemmas~\ref{lem:difference-multi} and~\ref{lem:biclique}
consider the two possibilities in this case: when there is a
difference counterexample, and when all counterexamples are union.

\begin{lemma}\label{lem:difference-multi}
Let $\Gamma$ be a core and let  $(R,\bt_1,\bt_2)$ be a difference
counterexample to weak separability 
satisfying the conditions of Lemma~\ref{lem:regularcounter}, and $a_1$
is self-producing. Then \ccspg\ is \textup{W[1]}-hard.
\end{lemma}

\begin{proof}
Since $a_1$ is self-producing, $\{a_1\}$ is a component, hence $a_1$ is
the only nonzero value appearing in  $\bt_1$ and $\bt_2$. This means that
$\Gamma_{|\{0,a_1\}}$ is not weakly separable, hence \ccspg\ is
\textup{W[1]}-hard.
\end{proof}

\begin{lemma}\label{lem:biclique}
Let $\Gamma$ be a core and let $(R,\bt_1,\bt_2)$ be a union
counterexample satisfying the conditions of 
Lemma~\ref{lem:regularcounter} such that $a_1$ and $a_2$ are
self-producing. Then \ccspg\ is \textsc{Biclique}-hard.
\end{lemma}

\begin{proof}

  We assume that $a_1$ and $a_2$ are weakly separable, otherwise we
  are done.  As $a_i$ is self-producing, $\{a_i\}$ is the component 
generated by $a_i$, hence $\bt_i$ is contained in $\{0,a_i\}$.

  First, we show that there is no inner homomorphism $h_{12}$ from
  $\{0,a_1\}$ with $h_{12}(a_1)=a_2$ and there is no inner homomorphism
  $h_{21}$ from $\{0,a_2\}$ with $h_{21}(a_2)=a_1$. Note that $a_1$
  produces itself, thus the existence of $h_{12}$ would mean that
  $a_1$ produces $a_2$. Since $a_2$ is self-producing, this would
  imply that $a_2$ produces $a_1$, and hence $h_{21}$ exists as well. A
  symmetrical argument shows that the existence of $h_{21}$ implies
  the existence of $h_{12}$. Suppose that both homomorphisms exist. In
  this case, $\bt_1+h_{21}(\bt_2)\in R$ follows from
  $\bt_1,h_{21}(\bt_2)\in R$ and from the fact that $a_1$ is weakly
  separable. By using $\bt_1\in R$ and Lemma~\ref{lem:substitute}, we
  get $\bt_1+h_{12}(h_{21}(\bt_2))=\bt_1+\bt_2\in R$, a contradiction.

We reduce \textsc{Biclique} (see Section~\ref{sec:reductions}) to
  $\ccsp(\Gamma_{|\{0,a_1,a_2\}})$.
  Consider the gadget $\mvm(\Gamma_{|\{0,a_1,a_2\}},\{0,a_1\})$ where the bag
  $B_{a_1}$ contains only one variable. Setting this variable to 0 or
  $a_1$ is a satisfying assignment of the gadget. However, there is no
  inner homomorphism $h$ from $\{0,a_1\}$ to $\{0,a_1,a_2\}$
  with $h(a_1)=a_2$, thus the variable cannot have value $a_2$. Thus
  the unary relation $U_1=\{0,a_1\}$ is intersection definable in
  $\Gamma_{|\{0,a_1,a_2\}}$ and the same holds for the unary relation
  $U_2=\{0,a_2\}$.

  Since $(R,\bt_1,\bt_2)$ is a union counterexample, we first obtain a
  relation $R'$ from $R$ by substituting constant 0 to positions in
  which both $\bt_1$ and $\bt_2$, and then identifying variables to
  intersection define a binary relation $R''$ such that $(0,0)$,
  $(a_1,0)$, $(0,a_2)\in R''$, but $(a_1,a_2)\not\in R''$. Let us
  consider the binary relation $R'''$ represented by the CSP instance
  $(\{x,y\},\C')$ where
  $\C'=\{\ang{(x,y),R''},\ang{(x),U_1},\ang{(y),U_2}\}$. Clearly, this
  relation is intersection definable in $\Gamma_{|\{0,a_1,a_2\}}$.  It
  is easy to see that $(0,0),(a_1,0),(0,a_2)\in R'''$ and $R'''$
  contains no other tuple.  Thus as observed in
  Example~\ref{exa:graphs}, $\ccsp(R''')$ is equivalent to
  \textsc{Biclique}.
\end{proof}

\subsubsection{Regular values}
\label{sec:regular-values}
A significant difference between the hardness proofs of \ocspg\ and
\ccspg\ is that it can be assumed in the case of \ocspg\ that no
proper contraction exists and this assumption can be used to show that
certain endomorphisms have to be permutations. In
Section~\ref{sec:ocsphardness}, we used such arguments to show that
gadgets are $\bt$-recoverable. For \ccspg, we cannot
make this assumption, thus the proof is based on a delicate argument
(Claim~\ref{claim:recoverable}), making use of the cardinality
constraint, to achieve a similar effect. The following lemma is not
used directly in the proof, but it demonstrates how we can deduce in
some cases that a multivalued morphism gadget essentially behaves as
if it had the standard assignment. Recall that for $D=\{0,1,\dots,\Delta\}$, we defined in Section~\ref{sec:gadgets} the constants
$$
Z_{i,d}^{t,D}=(4t\Delta)^{2t\Delta+(i\Delta+d)}+(4t\Delta)^{5t\Delta-(i\Delta+d)}.
$$
\begin{lemma}\label{lem:ccsp-demo}
  Let $\Gamma$ be a finite constraint language over
  $D=\{0,1,\dots,\Delta\}$. Consider an instance consisting of a
  single $\mvm(\Gamma,D)$ gadget where bag $B_b$ has size
  $Z^{1,D}_{1,b}$. Let the cardinality constraint be
  $\pi(b)=Z^{1,D}_{1,b}$. If $\phi$ is the maximal multivalued
  morphism given by the gadget in a solution, then there is a $p\ge 1$
  such that $b\in \phi^{p'}(b)$ for every $p'\ge p$ and nonzero $b\in D$.  In
  particular, the gadget is $\bt$-recoverable for any tuple $\bt$.
\end{lemma}
\begin{proof}
  We prove the statement by induction on $b$. Suppose that for every
  $a<b$, there is a $p_{a}$ such that $a\in \phi^{p'}(a)$ for every
  $p'\ge p_{a}$ (this statement is vacuously true if $b$ is the
  smallest nonzero value). Let $\phi_b=\pr_{\{0,1,\dots,b\}}\circ
  \phi$. Let $T=\bigcup_{p\ge 1}\phi^p_b(b)$, that is, those values
  that can be reached from $b$ by repeated applications of $\phi_b$
  (note that $T$ can contain values larger than $b$, but because of
  $\pr_{\{0,1,\dots,b\}}$ in $\phi_b$, such values can appear only during
  the last application of $\phi_b$). As the total cardinality
  constraint is exactly the number of variables, all the variables are
  nonzero. The bag $B_b$ and the bags $B_{a}$ for $a<b$ and $a\in T$
  contain nonzero values only from $T$ by definition. The total size
  of these bags is
\begin{equation}
\sum_{a\in T, a<b}Z^{1,D}_{1,a}+Z^{1,D}_{1,b}.\label{eq:ccsp-demo1}
\end{equation}
We claim that $b\in T$. Otherwise, the total cardinality constraint of the values in $T$ is
\begin{equation}
\sum_{a\in T, a<b}Z^{1,D}_{1,a}+\sum_{a\in T, a>b}Z^{1,D}_{1,a}.\label{eq:ccsp-demo2}
\end{equation}
As $Z^{1,D}_{1,b}>|D|Z^{1,D}_{1,a}$ for every $a>b$, the second term
in \eqref{eq:ccsp-demo1} is strictly larger than the second term of
\eqref{eq:ccsp-demo2}, a contradiction. Thus $b\in T$, and therefore
$b\in \phi^s_b(b)$ for some $s\ge 1$. Consider the smallest such $s$.
If $s=1$, then $b\in \phi^p_b(b)$ for every $p\ge 1$. Otherwise, as
$b\in\phi_b^{s-1}(\phi_b(b))$, there is some $a<b$, $a\in T$ such that
$a\in \phi_b(b)$ and $b\in \phi^{s-1}_b(a)$. By the induction
assumption, $a\in\phi^{p'}(a)$ for every $p'\ge p_{a}$. This means
that $b\in \phi^{1+p'+s-1}(b)$ for every $p'\ge p_{a}$, or in other
words, $b\in \phi^{p'}(a)$ for every $p'\ge p_{a}+s$. Thus
$p_b:=p_{a}+s$ proves the induction statement.

To see that the gadget is $\bt$-recoverable, observe that $b\in
\phi^{p+1}(b)$ implies that there is a $c_b\in \phi(b)$ such that
$b\in \phi^p(c_b)$. Let $h$ be the endomorphism that maps each $b\in
D$ to such a $c_b$ (note that $h$ is an endomorphism, as it is a
subset of $\phi$). Then $h$ is $t$-recoverable, as witnessed by
$\phi^p$.
\end{proof}

Thus in a sense we can assume that a gadget has the standard
assignment, even if the values appearing in the bags are
arbitrary. However, the situation is more complicated in an actual
hardness proof, where there are several gadgets and moreover value 0 can also
appear in some of the bags.  The following lemma contains the most
generic part of the hardness proof of Theorem~\ref{th:multimain}: we are proving hardness using a
counterexample to weak separability.

\begin{lemma}\label{lem:multihardcore}
  Let $\Gamma$ be a core that is not weakly separable, $\Gamma_{|D'}$ is
  weakly separable for every core $\Gamma_{|D'}$ 
with $0\in D'\subset \dom(\Gamma)$, there are no semiregular values in $\Gamma$, and there is a regular value in
$\Gamma$. Then 
  \ccspg\ is \textup{\textup{W[1]}}-hard.
\end{lemma}

\begin{proof}
  By Lemma~\ref{lem:regularcounter}, there is a counterexample with
  values contained in the component $C_1$ generated by some value
  $a_1\in \dom(\Gamma)$, or with values in $C_1\cup C_2$, where $C_1$
  (resp., $C_2$) is the component generated by some value $a_1$
  (resp., $a_2$). If $a_1$ is degenerate, then (as $\Gamma$ is a core)
  $a_1$ is in the component generated by the nondegenerate values.
  Thus by Proposition~\ref{prop:setitemgen}, there is a nondegenerate
  $a'_1$ such that $a_1$ is in the component $C'_1$ generated by
  $a'_1$. Since the intersection of components is also a component, we
  have $C_1\subseteq C'_1$. Thus we can assume that $a_1$ and $a_2$
  are nondegenerate. If $a_1$ and $a_2$ are both self-producing, then
  $C_1=\{a_1\}$, $C_2=\{a_2\}$, and hence $\Gamma_{|\{0,a_1,a_2\}}$ is
  not weakly separable. By Prop.~\ref{prop:compunion}, $C_1\cup C_2$
  is also a component. Therefore, as $a_1$ and $a_2$ are nondegenerate
  in $\Gamma_{|\{0,a_1,a_2\}}$ by Lemma~\ref{prop:comprestrict}(2), we
  have that $\Gamma_{|\{0,a_1,a_2\}}$ is a core. Thus
  $\dom(\Gamma)=\{0,a_1,a_2\}$ by the minimality of $\Gamma$, implying
  that there is no regular value in $\dom(\Gamma)$, a
  contradiction. Similarly, if the counterexample is contained in the
  component $C_1=\{a_1\}$ generated by the self-producing value $a_1$,
  then $\dom(\Gamma)=\{0,a_1\}$ and again there is no regular value in
  $\Gamma$.  By assumption, there are no semiregular values.
This means that there are
  only three cases to consider: we have a counterexample
  $(R,\bt_1,\bt_2)$ that is
\begin{enumerate}
\item a union or difference counterexample such that $\bt_1+\bt_2$ is contained in the
  component of a regular value $a_1$;
\item a union counterexample such that $\bt_1$ (resp., $\bt_2$) is contained in
  the component of some regular value $a_1$ (resp., regular value $a_2$);
\item a union counterexample such that $\bt_1$ (resp., $\bt_2$) is contained in
  the component of some regular value $a_1$ (resp., some
  self-producing value $a_2$);
\end{enumerate}
We present a unified \textup{W[1]}-hardness proof for the three cases.
The reduction is from \textsc{Multicolored Independent Set} in the
case of a union counterexample, while we are reducing from
\textsc{Multicolored Implications} in the case of a difference
counterexample. Let $v_{x,y}$ ($1 \le x \le t$, $1 \le y \le n$) be
the vertices of the graph in the instance we are reducing from. Let
$D:=\dom(\Gamma)$; we assume that $D=\{0,\dots,\Delta\}$.  For each
vertex $v_{x,y}$, we introduce an $\mvm(\Gamma,D)$ gadget
$G_{x,y}$. Without loss of generality, we can assume that the nonzero
values in $D$ are ordered such that the regular values precede all the
nonregular values. 
The bag $B_d$ of
$G_{x,y}$ has size
\[
z_{x,d}=
\begin{cases}
Z^{t,D}_{x,d}&\text{if $d$ is regular,}\\
(4t\Delta)^{2t\Delta-(x\Delta+d)}&\text{otherwise,}
\end{cases}
\]
where $Z^{t,D}_{x,d}$ is as defined in Section~\ref{sec:gadgets}.
Observe that the size of any bag representing a regular value is more
than $4t\Delta$ times larger than the size of any bag representing a
nonregular value.  The cardinality constraint $\pi(d)$ is set to be
$\sum_{x=1}^{t}z_{x,d}$. Observe that the cardinality constraint of
any regular value is larger than the total cardinality constraint of
all the nonregular values.

If the reduction is from \textsc{Multicolored Independent Set}, then we
introduce a $\nand(G_{x,y},G_{x',y'})$ constraint for each
edge $v_{x,y}v_{x',y'}$ of the graph. Furthermore, for every $1\le
x\le t$ and $1 \le y ,y' \le n$,  $y\neq y'$, we introduce a constraint
$\nand(G_{x,y},G_{x,y'})$. If the reduction is from
\textsc{Multicolored Implications}, then we introduce a
$\imp(G_{x,y},G_{x',y'})$ constraint for each edge
$\overrightarrow{v_{x,y}v_{x',y'}}$ of the graph. This completes the
description of the reduction. 

It is easy to see that the reduction works in one direction. Let $S$
be a set of vertices that form a solution for the 
instance we are reducing from. If $v_{x,y}\in S$, then let us give the
standard assignment to the gadget $G_{x,y}$. The fact that $S$
contains exactly one vertex of each color implies that the resulting
assignment satisfies the 
cardinality constraints. Furthermore, if $S$ is an independent set,
then by Lemma~\ref{lem:nandconstraint}(1), all the
$\nand(G_{x,y},G_{x',y'})$ constraints are satisfied; if $S$ is
solution of \textsc{Multicolored Implications}, then by
Lemma~\ref{lem:impconstraint}(1), all the 
$\imp(G_{x,y},G_{x',y'})$ constraints are satisfied.

For the other direction of the proof, we have to show that multivalued
morphisms given by the gadgets have certain properties that allow us
to invoke Lemma~\ref{lem:nandconstraint}(3) or
Lemma~\ref{lem:impconstraint}(3). For this purpose, we show that for $i=1,2$ and for
every $x$, $1\le x\le t$, there is an $y^i_x$ such that values from the component
generated by $a_i$ appear only on gadget $G_{x,y^i_x}$ and do not
appear on $G_{x,y}$ for any $y\neq y^i_x$. Furthermore, we have to
show that $G_{x,y^i_x}$ is $\bt_i$-recoverable. The
proof of these claims are based on the properties of regular
values and the way the cardinality constraints are defined.

\begin{claim}\label{claim:regularbags}
Let $K$ be the component generated by a regular value $d$ and let $\Kr$ be the regular values in $K$.
\begin{enumerate}
\item
  If $b\in \Kr$ appears in bag $B_a$ of some gadget
  $G_{x,y}$, then every value appearing in the bag is from $\Kr$.
  \item For every $1\le x \le t$,
  there is a unique $1\le w_{x,d} \le n$ such that values from $\Kr$ appear in bag $B_d$ of $G_{x,w_{x,d}}$ (and by (1), every value in  bag $B_d$ of
  $G_{x,w_{x,d}}$ is nonzero and from $\Kr$).
\end{enumerate}
\end{claim}
\begin{proof}
 If value $b\in \Kr$
  appears in bag $B_a$, then $a$ has to be regular as well by
  Proposition~\ref{prop:homomtype}.
  Furthermore, every variable of $B_a$ has to be nonzero and has to belong to
  $K$: otherwise the gadget would give a multivalued morphism $\phi$
  such that $0,b\in (\phi\circ \pr_{K})(a)$, contradicting the
  assumption that $b$ is regular.

We show next that every value appearing in $B_a$ is actually from $\Kr$, i.e., regular.
  Let $\B$ be the multiset consisting of the sizes of the bags
  containing values from $\Kr$ (as observed, each such bag
  represents a regular value hence its size
  is from $\Z^{t,D}$) and let $A=\{Z^{t,D}_{x,a}\mid \text{$1\le x
    \le t$, $a\in \Kr$}\}$.
 Note that the sum of the numbers in $A$ is exactly
  the sum of cardinality constraint $\pi(a)$ for every regular value
   $a\in \Kr$. The sum of the numbers in $\B$ cannot be less than that,
  but it might be larger: the bags where values
  from $\Kr$ appear might contain nonregular values from $K$ as well (but
no value outside $K$). However, the sum of cardinality constraints
  $\pi(a)$ of the nonregular values $a$ is at most
  $\Delta(4t\Delta)^{2t\Delta-(\Delta+1)}<(4t\Delta)^{2t\Delta}$. Thus
  Lemma~\ref{lem:unique} can be used to conclude that $A=\B$ and it
  follows that the bags where values from $\Kr$ appear contain
  only values from $\Kr$.

The second statement is also an immediate consequence of $A=\B$:
 there is exactly one bag with size $Z^{t,D}_{x,d}$ where values from
 $\Kr$ appear.
\end{proof}

Let $\phi_{x,d}$ be the maximal multivalued morphism given by
$G_{x,w_{x,d}}$, where $w_{x,d}$ is defined by
Claim~\ref{claim:regularbags}(2). Our aim is to show that then
$\phi_{x,a_i}$ is $\bt_i$-recoverable. For this purpose, we prove the
following claim, which is the main technical ingredient of the
proof. The proof idea was demonstrated in Lemma~\ref{lem:ccsp-demo},
but here we need additional arguments to handle zero values appearing
on variables, values not in $K$ appearing on variables, and the fact
that there are gadgets having bags of different sizes. The proof of the following
claim is delicate and technical, hence we defer it to
Section~\ref{sec:claimproof} to maintain the flow of the proof.
\begin{claim}\label{claim:recoverable}
Let $K$ be the component generated by a regular value $d$ and suppose that $d$ is the smallest value in $K$.
For every $1\le x \le t$, the following are true:
\begin{enumerate}
\item $\phi_{x,d}$ is $\bt$-recoverable if $\bt$ contains nonzero
  values only from $K$.

\item For any $a\in K$, bag $B_a$ of $G_{x,w_{x,d}}$ is fully nonzero
  and contains values only from $K$.
\item For any $a\in K$ and any $y\neq w_{x,d}$, bag $B_a$ of $G_{x,y}$ does not contain values from $K$.
\end{enumerate}
\end{claim}

Assuming Claim~\ref{claim:recoverable}, we consider the following three cases for
the counterexample $(R,\bt_1,\bt_2)$.
\medskip

\textit{Case 1:} {\em $(R,\bt_1,\bt_2)$ is a union or difference counterexample such that $\bt_1+\bt_2$ is contained in the
  component of a regular value $a_1$.}

Let us observe first that the minimality of $\Gamma$ implies that $D\setminus \{0\}$
is equal to the component $C_1$ generated by $a_1$. Indeed, by
Lemma~\ref{prop:comprestrict}(1-2), $a_1$ is regular in
$\Gamma_{|C_1\cup\{0\}}$ and $a_1$ generates $C_1$ in
$\Gamma_{|C_1\cup\{0\}}$ (thus $\Gamma_{|C_1\cup \{0\}}$ is also a core).  Set
$y_x=w_{x,a_1}$ and let $S$ contain $v_{x,y_x}$ for every $1\le x\le
t$. We may assume that $a_1$ is the smallest nonzero value. Therefore,
Claim~\ref{claim:recoverable}(1) implies that $G_{x,y_x}$ is both
$\bt_1$- and $\bt_2$-recoverable, while
Claim~\ref{claim:recoverable}(3) implies that $G_{x,y}$ gives a
homomorphism $h$ with $h(\bt_2)=0$ if $y\neq y_x$ (as we have $K=D\setminus \{0\}$,
Claim~\ref{claim:recoverable}(3) implies that $G_{x,y}$ is fully
zero). Thus we can use Lemma~\ref{lem:nandconstraint}(3) or
Lemma~\ref{lem:impconstraint}(3) to show that $S$ is a solution for
the instance we are reducing from.  \medskip

{\em Case 2: $(R,\bt_1,\bt_2)$ is a union counterexample such that for $i=1,2$, tuple $\bt_i$ is contained in
  the component $C_i$ of some regular value $a_i$.}

Set $y^1_x=w_{x,a_1}$ and $y^2_x=w_{x,a_2}$. We may assume that $a_i$
is the smallest value in component $C_i$ (if both $a_1$ and $a_2$
appear in the same component $C_i$, then we can set $a_1=a_2$ and we
are in Case 1).  Claim~\ref{claim:recoverable}(1) implies that
$G_{x,y^1_x}$ is $\bt_1$-recoverable and $G_{x,y^2_x}$ is
$\bt_2$-recoverable. Thus by Lemma~\ref{lem:nandconstraint}(3), a
$\nand(G_{x,y},G_{x',y'})$ constraint excludes the possibility that
$y^1_x=y$ and $y^2_x=y'$ for some $y\neq y'$. In particular, the constraints of the form
$\nand(G_{x,y},G_{x,y'})$, $y\neq y'$ ensure that $y^1_x=y^2_x$ for
every $1\le x \le t$; let $S$ be the set of vertices that contains
$v_{x,y}$ if and only if $y=y^1_x=y^2_x$. Note that $G_{x,y^1_x}$ is
both $\bt_1$- and $\bt_2$-recoverable. Now it is easy to see that $S$
is a multicolored independent set: if $v_{x,y},v_{x',y'}\in S$ are
adjacent vertices, then the constraint $\nand(G_{x,y},G_{x',y'})$
would be violated (Lemma~\ref{lem:nandconstraint}(3)).  \medskip

{\em Case 3: $(R,\bt_1,\bt_2)$ is a union counterexample such that for $i=1,2$, tuple $\bt_i$ is contained in
  the component $C_i$ of value $a_i$, where $a_1$ is regular and $a_2$ is self-producing.}

As in Case 1, the minimality of $\Gamma$ implies that $D=C_1\cup
C_2\cup\{0\}$; note that $C_2=\{a_2\}$ as $a_2$ is self-producing.
Furthermore, there are no self-producing values $c\in C_1$: since
$\{c\}$ is a component and $C_1\setminus \{c\}$ is not a component (as
$C_1$ is the smallest component containing $a_1\neq c$), by
Proposition~\ref{prop:compdiff}, this would imply that there is a
difference counterexample in $C_1$, contradicting the minimality of
$D$.  Thus $a_2$ is the only self-producing value in $D$. We may
assume that $a_1$ is the smallest value of $C_1$.

Set $y^1_x=y^2_x=w_{x,a_1}$. As in the previous case, $G_{x,y^1_x}$ is
$\bt_1$-recoverable. Let us show  that $a_2$ can only appear in a
bag representing $a_2$.  Indeed, suppose that $a_2$ appears in bag
$B_c$ of $G_{x,y}$ for some $c\in C_1$.  It is clear that $a_2$ cannot
appear in a bag representing a degenerate value
(Proposition~\ref{prop:homomtype}) and there are no self-producing
values in $C_1$, thus $c$ is regular.  If $y=w_{x,a_1}$, then
Claim~\ref{claim:regularbags}(2) states that bag $B_c$ contains only
values from $C_1$. Thus $y\neq w_{x,a_1}$, and hence by
Claim~\ref{claim:recoverable}(3), the bag does not contain any values
from $C_1$. As $c$ is regular, the size of the bag $B_c$ is larger
than the cardinality constraint $\pi(a_2)$, implying that 0 also
appears in bag $B_c$. Thus gadget $G_{x,y}$ gives a multivalued
morphism $\phi$ with $0,a_2\in \phi(c)$. If $\phi_{a_2}$ is the
multivalued morphism witnessing that $a_2$ is self-producing, then we
have $0,a_2\in (\phi\circ \phi_{a_2})(c)$ and $0\in
(\phi\circ \phi_{a_2})(c')$ for every value $c'$.
This shows that $c$ produces $a_2$. As $c$ is regular, value $a_2$
does not produce $c$, contradicting the assumption that $a_2$ is
self-producing. Thus value $a_2$ appears only in bags representing
$a_2$, and there are at least $t$ gadgets where $a_2$ appears. Note
that such a gadget is clearly $\bt_2$-recoverable. Furthermore, it is
not possible that $a_2$ appears in bag $B_{a_2}$ of $G_{x,y}$ for
$y\neq y^1_x$: as $G_{y,y^1_x}$ is $\bt_1$-recoverable and $G_{x,y}$
is $\bt_2$-recoverable, the $\nand(G_{x,y^1_x},G_{x,y})$ constraint
would not be satisfied. Thus value $a_2$ has to appear in the gadgets
$G_{x,y^1_x}$, $1\le x\le t$, implying that these gadgets are both
$\bt_1$- and $\bt_2$-recoverable. From this point, we can finish the
proof as in the previous case.
\end{proof}

\subsubsection{Proof of Claim~\ref{claim:recoverable}}\label{sec:claimproof}

 First, we show that Claim~\ref{claim:recoverable} follows from the following claim:
\begin{claim}\label{claim:recoverable2}
  Let $K$ be the component generated by a regular value $d$ and
  suppose that $d$ is the smallest value in $K$.
\begin{enumerate}
\item For every $1\le x \le t$ and $a\in K$, bag $B_a$ of
  $G_{x,w_{x,d}}$ contains only values from $K$, i.e.,
  $\phi_{x,d}(a)\subseteq K$.
\item   For every $1\le x
  \le t$, there is a multivalued morphism $\beta_x$ such that $\hat
  \phi_{x,d}=\phi_{x,d}\circ \beta_x$ is a multivalued morphism with
  $\phi_{x,d}(a)\cup \{a\}\subseteq \hat \phi_{x,d}(a)$ for every
  $a\in K$ and $\hat \phi_{x,d}(a)=\{0\}$ for every $a\not\in
  K$.
\end{enumerate}
\end{claim}
Note that $\phi_{x,d}(a)\subseteq K$ implies, in particular, $0\not\in \phi_{x,d}(a)$.
\begin{proof}[Proof (of Claim~\ref{claim:recoverable} from Section~\ref{sec:regular-values})]
  To see this, we first show that if $\bt$ contains nonzero values
  only from $K$, then $a\in\hat\phi_{x,d}(a)$ implies that a subset of
  $\phi_{x,d}$ is a $\bt$-recoverable endomorphism.  Indeed, it means
  that for every $a\in K$, there is a $c_a\in\phi_{x,d}(a)$ such that
  $a\in\beta_x(c_a)$. Then any endomorphism that maps $a$ to $c_a$ is
  $\bt$-recoverable, as witnessed by $\beta_x$. This proves Statement
  1 of Claim~\ref{claim:recoverable}.
  Statement 2 of Claim~\ref{claim:recoverable} is the same as
  Statement 1 of Claim~\ref{claim:recoverable2}.  Statement 3 of
  Claim~\ref{claim:recoverable} follows from Statement 2 of
  Claim~\ref{claim:recoverable} and from the fact that the cardinality
  requirement for the values in $K$ is exactly the same as the total
  size of the bags $B_a$ of the gadgets $G_{x,w_{x,d}}$ for every
  $a\in K$ and $1\le x \le t$. 
\end{proof}

In the rest of the section, we prove Claim~\ref{claim:recoverable2} by
induction on $x$.
\begin{proof}[Proof (of Claim~\ref{claim:recoverable2})]
We show that $\beta_x$ and $\hat \phi_{x,d}$ exist
if $\hat \phi_{i,d}$ exists for every $1\le i <x$ (in case of $x=1$,
this condition is vacuously true).  Let $\Psi_i$ be the set of those
multivalued morphisms that can be obtained as the product of an
arbitrary sequence of at least one multivalued morphisms composed from $\pr_{K}$, $\hat
\phi_{1,d}$, $\dots$, $\hat \phi_{i,d}$ (possibly using some of them multiple
times).  Observe that $\Psi_i$ has a unique maximal element $\psi_i$,
i.e, $\psi_i(a)\supseteq \psi'_i(a)$ for every $a\in K$ and
$\psi'_i\in \Psi_i$. This follows from the fact that if
$\psi'_i,\psi''_i\in\Psi_i$, then $a\in \psi'_i(a)$ and
$a\in\psi''_i(a)$ for every $a\in K$ implies $\psi'_i(a)\cup
\psi''_i(a) \subseteq (\psi'_i \circ \psi''_i)(a)$.  For notational
convenience, we define $\psi_0:=\pr_{K}$. Note that $\psi_i(a)=\{0\}$ for every $a\not\in K$ by definition.

First, we prove Claim~\ref{claim:recoverable2} assuming that the following claim is true:

\begin{claim}\label{claim:recoverable3} For every $b\in K$, there is an
  integer $p_b\ge 1$ such that for every $p\ge p_b$, either $0$ or $b$
  is in $( \phi_{x,d}\circ \psi_{x-1})^p(b)$. Furthermore, for $b=d$
  we actually have $d\in ( \phi_{x,d}\circ \psi_{x-1})^p(d)$.
  \end{claim}

  Let $p:=\max_{b\in K}p_b$. For every $b\in K$ and $p'\ge p$, either
  $0$ or $b$ is in $(\phi_{x,d}\circ \psi_{x-1} )^{p'}(b)$ and in
  particular $d\in (\phi_{x,d}\circ \psi_{x-1})^{p'}(d)$ holds.  For
  some $p'\ge p$, let $S$ contain those elements $b$ of $K$ for
  which $b\in (\phi_{x,d}\circ \psi_{x-1} )^{p'}(b)$ holds; we have
  $d\in S\subseteq K$.  Observe that $S$ is a component  containing $d$ (as $\pr_S$ is
  a subset of $(\psi_{x-1} \circ \phi_{x,d})^{p'}$),
  thus $S$ has to contain the component generated by $d$, i.e.,
  $S=K$. Therefore, we can assume that $b\in (\phi_{x,d}\circ
  \psi_{x-1})^{p'}(b)$ for every $b\in K$ and $p'\ge p$. Let us define
  $\beta_x:=\psi_{x-1} \circ (\phi_{x,d}\circ \psi_{x-1})^{p}$ and
  $\hat\phi_{x,d}=\phi_{x,d}\circ
  \beta_x=(\phi_{x,d}\circ\psi_{x-1})^{p+1}$; it is clear that $b\in
  \hat \phi_{x,d}(b)$ for every $b\in K$. Furthermore, $b\in
  (\phi_{x,d}\circ \psi_{x-1})^{p}(b)$ implies $\phi_{x,d}(b)
  \subseteq \hat \phi_{x,d}(b)$, proving the second statement of
  Claim~\ref{claim:recoverable2}.

  To prove the first statement, observe that we know from
  Claim~\ref{claim:regularbags}(2) that $\phi_{x,d}(d)\subseteq
  K$. For $b\neq d$, if $\phi_{x,d}(b)$ contains a value not in $K$,
  then $0\in ( \phi_{x,d}\circ \psi_{x-1})^{p}(b)$ would follow and
  then $\pr_{K\setminus \{b\}}$ is a subset of $( \phi_{x,d}\circ
  \psi_{x-1})^{p}$, again contradicting that $K$ is the component generated by $d$.  This proves the first
  statement of Claim~\ref{claim:recoverable2}.  
\end{proof}

Now to finish the proof of Claim~\ref{claim:recoverable2}, it suffices
to prove Claim~\ref{claim:recoverable3}.  We prove
Claim~\ref{claim:recoverable3} by double induction: for a fixed $1\le
x\le t$ and $b\in K$, we assume that Claim~\ref{claim:recoverable2} is
true for every $i <x$, and that Claim~\ref{claim:recoverable3} is true
for $x$ and for every $a<b$. It is clear that such a proof and the
proof of Claim~\ref{claim:recoverable2} above together prove both
Claim~\ref{claim:recoverable2} and \ref{claim:recoverable3}.
\begin{proof}[Proof of Claim~\ref{claim:recoverable3}]
  We prove the statement by induction on $b$. Suppose that the
  statement holds for every $a<b$; let us prove it for $b$.

Let $K_{b}=\{a \mid a\in K,a\le b\}$ and let us define
\[
T=\bigcup_{p\ge 1}( \pr_{K_{b}}\circ\phi_{x,d}  \circ
\psi_{x-1}\circ \pr_{K})^p(b)\setminus \{0\},
\]
which is a subset of $K$. (Note that we have no reason to assume that
$\pr_{K_{b}}$ is an endomorphism.) Intuitively, $T$ is the set of
values that can be obtained from $b$ by a sequence of endomorphisms
using $\hat \phi_{1,d}$, $\dots$, $\hat\phi_{x-1,d}$ and the endomorphisms given by
gadget $G_{x,w_{x,d}}$. The two retractions in the definition of $T$
introduce two additional technical conditions: we never leave $K$
(i.e., we consider every value outside $K$ to be 0) and we apply
$\phi_{x,d}$ only on values at most $b$ (i.e., we imagine the bags
$B_a$ of $G_{x,w_{x,d}}$ with $a>b$ to be fully zero).

We show that if $b\in T$, then Claim~\ref{claim:recoverable3} follows. Indeed, in that case if 
$b\in (\phi_{x,d}\circ \psi_{x-1})(b)$, then we can define $p_b:=1$
since $b\in (\phi_{x,d}\circ \psi_{x-1})^p(b)$ for every $p\ge
1$. Thus we are done in this case.  Suppose that $b\not\in
(\phi_{x,d}\circ\psi_{x-1})(b)$. As $b\in T$, this means that there is
an $a\in T$, $a<b$ such that $a\in (\phi_{x,d}\circ \psi_{x-1})(b)$ and
$b\in (\pr_{K_{b}} \circ \phi_{x,d}\circ \psi_{x-1}\circ
\pr_{K})^p(a)$ for some $p\ge 1$.  By the induction hypothesis on
$a<b$, there is a $p_{a}\ge 1$ such that either $a$ or $0$ is in
$(\phi_{x,d}\circ\psi_{x-1} )^{p_{a}+s}(a)$ for every $s\ge
0$. Putting together, we get that either $0$ or $b$ is in
$(\phi_{x,d}\circ \psi_{x-1})^{1+p_{a}+s+p}(b)$ for every $s\ge
0$. Thus we can define $p_b:=1+p_{a}+p$. To prove the second
statement, observe that if $b=d$, then (as $d$ is the smallest element
of $K$), there is no $a<b$ with $a\in T$. Thus we can only get $b\in
(\phi_{x,d}\circ \psi_{x-1})(b)$ in this case.

In the rest of the proof our goal is to argue that we can assume $b\in
T$.  We show next that Claim~\ref{claim:recoverable3} is true if a
value not in $T$ (including 0, which is not in $T$ by definition)
appears in any of the sets of bags (B1)--(B3) defined below, and then argue
that if all values in these bags are from $T$, then $b\in T$ as well.
\begin{enumerate}
\item[(B1)] bag $B_a$ ($a\in T$) of gadget $G_{i,w_{i,d}}$, $1
\le i < x$, 
\item[(B2)] bag $B_b$ of $G_{x,w_{x,d}}$
\item[(B3)] bag $B_{a}$ ($a\in T$, $a< b$) of gadget $G_{x,w_{x,d}}$
\end{enumerate}

For (B1), we know by induction that Statement 1 of
Claim~\ref{claim:recoverable2} holds for every $i<x$, showing that bag
$B_a$ ($a\in T$) of gadget $G_{i,w_{i,d}}$ contains values only from
$K$. This is the point where the definition of $\psi_{x-1}$ becomes
crucial. Intuitively, we can consider a sequence of multivalued
morphisms showing that $a\in T$, and then append an application of
$\phi_{i,d}$ to show that $\phi_{i,d}(a)$ is in $T$ as well; the
important point here is that $\psi_{x-1}\circ \phi_{i,x}$ can be
replaced by $\psi_{x-1}$. Formally, if $a\in T$, then there is a
$p'\ge 0$ and an $a'\in K$ such that $a'\in
\left((\pr_{K_{b}}\circ\phi_{x,d} \circ \psi_{x-1}\circ
  \pr_{K})^{p'}\circ (\pr_{K_b}\circ \phi_{x,d})\right)(b)$ and $a\in
\psi_{x-1}(a')$.  The definition of $\psi_{x-1}$ implies that
$(\psi_{x-1}\circ \phi_{i,d})(a')\subseteq \psi_{x-1}(a')$, and
therefore $\phi_{i,d}(a)\subseteq \psi_{x-1}(a')$ holds as well.  This
means that $\phi_{i,d}(a)\subseteq (\pr_{K_{b}}\circ\phi_{i,d} \circ
\psi_{x-1}\circ \pr_{K})^{p'+1}(b)$, that is, $\phi_{i,d}(a)\subseteq T$.
Thus, bags from (B1) cannot have values of $K$ not in $T$.

For (B2), observe that if we show that only values from $K$ appear in
the bag $B_b$ of $G_{x,w_{x,d}}$, then it follows that all these
values are in $T$.  For $b=d$, Claim~\ref{claim:regularbags}(2)
implies that all these values are in $K$. If $b\neq d$, then a value
not in $K$ appearing in bag $B_b$ of $G_{x,w_{x,d}}$ would imply that
$0\in (\phi_{x,d}\circ \psi_{x-1})(b)$ (recall that
$\psi_{x-1}(a)=\{0\}$ for any $a\not\in K$), and the statement of
Claim~\ref{claim:recoverable3} is true for $b$ with $p_b=1$.

For (B3), it is again sufficient to show that the bags from this set 
contain only values from
$K$; by the definition of $T$, then these values are in $T$. If $a=d$,
then Claim~\ref{claim:regularbags}(2) implies that the bag contains
values only from $K$. Otherwise, suppose that $a\neq d$, $a\in T$,
$a<b$, and there is a value $c\not\in K$ in the bag $B_a$ of
$G_{x,w_{x,d}}$.  As $a\in T$, we have that $a\in ( \pr_{K_{b}}\circ\phi_{x,d}
\circ \psi_{x-1}\circ \pr_{K})^p(b)$ for some $p\ge 1$, which clearly
implies $a\in (\phi_{x,d}\circ \psi_{x-1})^p(b)$. Since $0\in
(\phi_{x,d} \circ \psi_{x-1})(a)$ (recall that $\psi_{x-1}(c)=\{0\}$
if $c\not\in K$), it follows that $0\in
(\phi_{x,d}\circ \psi_{x-1})^{p+1}(b)$, proving Claim~\ref{claim:recoverable3} for $b$ with
$p_b=p+1$. Note that in these cases we assumed $a<b$, i.e., 
$b\neq d$ (as $d$ is the smallest value in $K$) and therefore we do not
have to prove the second statement of Claim~\ref{claim:recoverable3}.

Therefore, in the following, we can assume that bags (B1)--(B3) contain values only from $T$.
The total size of  these bags is
\begin{equation}
\sum_{a\in
  T}\sum_{i=1}^{x-1}z_{i,a}+\sum_{a\in T,a<b}z_{x,a}+z_{x,b}, \label{eq:totalbags}
\end{equation}
while the sum of cardinality constraints of the values in $T$  is (assuming $b\not\in T$) 
\begin{equation}
\sum_{a\in
  T}\sum_{i=1}^{t}z_{i,a}=
\sum_{a\in
  T}\sum_{i=1}^{x-1}z_{i,a}+
\sum_{a\in T,a<b}z_{x,a}+
\sum_{a\in T,a>b}z_{x,a}+
\sum_{a\in  T}\sum_{i=x+1}^{t}z_{x,a}. \label{eq:totalcard}
\end{equation}
Let us compare the last term of \eqref{eq:totalbags} with the last two
terms of \eqref{eq:totalcard}. Suppose first that $b$ is regular.
Then $z_{x,b}\ge (2t\Delta) z_{x,a}$ for any $a>b$ and $z_{x,b}\ge
(2t\Delta) z_{i,a}$ for any $a\in T$ and $i\ge x+1$.  Thus
\eqref{eq:totalbags} is strictly larger than \eqref{eq:totalcard} and
this contradiction proves that $b\in T$.  Suppose now that $b$ is not
regular. Then $T$ cannot contain any regular value
(Proposition~\ref{prop:homomtype}).  Again, we have $z_{x,b}\ge
(2t\Delta)z_{x,a}$ for any $a>b$ and $z_{x,b}\ge (2t\Delta) z_{i,a}$ for any
$a\in T$ and $i\ge x+1$, leading to a contradiction.  Therefore, we
can conclude that $b\in T$, concluding the proof of
Claim~\ref{claim:recoverable3}.
\end{proof}

\section{Examples}\label{sec:examples}

\begin{example}\label{exa:graphs}\rm
A number of graph problems can be represented in the form of
$\ccsp(\Gamma)$ or $\ocsp(\Gamma)$.\\[1mm] 
\textsc{Independent Set:} Given a graph $G$ with
  vertices $v_{i}$ ( $1\le j \le n$), find an
  independent set of size $t$. \textsc{Independent Set} is equivalent
  to $\ccsp(\{R_{IS}\})$, or, equivalently to $\ocsp(\{R_{IS}\})$,
  where $R_{IS}$ is a binary relation on $\{0,1\}$ given by 
$$
R_{IS}=\left(\begin{array}{ccc} 0&1&0\\ 0&0&1 \end{array} \right)
$$
(tuples are written vertically). The size constraint is set to be $t$.\\[1mm]
\textsc{$p$-Colorable Subgraph}: Given a graph $G$ 
and an integer $k$, find a set $S$ of $k$ vertices that induces a
$p$-colorable subgraph.
This problem is equivalent to
  $\ocsp(\{R_{p-COL}\})$, where $R_{p-COL}$ is a binary relation on
  $p+1$-element set $D=\{0,1,\ldots p\}$ given by 
$$
R_{p-COL}=D^2\setminus \left\{\left(\begin{array}{c} i\\i\end{array}
  \right) \mid i\in\{1,\ldots,p\}\right\}.
$$
The size constraint is $k$, the size of the $p$-colorable graph to be found.
\\[1mm] 
\textsc{Implications:} Given a directed graph $G$ and an integer
  $t$, find a set $C$ of vertices with exactly $t$ vertices such that
  there is no directed edge $\overrightarrow{uv}$ with $u\in C$ and
  $v\not\in C$. \textsc{Implications} can be represented as
  $\ccsp(\{R_{IM}\})$ and $\ocsp(\{R_{IM}\})$, where 
$$
R_{IM}=\left(\begin{array}{ccc} 0&0&1\\ 0&1&1 \end{array} \right)
$$
The size constraint is set to be $t$.\\[1mm]
\textsc{Vertex Cover:} Given a graph $G$ and an integer $t$,
  find a set $C$ of vertices such that every edge of $G$ is incident
  to at least one vertex from $C$. \textsc{Vertex Cover} is equivalent
  to the $\ocsp(R_{VC})$, where
$$
R_{VC}=\left(\begin{array}{ccc} 0&1&1\\ 1&0&1 \end{array} \right),
$$
and the size constraint is $t$.
\end{example}

Some problems reduce to the CSP with cardinality constraints in a less
straightforward way. 

\begin{example}\label{exa:biclique}\rm
\textsc{Biclique:} Given a bipartite graph $G(A,B)$, find two
  sets $A'\subseteq A$ and $B'\subseteq B$, each  of size exactly $t$, such
  that every vertex of $A'$ is adjacent with every vertex of $B'$. As
  it is mentioned in the introduction, \textsc{Biclique} is equivalent
  to $\ccsp(\{R_{BC}\})$, where $R_{BC}$ is a relation on $\{0,1,2\}$
  given by 
$$
R_{BC}=\left(\begin{array}{ccc} 0&1&0\\ 0&0&2 \end{array} \right).
$$
A \textsc{Biclique} instance $G(A,B)$ is reduced to $\ccsp(\{R_{BC}\})$
by first taking the (bipartite) complement $\overline{G}$ of $G$ and imposing
constraint $\ang{(v,w),R_{BC}}$ on every pair $v\in A$, $w\in B$ such
that $v,w$ are adjacent in $\overline{G}$. The cardinality constraint
$\pi:\{0,1,2\}\to \nat$ is chosen to be $\pi(1)=\pi(2)=t$.

The variant of \textsc{Biclique}, where the graph $G$ is not
necessarily bipartite, is equivalent to $\ccsp(\{R'_{BC}\})$, where $R'_{BC}$ is a relation on $\{0,1,2\}$
  given by 
$$
R'_{BC}=\left(\begin{array}{ccccc} 0&1&0&2&0\\ 0&0&2&0&1 \end{array} \right).
$$ 
\end{example}

The following examples generalize \textsc{Biclique} to finding
complete $p$-partite graphs. Let us first consider  the version
where the input graph is also $p$-partite:

\begin{example}\label{exa:multipartite-clique}\rm
\textsc{$p$-Partite Clique:} Given a $p$-partite graph $G$ with
partition $A_1,\ldots A_p$, find sets $A'_1\subseteq A_1,\ldots
A'_p\subseteq A_p$, each of size exactly $t$, such that for any $i,j$,
$1\le i<j\le k$ every vertex from $A'_i$ is adjacent with every vertex
of $A'_j$. The equivalent $\ccsp$ problem is $\ccsp(\{R_{p-MC}\})$ where
$R_{p-MC}$ is the $p$-ary relation given by 
$$
R_{p-MC}=\{0,1\}\times\{0,2\}\times\ldots\times\{0,p\}-\{(1,2,\ldots,p)\}.
$$

Reduction goes as follows. Given a $p$-partite graph $G$ with
partition $A_1,\ldots A_p$ we first take the ($p$-partite) complement $\overline{G}$
of $G$, and then introduce constraint $\ang{(v_1,\ldots, v_p),R_{p-MC}}$
for each $p$-tuple $(v_1,\ldots,v_p)$ such that $v_i\in A_i$ and some
of the vertices $v_1,\ldots, v_p$ are adjacent in $\overline{G}$. The cardinality 
constraint is chosen to be $\pi(1)=\ldots=\pi(p)=t$. 

Another way to represent \textsc{$p$-Partite Clique} by a \ccsp\ is
the following. Let
$$
R_{p-PC}=\left(\begin{array}{cccccc} 0&1&0&\cdots&p&0\\ 
0&0&1&\cdots&0&p \end{array} \right),
$$
and $R_i$ is the unary relation $\{0,i\}$ for
$i\in\{1,\ldots,p\}$. Then \textsc{$p$-Partite Clique} reduces to
$\ccsp(\{R_{p-PC},R_1,\ldots, R_p\})$ by imposing the constraint
$\ang{(v),R_i}$ on each $v\in A_i$, and $\ang{(v,w),R_{p-PC}}$ on each
pair $v,w$ adjacent in $\overline{G}$.
\end{example}

The following example formulates the version of \textsc{$p$-Partite
  Clique} where the input graph is not $p$-partite:

\begin{example}\label{exa:k-partite-complete-subgraph}\rm
\textsc{$p$-Partite Complete Subgraph}: Given a graph $G$ 
and integers
  $t_1,\ldots, t_p$, find sets $S_1,\ldots S_p$ of vertices such that
  $S_1\cup\ldots\cup S_p$ induces a complete $p$-partite graph with
  partition $S_1,\ldots,S_p$, and $|S_i|=t_i$ for $1\le i\le p$. 
This problem is equivalent to
  $\ccsp(\{R_{p-CS}\})$, where $R_{p-CS}$ is a binary relation on
  $p+1$-element set $D=\{0,1,\ldots p\}$ given by 
$$
R_{p-CS}=\{0,1\}^2\cup\{0,2\}^2\cup\ldots\cup\{0,p\}^2.
$$
To reduce the \textsc{$p$-Partite Complete Subgraph} problem to
$\ccsp(\{R_{p-CS}\})$ we, first, take the complement $\overline{G}$ of $G$, and
then impose constraint $\ang{(v,w),R_{p-CS}}$ on every pair $v,w$ such
that $v,w$ are adjacent in $\overline{G}$. The cardinality constraint is set to
be $\pi$ such that $\pi(i)=t_i$ 
for all $i$.\\[1mm] 
\end{example}

\begin{example}\label{exa:cc-hard}\rm
  Let $R=(\{0,1\}\times \{0,1\}\times \{0,2\})\setminus \{(1,1,2)\}$.
  Let $\Gamma$ be the cc-closure of $R$. Note that $\Gamma$ contains
  $R^{|3;2}=\{(0,0),(1,0),(0,1)\}$, which is the relation $R_{IS}$ of
  Example~\ref{exa:graphs}, showing that \ocspg\ is
  W[1]-hard. On the other hand, we show that $\ocsp(\{R\})$ is fixed-parameter tractable. 

  Let $S_1$ be the set of variables $v$ where value 1 can appear, that
  is, there is no constraint $\ang{(v',v'',v),R}$ on $v$. Observe that
  any combination of 0 and 1 on $S_1$ is a satisfying
  assignment. Therefore, if $|S_1|\ge k$, then there is a solution.
  Let $S_2$ be the set of variables $v$ where value 2 can appear, that
  is, there is no constraint $\ang{(v,v',v''),R}$ or
  $\ang{(v',v,v''),R}$ on $v$. Observe that any combination of 0 and 2
  on $S_2$ is a satisfying assignment. Therefore, if $|S_2|\ge k$,
  then there is a solution. On the variables not in $S_1$ and $S_2$,
  only value 0 can appear. Therefore, if $|S_1|,|S_2|<k$, then the
  number of possible assignments that we need to try is
  $2^{|S_1|+|S_2|}<2^{2k}$.
\end{example}

\begin{example}\label{exa:cc-hard2}\rm
  Let $R=(\{0,1,2\}\times \{0,1,2\}\times \{0,2\})\setminus
  \{(1,1,2)\}$.  Let $\Gamma$ be the cc-closure of $R$. Note that
  $\Gamma$ contains $R^{|3;2}=(\{0,1,2\}\times \{0,1,2\})\setminus
  \{(1,1)\}$. Therefore, $\Gamma_{|\{0,1\}}$ contains the relation
    $\{(0,0),(1,0),(0,1)\}$, which is the relation $R_{IS}$ of
    Example~\ref{exa:graphs}, showing that $\ccsp(\Gamma_{|\{0,1\}})$
    is W[1]-hard. We can reduce $\ccsp(\Gamma_{|\{0,1\}})$ to \ccspg\
    by setting the cardinality constraint of value 2 to 0, thus the
    W[1]-hardness of \ccspg\ follows. On the other hand, we show that
    $\ccsp(\{R\})$ is fixed-parameter tractable.

    Let $S$ be the set of variables $v$ where value 1 can appear, that
    is, there is no constraint $\ang{(v',v'',v),R}$ on $v$. Observe
    that any combination of 0, 1, and 2 on $S$ is a satisfying
    assignment. Therefore, if $|S|\ge k$, then there is a solution.
    If $|S|<k$, then we can try every possible substitution of
    constants into $S$ and obtain instances where 1 can no longer
    appear. As the problem on $\Gamma_{|\{0,2\}}$ is trivial,
    fixed-parameter tractability follows.
\end{example}

\begin{example}\label{exa:cc0-languages}\rm
The language consisting of the relations
\[
R_1=\left(\begin{array}{llll}
0&1&0&1\\
0&1&0&1\\
0&0&2&2\\
  \end{array}\right) \
R_2=\left(\begin{array}{ll}
0&0\\
0&2\\
  \end{array}\right) \
R_3=\left(\begin{array}{ll}
0&2
  \end{array}\right) \
R_4=\left(\begin{array}{ll}
0&1\\
0&1\\
  \end{array}\right) \
R_5=\left(\begin{array}{ll}
0&1\\
  \end{array}\right) \
R_6=\left(\begin{array}{l}
0\\
  \end{array}\right) 
  \]
is a cc0-language, but not a cc-language: substituting 1 into the
first coordinate of $R_1$ results in the (non 0-valid) relation $\{(1,0),(1,2)\}$, which is not in the language.
\end{example}

\begin{example}\label{exa:weaklysep1}\rm
\begin{itemize}
\item Relation $R_{IS}$ is not weakly separable:
  $(R_{IS},(1,0),(0,1))$ is a union counterexample (i.e., $(1,1)\not\in R_{IS}$).
\item Relation $R_{IM}$ is not weakly separable:
  $(R_{IM},(1,0),(0,1))$ is a difference counterexample.
\item The $r$-ary relation \[R_\text{even}=\{(a_1,\dots,a_r)\in
  \{0,1\}^r \mid \text{$a_1+\dots +a_r$ is even}\}\] is weakly
  separable. Indeed,  if there is an even number
  of 1's in each of $\bt_1$, $\bt_2$ and they are disjoint, then their
  union also contains an even number of 1's. Similarly, if $\bt_1+\bt_2$
  and $\bt_2$ contain even number of 1's, then so does $\bt_1$.

\item We can generalize $R_\text{even}$ the following way: let us define the
  $r$-ary $R_\text{mod-$p$}$ relation over the domain $\{0,\dots, d\}$
  the following way: \[
R_\text{mod-$p$}=\{(a_1,\dots,a_r)\in
  \{0,\dots,d\}^r \mid \text{$a_1+\dots +a_r=0 \mod p$} \}.\]
It is easy to see that this relation is weakly separable as well.
\end{itemize}

\end{example}

\begin{example}\label{exa:weaklysep}\rm
We demonstrate how hardness is proved in the Boolean case \cite{Marx05:parametrized} if
there is a relation that is not weakly separable.
Let $\Gamma$ be a cc0-language over $\{0,1\}$ that is not weakly
  separable. Suppose that there is a union counterexample
  $(R,\bt_1,\bt_2)$ in $\Gamma$; suppose for example that
\[
  \bt_1=(\overbrace{1,\dots,1}^p,0,\dots,0),\
  \bt_2=(0,\dots,0,\overbrace{1,\dots,1}^q,0,\dots,0).
\] Since $\Gamma$ is a
  cc0-language, by substituting 0's in the last coordinates, we can
  obtain a relation $R'\in \Gamma$ and a union counterexample
  $(R',\bt'_1,\bt'_2)$ with 
\[
\bt'_1=(\overbrace{1,\dots,1}^p,0,\dots,0),\ 
 \bt'_2=(0,\dots,0,\overbrace{1,\dots,1}^q).\] Now it is easy to see that
 the W[1]-hard problem $\ocsp(R_{IS})$ is reducible to $\ocsp(\Gamma)$: a constraint
  $\ang{(x,y),R_{IS}}$ can be simulated by a constraint
  $\ang{(x,\dots,x,y,\dots,y),R'}$. The correctness of this
  reduction follows from 
\begin{align*}
(0,\dots,0,0,\dots,0)&\in R'\\
(\overbrace{1,\dots,1}^p,0,\dots,0)&\in R'\\
(0,\dots,0,\overbrace{1,\dots,1}^q)&\in R'\\
(\overbrace{1,\dots,1,1,\dots,1}^{p+q})&\not\in R'.
\end{align*}

Suppose now that there is a difference counterexample
$(R,\bt_1,\bt_2)$. As above, let us suppose that there is also a
difference counterexample $(R',\bt'_1,\bt'_2)$ with
$\bt'_1=(1,\dots,1,0,\dots,0)$ and $\bt'_2=(0,\dots,0,1,\dots,1)$. Now
we can reduce the W[1]-hard problem $\ocsp(R_{IM})$ to
$\ocsp(\Gamma)$. Since we have $(0,\dots,0),(1,\dots,1),
(0,\dots,0,1,\dots,1)\in R'$ and $(1,\dots,1,0,\dots,0)\not\in R'$,
now a constraint $\ang{(x,\dots,x,y,\dots,y),R'}$ expresses the
relation $R_{IM}$.
\end{example}

\begin{example}\label{exa:nonweaklysepeasy}\rm
Consider the cc0-closure $\Gamma$ of the following relation:
\[
R=\left(\begin{array}{lllll}
0 & 1 & 2 & 0 & 2\\
0 & 0 & 0 & 2 & 2\\
\end{array}\right),
  \]
  Clearly, $\Gamma$ is not weakly separable: $(R,(1,0),(0,2))$ is a
  union counterexample. Nevertheless, both \ocspg\ and \ccspg\ are
  fixed-parameter tractable. In the case of \ocspg, every 1 in a
  solution can be replaced by 2. Hence it is sufficient to solve the
  problem restricted to $\{0,2\}$, in which case the problem is
  trivial, as every combination of 0 and 2 is a satisfying assignment.
  For \ccspg, we argue as follows.  Let $S$ be the set of variables
  $v$ where value 1 can appear, that is, there is no $\ang{(v',v),R}$
  or any unary constraint excluding 1 on $v$. Observe that any
  combination of 0, 1, and 2 on $S$ is a satisfying assignment. Therefore,
  if $|S|\ge k$, then there is a solution.  If $|S|<k$, then we can
  try every possible substitution of constants into $S$ and obtain
  instances where 1 can no longer appear. As
  $\ccsp(\Gamma_{|\{0,2\}})$ is trivial, solving these instances is
  fixed-parameter tractable.
\end{example}

\begin{example}\label{exa:nonweaklysep}\rm
The relation $R_{IM}$ of Example~\ref{exa:graphs} is not weakly
separable. Consider a CSP instance on three variables $v_1$, $v_2$,
$v_3$ having two constraints $\ang{(v_1,v_3),R_{IM}}$ and
$\ang{(v_2,v_3),R_{IM}}$. The assignment $(0,0,1)$ is the only minimal
satisfying assignment. Assignment $(1,1,1)$ is satisfying, but it
cannot be obtained as the union of pairwise disjoint satisfying
assignments (even if we do not require {\em minimal} satisfying assignments).
\end{example}

\begin{example}\label{exa:endo-homo}\rm
\begin{itemize}
\item
The cc0-closure $\Gamma_{p-PC}$ of $\{R_{p-PC}\}$ (defined in Example~\ref{exa:multipartite-clique}) consists of
$R_{p-PC}$ itself and unary relations $\{(0)\}$ and
$D=\{0,1,\ldots,p\}$. It is easy to see that every mapping $h:D\to D$
with $h(0)=0$ is an endomorphism of $\Gamma_{p-PC}$, for any sets
$0\in D_1,D_2\sse D$, any mapping $f:D_1\to D_2$ with
$f(0)=0$ is an inner homomorphism of $\Gamma_{p-PC}$.
 Finally, any mapping $\phi:D_1\to 2^{D_2}$ such that $0\in
 D_1,D_2\sse D$ and $\phi(0)=\{0\}$ is a
 multivalued morphism or an inner multivalued morphism of $\Gamma_{p-PC}$. 

\item Multivalued morphisms of relation $R_{p-MC}$ and relations from
  its cc0-closure are the mappings $\phi:D\to 2^D$ satisfying
  $\phi(d)\subseteq\{0,d\}$ for every $d\in D$. A mapping $g:D\to D$
  is an endomorphism of $R_{p-MC}$ if and only if $g(d)\in\{0,d\}$ for
  each $d\in D-\{0\}$ and $g(0)=0$. Inner homomorphisms and inner
  multivalued morphisms can be described in a similar way.

\item The constraint $R_{\le,d}$ generalizes $R_{IM}$ to the domain
  $\{0,1,\dots,d\}$: $R_{\le,d}=\{(x,y)\in \{0,1,\dots,d\}^2 \mid x\le
  y\}$. A function $h$ with $h(0)=0$ is an endomorphism of $R_{\le,d}$
  if and only if it is monotone, i.e., $h(x)\le h(y)$ for every $x\le
  y$. $R_{\le,d}$ does not have a multivalued morphism that is not an
  endomorphism: if $\psi$ is a multivalued morphism of $R_{\le,d}$
  with $a,b\in \psi(x)$, then $(x,x)\in R_{\le,d}$ implies that both
  $(a,b)$ and $(b,a)$ are in $R_{\le,d}$.

\item The constraint $R_{<,d}$ is defined similarly to $R_{\le,d}$,
  but we also add the tuple $(0,0)$ to make it 0-valid:
  $R_{\le,d}=(0,0)\cup \{(x,y)\in \{0,1,\dots,d\}^2 \mid x< y\}$. In
  this case, there are nontrivial multivalued morphisms: for example,
  $\psi(0)=\psi(1)=\psi(2)=\{0\}$, $\psi(3)=\{1,2\}$,
  $\psi(4)=\{3\}$ is a multivalued morphism of $R_{<,4}$.
\end{itemize}
\end{example}

\begin{example}\label{exa:produces}\rm
In the cc0-closure of the relation
\[
R=\left(\begin{array}{lllll}
0 & 1 & 0 & 1 & 2\\
0 & 0 & 1 & 1 & 2\\
\end{array}\right),
  \]
both 1 and 2 produce 1, but neither 1 nor 2 produces 2.
\end{example}

\begin{example}\label{exa:types}\rm
Consider the cc0-closure of the following two relations:
\[
R_1=\left(\begin{array}{lllllllllll}
0 & 1 & 1 & 2 & 0 & 2 & 3 & 0 & 3 & 4& 4 \\
0 & 0 & 1 & 0 & 2 & 2 & 0 & 3 & 3 & 5& 4\\
\end{array}\right)\
R_2=\left(\begin{array}{lllllll}
0& 1& 0 & 1&3 & 4& 5\\
0& 0& 1 & 1&3 & 4& 5\\
\end{array}\right).
  \]
Value 2 produces 2 and 3, but there are no other producing
relations. Thus 3 is degenerate and 2 is self-producing. The mapping
$\psi(0)=\psi(1)=\psi(2)=\psi(3)=\{0\}$, $\psi(4)=\{1\}$, $\psi(5)=\{0,1\}$ is
a multivalued morphism, thus 1 is semiregular. Values 4 and 5 are regular.
\end{example}

\begin{example}\label{exa:nonweaklysep-easy}\rm

Consider the cc0-closure $\Gamma$ of following two relations:
\[
R_1=\left(\begin{array}{lll}
0 & 3 & 2\\
0 & 1 & 2\\
0 & 2 & 1\\
\end{array}\right)\
R_2=\left(\begin{array}{lllllllllllll}
0 & 1 & 0 & 1 & 2 & 0 & 2 & 0 & 3 & 2 & 1 & 2 & 1\\
0 & 0 & 1 & 1 & 0 & 2 & 2 & 3 & 3 & 1 & 2 & 3 & 3\\
\end{array}\right).
  \]
$\Gamma$ is not weakly separable: $(R_2,(3,0),(0,3))$ is a difference
counterexample. Let us observe that $h(0)=0$, $h(1)=2$, $h(2)=1$,
$h(3)=2$ is an endomorphism (a proper contraction) of $\Gamma$. Therefore, by applying $h$ on a
solution for an \ocspg\ instance, we can obtain a solution using only
the values 0, 1, and 2. $\Gamma$ restricted to $\{0,1,2\}$ is weakly
separable, thus finding such solutions is FPT.
\end{example}

\begin{example}\label{exa:ccsphard}\rm
We have seen in Example~\ref{exa:nonweaklysepeasy} that \ccspg\ is
 fixed-parameter tractable for the cc0-closure $\Gamma$ of the following relation:
\[
R=\left(\begin{array}{lllll}
0& 1& 2 & 0&2 \\
0& 0& 0 & 2&2 \\
\end{array}\right).
\]
Let us verify that Theorem~\ref{th:multimain} indeed classifies \ccspg\ as
fixed-parameter tractable.  $\Gamma$ is not weakly separable:
$(R,(1,0),(0,2))$ is a union counterexample. However, $\Gamma$ is not
a core: value 1 is self-producing, 2 is degenerate, and the component
generated by 1 is $\{1\}$. $\Gamma_{|\{0,1\}}$ and $\Gamma_{|\{0,2\}}$
are cores, but they are weakly separable. Thus by
Theorem~\ref{th:multimain}, \ccspg\ is fixed-parameter tractable.

Consider the cc0-closure $\Gamma$ of following relation:
\[
R=\left(\begin{array}{lllllll}
0& 0& 0 & 3&3 & 0&3\\
0& 1& 2 & 0&2 & 4&4\\
0& 2& 0 & 0&0 & 0&0\\
\end{array}\right).
\]
Value 3 produces 3, 1 produces 4, and both 1 and 4 produces 2, but
there are no other producing relations. Thus 2 and 4 are degenerate
and 3 is self-producing. We can also see that 1 is regular. The core
of $\Gamma$ is the component generated by $\{1,3\}$ which is
$K=\{1,2,3\}$. Thus $\Gamma$ is not a core. However,
$\Gamma_{|\{0,1,2,3\}}$ is a core and it is not weakly separable:
$(R,(3,0,0),(0,1,2))$ is a union counterexample. Thus by
Theorem~\ref{th:multimain}, \ccspg\ is W[1]-hard.
\end{example}

 \bibliographystyle{abbrv}  
\bibliography{global}

\begin{thebibliography}{10}

\bibitem{Barto08:graphs}
L.~Barto, M.~Kozik, and T.~Niven.
\newblock Graphs, polymorphisms and the complexity of homomorphism problems.
\newblock In {\em STOC}, pages 789--796, 2008.

\bibitem{Bessiere04:global}
C.~Bessi{\`e}re, E.~Hebrard, B.~Hnich, and T.~Walsh.
\newblock The complexity of global constraints.
\newblock In {\em AAAI}, pages 112--117, 2004.

\bibitem{Bourdais03:hibiscus}
S.~Bourdais, P.~Galinier, and G.~Pesant.
\newblock {HIBISCUS}: A constraint programming application to staff scheduling
  in health care.
\newblock In {\em CP}, pages 153--167, 2003.

\bibitem{bulatov-marx-ccsp-icalp2011}
A.~Bulatov and D.~Marx.
\newblock Constraint satisfaction parameterized by solution size.
\newblock In {\em 38th International Colloquium on Automata, Languages and
  Programming}, volume 6755 of {\em Lecture Notes in Computer Science}, pages
  424--436. Springer, 2011.

\bibitem{Bulatov06:3-element}
A.~A. Bulatov.
\newblock A dichotomy theorem for constraint satisfaction problems on a
  3-element set.
\newblock {\em J. ACM}, 53(1):66--120, 2006.

\bibitem{DBLP:journals/tocl/Bulatov11}
A.~A. Bulatov.
\newblock Complexity of conservative constraint satisfaction problems.
\newblock {\em ACM Trans. Comput. Log.}, 12(4):24, 2011.

\bibitem{Bulatov05:classifying}
A.~A. Bulatov, P.~Jeavons, and A.~A. Krokhin.
\newblock Classifying the complexity of constraints using finite algebras.
\newblock {\em SIAM J. Comput.}, 34(3):720--742, 2005.

\bibitem{bulatov-marx-ccsp-lmcs}
A.~A. Bulatov and D.~Marx.
\newblock The complexity of global cardinality constraints.
\newblock {\em Logical Methods in Computer Science}, 6:1--27, 2010.

\bibitem{bulatov-marx-ccsp-cacm}
A.~A. Bulatov and D.~Marx.
\newblock Constraint satisfaction problems and global cardinality constraints.
\newblock {\em Commun. ACM}, 53(9):99--106, 2010.

\bibitem{DBLP:journals/tocl/CreignouSS10}
N.~Creignou, H.~Schnoor, and I.~Schnoor.
\newblock Nonuniform {B}oolean constraint satisfaction problems with
  cardinality constraint.
\newblock {\em ACM Trans. Comput. Log.}, 11(4), 2010.

\bibitem{MR2001b:68042}
R.~G. Downey and M.~R. Fellows.
\newblock {\em Parameterized Complexity}.
\newblock Monographs in Computer Science. Springer, New York, 1999.

\bibitem{Feder98:monotone}
T.~Feder and M.~Vardi.
\newblock The computational structure of monotone monadic {SNP} and constraint
  satisfaction: A study through datalog and group theory.
\newblock {\em {SIAM} J. Computing}, 28:57--104, 1998.

\bibitem{DBLP:journals/tcs/FellowsHRV09}
M.~R. Fellows, D.~Hermelin, F.~A. Rosamond, and S.~Vialette.
\newblock On the parameterized complexity of multiple-interval graph problems.
\newblock {\em Theor. Comput. Sci.}, 410(1):53--61, 2009.

\bibitem{grohe-flum-param}
J.~Flum and M.~Grohe.
\newblock {\em Parameterized Complexity Theory}.
\newblock Springer, Berlin, 2006.

\bibitem{DBLP:conf/ijcai/GaspersS11}
S.~Gaspers and S.~Szeider.
\newblock Kernels for global constraints.
\newblock In {\em IJCAI}, pages 540--545, 2011.

\bibitem{Grohe07:other-side}
M.~Grohe.
\newblock The complexity of homomorphism and constraint satisfaction problems
  seen from the other side.
\newblock {\em J. ACM}, 54(1), 2007.

\bibitem{Jeavons97:closure}
P.~Jeavons, D.~Cohen, and M.~Gyssens.
\newblock Closure properties of constraints.
\newblock {\em J. ACM}, 44:527--548, 1997.

\bibitem{Jeavons99:expressive}
P.~Jeavons, D.~Cohen, and M.~Gyssens.
\newblock How to determine the expressive power of constraints.
\newblock {\em Constraints}, 4:113--131, 1999.

\bibitem{MR2002k:68058}
S.~Khanna, M.~Sudan, L.~Trevisan, and D.~P. Williamson.
\newblock The approximability of constraint satisfaction problems.
\newblock {\em SIAM J. Comput.}, 30(6):1863--1920, 2001.

\bibitem{kratsch-mfcs2010-maxexact}
S.~Kratsch, D.~Marx, and M.~Wahlstr{\"o}m.
\newblock Parameterized complexity and kernelizability of {M}ax {O}nes and
  {E}xact {O}nes problems.
\newblock In {\em MFCS}, pages 489--500, 2010.

\bibitem{DBLP:conf/icalp/KratschW10}
S.~Kratsch and M.~Wahlstr{\"o}m.
\newblock Preprocessing of {M}in {O}nes problems: A dichotomy.
\newblock In {\em ICALP (1)}, pages 653--665, 2010.

\bibitem{krokhin-marx-icalp2008}
A.~Krokhin and D.~Marx.
\newblock On the hardness of losing weight.
\newblock {\em ACM Trans. Algorithms}, 8(2):19, 2012.

\bibitem{Kun-Szegedy-STOC2009}
G.~Kun and M.~Szegedy.
\newblock A new line of attack on the dichotomy conjecture.
\newblock In {\em STOC}, pages 725--734, 2009.

\bibitem{Marx05:parametrized}
D.~Marx.
\newblock Parameterized complexity of constraint satisfaction problems.
\newblock {\em Computational Complexity}, 14(2):153--183, 2005.
\newblock Special issue ``Conference on Computational Complexity (CCC) 2004.''.

\bibitem{Gomes04:cardinality}
J.-C. R{\'e}gin and C.~P. Gomes.
\newblock The cardinality matrix constraint.
\newblock In {\em CP}, pages 572--587, 2004.

\bibitem{Rosenberg98:hyperstructures}
I.~Rosenberg.
\newblock Multiple-valued hyperstructures.
\newblock In {\em Proceedings of the 28th International Symposium on
  Multiple-Valued Logic (ISMVL '98)}, pages 326--333, 1998.

\bibitem{DBLP:journals/constraints/SamerS11}
M.~Samer and S.~Szeider.
\newblock Tractable cases of the extended global cardinality constraint.
\newblock {\em Constraints}, 16(1):1--24, 2011.

\bibitem{MR80d:68058}
T.~J. Schaefer.
\newblock The complexity of satisfiability problems.
\newblock In {\em STOC}, pages 216--226. ACM, 1978.

\bibitem{DBLP:journals/disopt/Szeider11}
S.~Szeider.
\newblock The parameterized complexity of $k$-flip local search for {S}{A}{T}
  and {M}{A}{X} {S}{A}{T}.
\newblock {\em Discrete Optimization}, 8(1):139--145, 2011.

\end{thebibliography}

\end{document}